\documentclass{lmcs}
\pdfoutput=1
\usepackage[utf8]{inputenc}

\usepackage{lastpage}
\lmcsdoi{22}{2}{27}
\lmcsheading{}{\pageref{LastPage}}{}{}%
{Dec.~03,~2024}{Jun.~12,~2026}{}


\usepackage{hyperref}

\usepackage{amssymb}
\usepackage{xspace}
\usepackage{multirow}
\usepackage{booktabs}
\usepackage{csquotes}
\usepackage{makecell}
\usepackage{thm-restate}
\newcommand{\mathnotation}[1]{\ensuremath{#1}\xspace}

\newcommand{\tuple}[1]{\bar{#1}}
\newcommand{\range}[2][1]{\ensuremath{\{#1,\ldots,#2\}}}

\newcommand{\bigO}{\ensuremath{\operatorname{\mathcal{O}}}}

\newcommand{\IN}{\ensuremath{\textnormal{\textsf{IN}}}\xspace}
\newcommand{\OUT}{\ensuremath{\textnormal{\textsf{OUT}}}\xspace}

\newcommand{\ctpalgo}{$\bigO(1)$-time parallel algorithm\xspace}
\newcommand{\ctpalgos}{$\bigO(1)$-time parallel algorithms\xspace}
\newcommand{\propertuple}{proper tuple\xspace}
\newcommand{\propertuples}{proper tuples\xspace}

\newcommand{\attrof}[1]{\mathnotation{\textnormal{\texttt{attr}}(#1)}}
\newcommand{\attA}{A}

\newcommand{\attsetX}{\mathnotation{\mathcal{X}}}
\newcommand{\attsetY}{\mathnotation{\mathcal{Y}}}
\newcommand{\attsetZ}{\mathnotation{\mathcal{Z}}}

\newcommand{\relarr}[1]{\mathnotation{\mathbf{#1}}}
\newcommand{\arr}[1]{\mathnotation{\mathbf{#1}}}
\newcommand{\subarr}[2]{\mathnotation{\mathbf{#1}_{#2}}}
\newcommand{\emptytuple}{\mathnotation{\bot}}
\newcommand{\lenof}[1]{\ensuremath{\left\lvert #1\right\rvert}}%
\newcommand{\schema}{\mathnotation{\Sigma}}
\newcommand{\db}{\mathnotation{D}}
\newcommand{\ar}{\mathnotation{\texttt{ar}}}
\newcommand{\adomof}[1]{\textnormal{\textsf{adom}}(#1)}
\newcommand{\cval}{\mathnotation{c_{\texttt{val}}}}

\newcommand{\sel}[3]{\mathnotation{\sigma_{#2=#3}(#1)}}
\newcommand{\proj}[2]{\mathnotation{\pi_{#2}(#1)}}
\newcommand{\join}[2]{\mathnotation{#1 \Join #2}}
\newcommand{\sjoin}[2]{\mathnotation{#1 \ltimes #2}}

\newcommand{\agmterm}[1][R_i]{\mathnotation{\prod_{i=1}^{m}\lvert #1\rvert^{x_i}}}
\newcommand{\agmbound}[1][q,D]{\ensuremath{\textnormal{\textsf{agm}}(#1)}}
\newcommand{\minrelsize}[1][R_i]{\mathnotation{\min_{1\le i\le m}\lvert #1\rvert}}

\newcommand{\query}{\mathnotation{q}}
\newcommand{\queryresult}[2]{#1(#2)}

\newcommand{\atom}{\mathnotation{\mathsf{A}}}
\newcommand{\freeof}[1]{\ensuremath{\mathsf{free}(#1)}}
\newcommand{\querytree}[1]{\ensuremath{T_{#1}}}
\newcommand{\bag}{\ensuremath{\chi}}
\newcommand{\cover}{\ensuremath{\mu}} %
\newcommand{\ghw}{\mathnotation{\mathsf{ghw}}}
\newcommand{\fghw}{\mathnotation{\mathsf{fghw}}}

\newcommand{\relsof}[1]{\ensuremath{\mathcal{R}_{#1}}}

\newcommand{\padeps}{\lambda}

\newcommand{\WSBounds}[2]{Work \(\bigO(#1)\) and space \(\bigO(#2)\)}
\newcommand{\wsbounds}[2]{work \(\bigO(#1)\) and space \(\bigO(#2)\)}

\newcommand{\WSBoundsDictionary}[2]{\WSBounds{#1}{#2} in the dictionary setting}
\newcommand{\keyf}{\textnormal{\textsf{key}}}

\newcommand{\calM}{\mathcal{M}}
\newcommand{\tree}{T}
\newcommand{\tnode}{v}
\newcommand{\tnodeB}{w}

\newcommand{\RestateRemark}[1]{{\normalfont\bfseries #1}}
\newcommand{\RestateInit}[1]{\newcommand{#1}{}}
\newcommand{\RestateGo}[1]{\renewcommand{#1}{(Restated)}}

\newcommand{\qplan}{\mathnotation{\mathcal{P}}} %
\newcommand{\operation}[2]{\mathnotation{#1(#2)}}
\newcommand{\operationname}[1]{\mathnotation{\operatorname{\textnormal{\texttt{#1}}}}}

\newcommand{\opNumTuples}[1]{\opnameNumTuples(#1)}
\newcommand{\opnameNumTuples}{\mathnotation{\operationname{NumTuples}}}

\newcommand{\opEqual}[1]{\operation{\opnameEqual}{#1}}
\newcommand{\opnameEqual}{\mathnotation{\operationname{Equal}}}

\newcommand{\opLessThan}[1]{\operation{\opnameLessThan}{#1}}
\newcommand{\opnameLessThan}{\mathnotation{\operationname{LessThan}}}

\newcommand{\opEqualConst}[1]{\operation{\opnameEqualConst}{#1}}
\newcommand{\opnameEqualConst}{\mathnotation{\operationname{EqualConst}}}

\newcommand{\opOutput}[1]{\operation{\opnameOutput}{#1}}
\newcommand{\opnameOutput}{\mathnotation{\operationname{Output}}}

\newcommand{\opSort}[2][\padeps]{\operation{\opnameSort[#1]}{#2}}
\newcommand{\opnameSort}[1][\padeps]{\mathnotation{\operationname{Sort}_{#1}}}

\newcommand{\opKey}[1]{\operation{\opnameKey}{#1}}
\newcommand{\opnameKey}{\mathnotation{\operationname{KeyOf}}}

\newcommand{\opKeyOut}[1]{\operation{\opnameKeyOut}{#1}}
\newcommand{\opnameKeyOut}{\operationname{KeyOutput}}

\newcommand{\opSearchTuple}[1]{\operation{\opnameSearchTuple}{#1}}
\newcommand{\opnameSearchTuple}{\operationname{SearchTuple}}

\newcommand{\opSearchTuples}[1]{\operation{\opnameSearchTuples}{#1}}
\newcommand{\opnameSearchTuples}{\operationname{SearchTuples}}

\newcommand{\opnameComputeHashvalues}{\operationname{ComputeHashvalues}}

\newcommand{\opCompact}[2][\padeps]{\operation{\opnameCompact[#1]}{#2}}
\newcommand{\opnameCompact}[1][\padeps]{\mathnotation{\operationname{Compact}_{#1}}}

\newcommand{\opDeduplicate}[1]{\operation{\opnameDeduplicate}{#1}}
\newcommand{\opnameDeduplicate}{\operationname{Deduplicate}}

\newcommand{\opSelection}[2][{\attA=a}]{\operation{\opnameSelection[#1]}{#2}}
\newcommand{\opnameSelection}[1][{\attA=a}]{\mathnotation{\operationname{Select}_{#1}}}

\newcommand{\opProjection}[2][\attA]{\operation{\opnameProjection[#1]}{#2}}
\newcommand{\opnameProjection}[1][\attA]{\mathnotation{\operationname{Project}_{#1}}}

\newcommand{\opJoin}[1]{\operation{\opnameJoin}{#1}}
\newcommand{\opnameJoin}{\mathnotation{\operationname{Join}}}

\newcommand{\opSemiJoin}[1]{\operation{\opnameSemiJoin}{#1}}
\newcommand{\opnameSemiJoin}{\mathnotation{\operationname{SemiJoin}}}

\newcommand{\opDifference}[1]{\operation{\opnameDifference}{#1}}
\newcommand{\opnameDifference}{\mathnotation{\operationname{Difference}}}

\newcommand{\opUnion}[1]{\operation{\opnameUnion}{#1}}
\newcommand{\opnameUnion}{\mathnotation{\operationname{Union}}}

\setlist[itemize]{labelindent=*,leftmargin=*,label={{\small\(\blacktriangleright\)}}}

\begin{document}
\title{Work-Efficient Query Evaluation in Constant Time with PRAMs}\thanks{This work was funded by the Deutsche Forschungsgemeinschaft (DFG, German
Research Foundation) --- SCHW 678/8-1.}

\author[J.~Keppeler]{Jens Keppeler}[a]
\author[T.~Schwentick]{Thomas Schwentick\lmcsorcid{0000-0002-1062-922X}}[a]
\author[C.~Spinrath]{Christopher Spinrath}[b]

\address{TU Dortmund University, Germany}%
\email{jens.keppeler@tu-dortmund.de, thomas.schwentick@tu-dortmund.de}

\address{Lyon 1 University, Liris CNRS, France}%
\email{christopher.spinrath@liris.cnrs.fr}
\thanks{The third author made his main contributions to this article while working at the TU Dortmund University.}

\subjclass[2012]{Theory of computation~Shared memory algorithms}
\subjclass[2012]{Theory of computation~Database query processing and optimization (theory)}

\titlecomment{%
	A preliminary version of this article was published in the proceedings of the 26th International Conference on Database Theory \cite{DBLP:conf/icdt/KeppelerSS23}.
}

\begin{abstract}
The article studies query evaluation in parallel constant time in the CRCW PRAM model. While it is well-known that all relational algebra queries can be evaluated in constant time on an appropriate CRCW PRAM model, this article is interested in the efficiency of evaluation algorithms, that is, in the number of processors or, asymptotically equivalent, in the work. Naive evaluation in the parallel setting results in huge (polynomial) bounds on the work of such algorithms and in presentations of the result sets that can be extremely scattered in memory.
The article discusses some obstacles for constant-time PRAM query evaluation. It presents algorithms for relational operators and  explores three settings, in which efficient sequential query evaluation algorithms exist:  acyclic queries, semijoin algebra queries, and join queries --- the latter in the worst-case optimal framework.
Under mild assumptions --- that data values are numbers of polynomial size in the size of the database or that the relations of the database are suitably sorted --- constant-time algorithms are presented that are weakly work-efficient in the sense that work  $\mathcal{O}(T^{1+\varepsilon})$ can be achieved, for every $\varepsilon>0$, compared to the time $T$ of an optimal sequential algorithm.
Important tools are the algorithms for approximate prefix sums and compaction from Goldberg and Zwick (1995).
 \end{abstract}

\maketitle
\section{Introduction}\label{section:ctp:intro}

Parallel query evaluation has been an active research area during the
last 10+ years. Parallel evaluation algorithms have been thoroughly
investigated, mostly using the Massively Parallel Communication (MPC) model
\cite{DBLP:journals/jacm/BeameKS17}. For surveys, we refer to
\cite{KoutrisSS18,HuY20}.

Although the MPC model is arguably very well-suited
to study parallel query evaluation, it is not the
only model to investigate parallel query evaluation. Indeed, there is also the
Parallel Random Access Machine (PRAM) model, a more
``theoretical'' model which allows for a more
fine-grained analysis of parallel algorithms, particularly in
non-distributed settings. It was shown by Immerman
\cite{Im88a,Immerman} that PRAMs with polynomially many processors can evaluate first-order formulas
and thus relational algebra queries in time $\bigO(1)$.

In the study of PRAM algorithms it is considered very important that
algorithms perform well, compared with optimal sequential algorithms.  The
overall number of computation steps in a PRAM-computation is
called its \emph{work}. An important goal is
to design parallel algorithms that are \emph{work-optimal} in the
sense that their work asymptotically matches the running time of the
best sequential algorithms.

Obviously, for $\bigO(1)$-time PRAM
algorithms the work and the number of processors differ at most by a
constant factor. Thus, the result by Immerman shows that   relational
algebra queries can be evaluated in parallel constant time with polynomial work.
Surprisingly, to the best of our knowledge, work-efficiency of $\bigO(1)$-time PRAM
algorithms for query evaluation has not been investigated in the
literature. This article is meant to lay some groundwork in this direction.

The proof
of the afore-mentioned result that each relational algebra query
can be evaluated in constant-time by a PRAM with polynomial work is
scattered over several articles. It consists basically of three steps, translating from queries
to first-order logic formulas \cite{Codd72}, to bounded-depth
circuits \cite{BIS90}, and then to PRAMs \cite{Im88a}. A complete
account can be found in a monograph by Immerman \cite{Immerman}. It was
not meant as a ``practical translation'' of queries and does not yield
one. However, it is not hard to see directly that the operators of the
relational algebra can, in principle, be evaluated in constant time on
a PRAM. It is slightly less obvious, though, how the output of such an
operation is represented, and how it can be fed into the next operator.

\begin{exa}
Let us consider an example to illustrate some issues of constant-time
query evaluation on a PRAM. Let
$q$ be the following conjunctive query, written in a rule-based
fashion, for readability.
\[ q(x,y,z) \gets E(x,x_1), E(x_1,x_2), E(y,y_1), E(y_1,y_2), E(z,z_1),
  E(z_1,z_2), R(x_2,y_2,z_2)
  \]

  A (very) naive evaluation algorithm can assign one processor to each
  combination of six $E$-tuples and one $R$-tuple, resulting in work
  $\bigO(|E|^6 |R|)$. Since the query is obviously acyclic, it can
  be evaluated  more efficiently  in the spirit of Yannakakis'
  algorithm \cite{DBLP:conf/vldb/Yannakakis81}. After its
  semijoin phase, the join underlying the
  first two atoms $E(x,x_1), E(x_1,x_2)$ can be computed as a sub-query $q_1(x,x_2)
  \gets E(x,x_1), E(x_1,x_2)$. This can be evaluated by $|E|^2$ many
  processors, each getting a pair of tuples $E(a_1,a_2), E(b_1,b_2)$
  and producing output tuple $(a_1,b_2)$ in case $a_2=b_1$. The output
  can be written into a two-dimensional table, indexed in both
  dimensions by the tuples of $E$. In the next round, the output
  tuples can be joined with tuples from $R$ to compute  $q_2(x,y_2,z_2)
  \gets E(x,x_1), E(x_1,x_2), R(x_2,y_2,z_2)$. However, since it is
  not clear in advance, which entries of the two-dimensional table
  carry $q_1$-tuples, the required work is about $|E|^2 \cdot
  |R|$. Output tuples of $q_2$ can again be written into a two-dimensional
  table, this time indexed by one tuple from $E$ (for $x$) and one tuple from
  $R$ (for $y_2$ and $z_2$). Proceeding in a similar fashion,
  $q$ can be evaluated with work   $\bigO(|E|^3|R|)$.  Other evaluation orders
  are possible but result in similar work bounds.
 In terms of the input size of the database, this amounts to
 $\bigO(\IN^4)$, whereas Yannakakis' algorithm yields
 $\bigO(\IN\cdot \OUT)$ in the sequential setting, where \IN denotes
 the number of tuples in the given relations  and \OUT
 the size of the query result, respectively.

 Let us take a look at the representation of the output of this
 algorithm. The result tuples reside in a three-dimensional table, where each dimension is indexed
 by tuples from $E$. They are thus scattered over a space of size
 $|E|^3$, no matter the size of the result. Furthermore, the same
 output tuple might occur several times due to  several valuations. To
 produce a table, in which each output tuple occurs exactly once, a
 deduplication step is needed, which could yield additional work in the order of
  $|E|^6$.
 \end{exa}

The example illustrates two challenges posed by the  $\bigO(1)$-time PRAM
setting, which will be discussed in more detail in \autoref{section:ctp:lower-bounds}.
\begin{itemize}
\item It is, in general, not
  possible to represent the result of a query in a compact form, say,
  as an array, contiguously filled with result tuples. In the above
  analysis, this obstacle
  yields upper bounds in terms of the size of the input database,
  but not in the size of the query result.
\item It might be necessary to deduplicate output (or intermediate)
  relations, but, unfortunately, this cannot be done by sorting a
  relation, since (compact) sorting in constant time is not possible on PRAMs, either.
\end{itemize}

We will use compaction techniques for
PRAMs from Goldberg and Zwick \cite{GoldbergZ95} to deal with the first
challenge. The second challenge is addressed by a suitable
representation of the relations. Besides a general setting without any
assumptions, we consider a setting, where data items are mapped to
a linear-sized\footnote{As will be discussed at the end of \autoref{section:ctp:othersettings} an
  initial segment of polynomial size yields the same results, in
  most cases.} initial segment of the natural numbers by a dictionary, and a
setting where a linear order for the data values is available.

~

Although every query of the relational algebra can be evaluated by a
\ctpalgo with polynomial work, the polynomials can be arbitrarily
bad.
In fact, that a graph has a $k$-clique can be expressed by a
conjunctive query (cf.\ \autoref{section:ctp:acyclic-queries}) with $k$ variables, and it follows from Rossman's
$\omega(n^{k/4})$ lower bound for the size of bounded-depth circuit
families for $k$-Clique \cite{Rossman08} that any \ctpalgo that evaluates this query
needs work $\omega(n^{k/4})$.  Furthermore, it is  conjectured
that the $k$-clique problem cannot be solved in time significantly less than $\bigO(n^k)$ by any sequential combinatorial\footnote{The term \emph{combinatorial} is frequently used in the Algorithms community. It does not have a formal definition. In \cite{AbboudBW15} it is said that it ``should be interpreted as any practically eﬃcient algorithm that
does not suﬀer from the issues of FMM such as large constants and ineﬃcient memory usage''. Here, FMM refers to Fast Matrix Multiplication.} algorithm \cite{AbboudBW15,AbboudBBK17} and therefore it cannot be solved  by a constant-time parallel algorithm with work less than $\bigO(n^k)$, either.\footnote{For arbitrary algorithms this conjectured bound is $\bigO(n^{\smash{\frac{k\omega}{3}}})$, where $\omega$ is the fast matrix multiplication exponent. }

Hence, it cannot be expected that conjunctive queries can be evaluated in constant time and small work, in general. Therefore, we consider  restricted query classes like acyclic conjunctive queries or the semijoin algebra, where much better bounds for sequential algorithms are possible.
We show that, for each $\varepsilon>0$, there are Yannakakis-based algorithms for
acyclic join queries and free-connex acyclic join queries in the dictionary setting with upper work bounds of $\bigO((\IN\cdot
\OUT)^{1+\varepsilon})$ and  $\bigO((\IN+
\OUT)^{1+\varepsilon})$, respectively.  Two other results are work-optimal algorithms
for queries of the
semijoin algebra, and  algorithms for join queries that are almost
worst-case optimal and almost
work-optimal.

We emphasise that the stated results do not claim a fixed algorithm
that has an upper work bound $W$ such that, for every
$\varepsilon>0$,  it holds $W\in\bigO((\IN\cdot
\OUT)^{1+\varepsilon})$. It rather means that there is
a uniform algorithm that has $\varepsilon$  as a parameter and has the stated
work bound, for each fixed $\varepsilon>0$. The linear factor hidden
in the $\bigO$-notation thus depends on $\varepsilon$. This holds analogously
for our other upper bounds of this form.

Furthermore, we consider the semijoin algebra, conjunctive queries with their generalised hypertree width as a
parameter, and worst-case optimal join
evaluation.

\paragraph*{Outline.}
This article roughly consists of three parts.

The first part consists of three sections. \autoref{section:ctp:lower-bounds} reviews some upper and lower bounds for CRCW PRAMs for
some basic operations from the literature and  presents a constant
time algorithm that sorts numbers of polynomial size in a padded
fashion. \autoref{section:ctp:db-basics} explains how databases are
represented by arrays and discusses the three settings that we
consider.
\autoref{section:ctp:alg-for-basic-ops} presents some algorithms for basic operations on such arrays.

The second part presents algorithms for relational operators in
\autoref{section:ctp:alg-for-db-ops}.

The third part studies query evaluation in the above mentioned
three contexts in which (arguably) efficient algorithms exist for
sequential query evaluation.  \autoref{section:ctp:query-evaluation} presents our results for the dictionary setting and
\autoref{section:ctp:othersettings} discusses how results for the other settings can be obtained,
primarily by transforming the database at hand into the dictionary
setting. The results for query evaluation are summarised in \autoref{table:overview-main-results}.

The article ends with a conclusion and some
additional details about the results of Goldberg and Zwick \cite{GoldbergZ95} in the
appendix.

\paragraph*{Related versions.}
This article is an extended version of our ICDT 2023 paper
\cite{DBLP:conf/icdt/KeppelerSS23}. However, whereas that paper relied
on the compaction technique of Hagerup
\cite{DBLP:journals/ipl/Hagerup92}, this article uses
techniques from Goldberg and Zwick \cite{GoldbergZ95}, which came to our attention
shortly before the final version of \cite{DBLP:conf/icdt/KeppelerSS23}
was due. These techniques resulted in considerably better  bounds for acyclic and
free-connex acyclic queries. %
Furthermore, in contrast to the conference paper, we additionally address the evaluation of these queries and queries of the semijoin algebra in the ordered setting, and the (weakly) worst-case optimal evaluation of join queries in the general and the dictionary setting.

Most results are also part of the thesis of one of the authors \cite[Chapter~3]{thesis:christopher}, which has been written in parallel to this article.
However, the exact definitions of the settings and the presentation of the algorithms and results have diverged and differ slightly.

\paragraph*{Related work.} %
The initial motivation for our research was to study parallel constant-time algorithms that maintain query results in a dynamic setting, where the database can change over time. This is related to the \emph{dynamic complexity} framework of Patnaik and Immerman
\cite{PatnaikI97} and similarly Dong and Su \cite{DongS95}. In this framework, it is studied, which queries can be maintained by dynamic programs
that are allowed to perform a constant time parallel PRAM computation
after each change in the database. In fact, the dynamic programs are
usually expressed by first-order formulas, but, as stated above, the
$\bigO(1)$-time PRAM point of view is equivalent. Research on
dynamic complexity  has mostly concentrated on queries that cannot be expressed in the
relational algebra (e.g., involving transitive closure as for the
reachability query). It has mainly considered expressibility
questions, i.e., which queries can be maintained in constant time, at
all. Recently, the investigation of such dynamic programs from the work-efficiency point of
view has been initiated \cite{SchmidtSTVZ21}. In this line of research, the plan was  to study
the work-aware dynamic complexity of database queries. However, it then
turned out that, due to a lack of existing results, it makes sense to first study how work-efficient queries can be evaluated by constant-time PRAM algorithms in the purely static setting.

As mentioned above, to the best of our knowledge, work-efficiency of \ctpalgos has not been investigated in the literature.
However, parallel query evaluation has, of course, been studied for various other models.
In a recent article, Wang and Yi studied query evaluation by circuits
\cite{WangY22}.
Although this is in principle closely related, the article ignores polylogarithmic depth factors and therefore does not study $\bigO(1)$-time.

Denninghoff and Vianu \cite{DenninghoffV91} proposed \emph{database method schemas} as model for studying parallel query evaluation.
While these schemas are defined for \emph{object-oriented databases} they can also be applied to relational databases.
However, although \cite{DenninghoffV91}  considers constant-time parallel evaluation, they do not study the work of \ctpalgos.
Notably, they prove that PRAMs and database method schemas can simulate each other within a logarithmic time factor under some (reasonable) assumptions \cite[Section~5]{DenninghoffV91}.

Finally, query evaluation has been studied extensively in the Massively Parallel Communication (MPC) model, originally introduced by Beame et al. \cite{DBLP:journals/jacm/BeameKS17}.
Particularly noteworthy in the context of our research is that worst-case optimal algorithms \cite{Koutris2016} and Yannakakis' algorithm \cite{DBLP:conf/pods/Hu019} have been studied within the MPC model.
However, the research focussed -- as the name of the model suggests -- on the communication between servers, or more precisely the number of required communication rounds and the \emph{load} -- that is, the maximal amount of data a server receives during a round.
For the actual computation in between the communication phases, the MPC model has no notion of computation steps; which is why, for studying work bounds, we opted for the PRAM model instead.

\section{A Primer on Constant-Time CRCW PRAM Algorithms}\label{section:ctp:lower-bounds}

In this section, we define PRAMs and review some results about $\bigO(1)$-time parallel
algorithms on CRCW PRAMs from the literature. More precisely, we recall some negative
results, in particular for sorting and some positive results,
involving quite powerful constant-time algorithms for approximative
prefix sums and linear compaction, and padded sorting.
We also give a  $\bigO(1)$-time algorithm for padded
sorting of sequences of numbers of polynomial size that will be useful
for our purposes.

\subsection{Parallel Random Access Machines (PRAMs)}\label{section:prams}
A \emph{parallel random access machine (PRAM)} consists of a number of
processors that work in parallel and use a shared memory.
The processors are consecutively numbered from \(1\) to some number
\(p_\textnormal{max}\) that depends on the size of the input.
Each processor can access its \emph{processor number} and freely use it in computations.
For instance, a common application is to use it as an array index.
The memory consists of finitely many, consecutively numbered memory cells.
The number of a memory cell is called its \emph{address}.
The input for a PRAM consists of a sequence of $n$ but strings of
length $\bigO(\log n)$, residing in some designated input memory cells.
We presume that the \emph{word size}, that is, the number of bits that
can be stored in a memory cell, is bounded by \(\bigO(\log
n)\).

A processor can access any memory cell in
one time step, given its address.
Furthermore, we assume that the usual arithmetic and bitwise
operations can be performed in one time step by any processor.

The \emph{work}  of a
PRAM computation is the sum of the number of all computation steps of
all processors made during the computation.
We define the space  required by a PRAM computation as the maximal
address of any memory cell accessed during the
computation.\footnote{For simplicity, we do not count the space that
  a processor might need for its working register. For some of our
  algorithms, the space of a
  computation under this definition might be significantly smaller than the number of
  processors. However, we never state bounds of that kind, since we do
  not consider them useful. Thus, our space bounds for an algorithm are always at least
  as large as the work bounds.}

We mostly use
the \emph{Concurrent-Read Concurrent-Write model (CRCW PRAM)},
i.e., processors are allowed to read and write concurrently from and to
the same memory location. Weaker PRAM models as CREW or EREW are not able to evaluate relational algebra queries in constant time. In fact, CREW PRAMs require time $\Omega(\log n)$ to compute the logical OR of $n$ bits, and therefore cannot compute the query $\pi_\emptyset(R)$ deciding whether relation $R$ is empty in constant time \cite[Theorem 7]{CookDR86}.

More precisely, we mainly assume the
\emph{Arbitrary} CRCW PRAM model: if multiple processors concurrently write
to the same memory location, one of them, \enquote{arbitrarily},
succeeds. This is our standard model, and we often refer to it simply by \emph{CRCW PRAM}.

However, for some algorithms the weaker \emph{Common} model suffices,
where all processors need to write the same value into the same
location. On the other hand, the lower bounds below even hold for the
stronger \emph{Priority} model, where always the processor with the
smallest processor number succeeds in writing into a memory
cell.

It is well known that the
weakest of the three models, the Common CRCW PRAM, can simulate one step of the strongest model, the Priority CRCW PRAM, in
constant time
\cite[Theorem 3.2]{ChlebusDHR88}. The simulation in \cite{ChlebusDHR88} only requires an extra $\bigO(\log n)$ work factor,  but  might require an extra $\bigO(n)$-factor of space. The technique of \cite[Theorem 3.2]{ChlebusDHR88} can be easily adapted to yield an extra $\bigO(n^\epsilon)$ work factor,  but  also only an extra $\bigO(n^\epsilon)$-factor of space, for any $\epsilon>0$. Many of our upper bounds have an extra $n^\epsilon$-factor for work. For those results, the difference between the precise concurrent write mechanisms is not essential.

For a few algorithms even the Exclusive-Read\emph{
Exclusive-Write model (EREW PRAM)} model, where concurrent access is
forbidden (Proposition \ref{alg:selection}) or the \emph{CREW PRAM} model, where concurrent writing  is
forbidden (Propositions \ref{result:alg-search-fullylinkedB} and \ref{result:alg-search-fullylinkedA}), are sufficient.

We refer to \cite{DBLP:books/aw/JaJa92} for more details on PRAMs and to \cite[Section~2.2.3]{DBLP:books/el/leeuwen90/Boas90} for a discussion of alternative space measures.

As mentioned in  the introduction, the PRAM model is a theoretical algorithmic model. Most of the PRAM-related research has been done in the 1990s or early 2000s. However, given the growing built-in parallelism of modern hardware there has been considerable recent interest in the model both with respect to algorithms (e.g., \cite{CaoF23,AgarwalKLPWWZ24}), sometimes in the more abstract work-depth (or work-depth) model (see, e.g., \cite{LeeVC07}) and with respect to the transferability of PRAM algorithms to real computers (e.g., \cite{DhulipalaBS21,GhanimEB21,LiuFLZL24}).

\subsection{Lower and Upper Bounds for Constant-Time CRCW PRAM Algorithms}

Next, we are going to review some known results about PRAM algorithms
for some basic algorithmic tasks from the literature. These tasks deal
with arrays, a very natural data structure for PRAMs. Before we
describe the results, we first fix some notation for arrays.

An array \arr{A} is a sequence of consecutive memory cells.
The \emph{size} \(\lenof{\arr{A}}\) of \(\arr{A}\) is the number of memory cells it consists of.
By $\arr{A}[i]$ for $1\le i\le \lenof{\arr{A}}$ we refer to the $i$-th cell of \arr{A}, and to the
(current) content of that cell.
We call $i$ the \emph{index} of cell \(\arr{A}[i]\).
Furthermore, we assume that the size of an array is always available to all processors (for instance, it might be stored in a \enquote{hidden} cell with index $0$).
Given an index $i$, any processor can access cell $\arr{A}[i]$ in $\bigO(1)$ time with $\bigO(1)$ work.

We will also use arrays whose cells correspond to \(c\) underlying memory cells each, for some constant integer \(c > 0\).
In terms of some programming languages this corresponds to an array of structs (or records).
Since a single processor can read and write \(c\) memory cells in
constant time, we often blur this distinction.
For instance, we might write \(\arr{A}[i]\) for the $i$-th block of
\(c\) memory cells in some  array with  \(cn\) memory cells. In
particular, this will allow us to refer to the $i$-th tuple of a
relation that is stored in that array in a concise way.

\autoref{result:sort-lower} below  implies that we cannot expect, in general, that
query results can be stored in a compact fashion, that is, the result
$t_1, \ldots, t_m$ of $m$ tuples is stored in an array of size $m$.
This follows from the following lower bound for computing the parity
function on CRCW PRAMs. This Boolean function tests whether the number of ones in a given bit string is odd.

\begin{propC}[{\cite[Theorem~10.8]{DBLP:books/aw/JaJa92}}]\label{prop:lb-parity}
	Any algorithm that computes the parity function of $n$
        variables and uses a polynomial number of processors requires $\Omega\left(\frac{\log{n}}{\log{\log{n}}}\right)$ time on a Priority CRCW PRAM.
\end{propC}

We thus consider arrays which may have \emph{empty cells}.\footnote{In terms of some programming languages an empty cell corresponds to a cell containing a special \texttt{null} value. It can be indicated by storing a special value.}
An array \emph{without} empty cells is \emph{compact}.
We say that an algorithm \emph{compacts} an array \(\arr{A}\) with \(k\) non-empty cells, if it computes a compact array of size \(k\) containing all values from the non-empty cells of~\(\arr{A}\).
\autoref{prop:lb-parity} implies that such an algorithm cannot run in
constant-time on our PRAM models, since otherwise the parity function
could be computed by making all cells with value~\(0\) empty, compacting
and testing whether the compacted array has odd length.
  Similarly, \autoref{prop:lb-parity}  also implies a lower bound for sorting an array.
\begin{cor}\label{result:sort-lower}\hfill
  \begin{enumerate}[(a)]
  \item 	Any algorithm that compacts an array with $n$ cells,
    some of which might be empty,  and uses a polynomial number of processors,
    requires $\Omega\left(\frac{\log{n}}{\log{\log{n}}}\right)$ time
    on a Priority CRCW PRAM.\label{result:sort-lower-compaction}
    \item 	Any algorithm that sorts an array with $n$ natural
    numbers and uses a polynomial number of processors
    requires $\Omega\left(\frac{\log{n}}{\log{\log{n}}}\right)$ time
    on a Priority CRCW PRAM.
    This holds even if all numbers in the array are from \(\{0,1\}\).
  \end{enumerate}
\end{cor}
Since the priority model is a stronger model, the lower bounds stated in \autoref{result:sort-lower} also apply to the Arbitrary and Common CRCW PRAM models, cf., e.g.\ \cite{DBLP:books/aw/JaJa92}.

In the remainder of the section, we have a closer look at further lower bounds and at almost matching upper bounds for approximate compaction and sorting.

\subsubsection{Approximate Compaction}\label{subsec:aprroxcompact}

\autoref{result:sort-lower}
rules out any constant-time algorithms that compact a
non-compact array (for the stated models). Therefore, the best we can
hope for is approximate compaction, i.e., compaction that guarantees
some maximum size of the result array in terms of the number of
non-empty entries. In the following, $\lambda(n)$ will be an
\emph{accuracy function}.

We say that an array is
of size $n$ is \emph{\(\lambda\)-compact}, for a function
\(\lambda\), if it has at least \(\frac{n}{1+\lambda(n)}\)
non-empty cells. That is, $n$ is at most $(1+\lambda(n))m$, where $m$
is the number of non-empty cells. An
algorithm achieves \emph{$\lambda$-approximate
  compaction}, if it compacts   arrays into  \(\lambda\)-compact
arrays. We call it \emph{order-preserving}, if the relative order of
array elements does not change.

We can always assume that $\lambda(n)\ge
\frac{1}{n}$ holds, since otherwise, even with only one empty cell, we would have  $n\le (1+\lambda(n))(n-1)<n$.
Later on in this article $\lambda$ will be a constant or of
the form \((\log n)^{-a}\), for some $a$.

We will see that approximate compaction can be done quite well, but the
following result implies a trade-off between time and number of
processors.

\begin{propC}[{\cite[Theorem~5.3]{DBLP:journals/iandc/Chaudhuri96}}]\label{prop:lowerbound:lac}
Any algorithm that achieves  $\lambda$-approximate compaction, for
an accuracy function $\lambda$, and uses
$\mu(n)n$ processors on a Priority CRCW PRAM, for some function $\mu$, requires time
$\Omega(\log{\frac{\log{n}}{\log{(\mu(n)\lambda(n) + 2)}}})$.
\end{propC}
\begin{cor}\label{coro:compact-delta}
  Any algorithm that achieves  $\lambda$-approximate compaction  for a
  constant $\lambda>0$  within
  constant time requires $\Omega(n^{1+\varepsilon})$ processors, for some constant
  $\varepsilon>0$.
\end{cor}

\begin{proof}
  Let an algorithm be given that achieves  $\lambda$-approximate
  compaction in constant time~$t$. By \autoref{prop:lowerbound:lac} there
  is a constant $c$ such that $t\ge c
  \log{\frac{\log{n}}{\log{(\mu(n)\lambda + 2)}}}$.
  Therefore, $2^t \ge \big(\frac{\log{n}}{\log{(\mu(n)\lambda +
      2)}}\big)^c$, and thus $\frac{\log{n}}{\log{(\mu(n)\lambda +
      2)}}\le 2^{t/c}$ and $\log{n}\le 2^{t/c} \log{(\mu(n)\lambda +
      2)}$.
  Finally, we get $(\mu(n)\lambda + 2)^{2^{t/c}} \ge n$ and therefore
  $\mu(n)\lambda + 2 \ge n^{1/c'}$, where $c'=2^{t/c}$
  is a constant. Since we assumed $\lambda$ to be constant, we can
  conclude $\mu(n)=\Omega(n^{1/c'})$ and the statement of the
  corollary follows.
\end{proof}

As a consequence of \autoref{result:sort-lower} and \autoref{coro:compact-delta}, our
main data structure to represent intermediate and result relations are
(not necessarily compact) arrays. Furthermore, the best work bounds we
can expect for algorithms that compute $c$-approximate arrays, for
some constant $c$,  are of the form $\bigO(n^{1+\varepsilon})$, for
some $\varepsilon>0$.

It turns out that $\lambda$-approximate compaction algorithms matching the
lower bound of \autoref{coro:compact-delta} indeed exist, even
order-preserving ones, and for $\lambda\in o(1)$.

\begin{propC}[{\cite{GoldbergZ95}}]\label{result:lin-compaction}
	For every \(\varepsilon>0\) and \(a>0\) there is an \ctpalgo
        that achieves order-preserving $\lambda$-approximate
        compaction for $\lambda(n)=(\log n)^{-a}$.
	The algorithm requires work and space \(\bigO(n^{1+\varepsilon})\) on a Common CRCW PRAM.
      \end{propC}
In particular, order-preserving $c$-approximate
compaction is possible for every constant $c>0$.

      \autoref{result:lin-compaction} is not stated in this form in \cite{GoldbergZ95}, but
      it readily follows from the next result on consistent
      approximate prefix sums.

      Let  \arr{A} be an array of size $n$ containing  integers. An
      array \arr{B} of the same length \emph{contains $\lambda$-consistent
		prefix sums}  for \arr{A}, if, for each $i\in\range{n}$, it holds
      \begin{itemize}
      \item $\sum_{j=1}^i \arr{A}[j] \le \arr{B}[i] \le (1+\lambda(n))
        \sum_{j=1}^i \arr{A}[j]$ and
      \item $\arr{B}[i]-\arr{B}[i-1] \ge \arr{A}[i]$.
      \end{itemize}
		We revisit the proof, outlined by \cite{GoldbergZ95}, of the following result in \autoref{sec:approx-prefix-sums} for the sake of completeness, and to assert the space bound, which is not explicitly stated by~\cite{GoldbergZ95}.
\RestateInit{\restateprefixsums}
\begin{restatable}[{{\cite[Theorem~4.2]{GoldbergZ95}}}]{propC}{prefixsums}
   \label{result:prefix-sums}\RestateRemark{\restateprefixsums}
	For every \(\varepsilon>0\) and \(a>0\) there is an \ctpalgo
        that computes $\lambda$-consistent prefix sums for any array of
		  length $n$ with $(\log n)$-bit integers, for $\lambda(n)=(\log n)^{-a}$.
	The algorithm requires work and space \(\bigO(n^{1+\varepsilon})\) on a Common CRCW PRAM.
      \end{restatable}
      Indeed, \autoref{result:lin-compaction} follows by using an
      array \arr{A} such that $\arr{A}[j]=1$ if the $j$-th cell of the
      array that is to be compacted  is non-empty and 0, otherwise. The approximate prefix sums then
      yield the new positions for the non-empty cells.

	  We emphasise that in \autoref{result:lin-compaction} the size of the cell \emph{contents} does not matter.

\subsubsection{Processor Allocation}

A crucial task for PRAM algorithms is to assign processors properly to
tasks. Sometimes this is straightforward, e.g., if each element of
some array needs to be processed by a separate processor. However,
sometimes the assignment of processors to tasks is not clear, a
priori. This will, in particular, be the case for the evaluation of
the Join operator. Processor allocation can be formalised as follows.

A \emph{task description} \(d = (m, \tuple{x})\) consists of a number
\(m\) specifying how many processors are required for the task, and a
constant number of additional numbers \(\tuple{x} = (x_1,\ldots,x_k)\)
serving as \enquote{input} for the task, e.g.\ \(\tuple{x}\) may
contain links to input arrays or which algorithm is used to solve the
task.%
A \emph{task schedule} for a sequence \(d_1,\ldots,d_n\) of task descriptions is an array \(\arr{C}\) of size at least \(\sum_{i=1}^n m_i\) such that, for every \(i\in\range{n}\), there are at least \(m_i\) consecutive cells \(\arr{C}[j_i], \ldots, \arr{C}[j_i+m_i-1]\) with content \(d_i\) and each of these cells is augmented by a link to \(\arr{C}[j_i]\), i.e.\ the cell with the smallest index in the sequence.
The tasks specified by \(d_1,\ldots,d_n\) can then be solved in parallel using \(\lenof{\arr{C}}\) processors: Processor \(j\) can lookup the task it helps to solve in cell \(\arr{C}[j]\), and using the link it can also determine its relative processor number for the task.
If a cell \(\arr{C}[j]\) is empty, processor \(j\) does nothing.

The processor allocation problem is closely related to the prefix sums and interval allocation problems.
We believe the following lemma is folklore.
The relationship of these problems and an analogous result for randomised PRAMs and \(\bigO(\log^* n)\) time are, for instance, discussed in \cite[Section~2]{DBLP:conf/stacs/Hagerup92a}.
For the sake of completeness we provide a proof for the deterministic constant-time case here.

\begin{lem}\label{result:task-scheduling}
	For every \(\varepsilon > 0\) and \(\padeps > 0\) there is a \ctpalgo that, given an array \(\arr{T}\) containing a sequence of task descriptions \(d_1 = (m_1,\tuple{x}_1),\ldots,d_n=(m_n,\tuple{x}_n)\), computes a schedule \(\arr{C}\) for \(d_1,\ldots,d_n\) of size \((1+\padeps) \sum_{i=1} m_i\).
	It requires work and space \(\bigO(\lenof{\arr{T}}^{1+\varepsilon} + \lenof{\arr{C}}^{1+\varepsilon})\) on a Common CRCW PRAM.
\end{lem}

\begin{proof}
	The algorithm first determines, for each task, a \enquote{lead processor}, which will be the processor with the lowest processor number assigned to the task.

	For this purpose, it computes consistent \(\padeps\)-approximate prefix sums \(s_1,\ldots,s_n\) for the sequence \(m_1,\ldots,m_n\) using \autoref{result:prefix-sums}.
	It then assigns the first task \(d_1\) to processor \(1\), and, for \(i\in\range[2]{n}\), task \(d_i\) to processor \(s_{i-1} + 1\).
	Consequently, the algorithm allocates an array \(\arr{C}\) of size \(s_n\) for the schedule and sets \(\arr{C}[0] = d_1\), and \(\arr{C}[s_{i-1} + 1] = d_i\), for \(i \ge 2\).
	All remaining cells are initially empty.
	Since \(s_n \le (1+\padeps)\sum_{i=1}^n m_i\), the array \(\arr{C}\) has the desired size.

	Thanks to \autoref{result:prefix-sums} computing the prefix sums requires work and space \(\bigO(\lenof{\arr{T}}^{1+\varepsilon})\).

	To assign the remaining \(m_i - 1\) processors to task
        \(d_i\), for each $i$, we observe that, for each \(i\), there are at least \(m_i - 1\) empty cells in \(\arr{C}\) between the cell containing \(d_i\) and the cell containing \(d_{i+1}\), because the prefix sums \(s_1,\dots,s_n\) are consistent.
	Thus, it suffices to compute links from each empty cell
        \(\arr{C}[j]\) to the non-empty cell  \(\arr{C}[k]\) with
        maximal number $k\le j$ and to test, whether $j-k$ is at most
        $m_i$ for the task $d_i$ with  \(\arr{C}[k]=d_i\).
	Due to the following \autoref{result:predsucc} this can be done with work and space \(\bigO(\lenof{\arr{C}}^{1+\varepsilon})\).
\end{proof}

We say that an array \arr{B} \emph{provides predecessor links} for a
(possibly non-compact)
array~\arr{A}, if it has the same length as \arr{A} and, for each
index $j$, \(\arr{B}[j]\) is the largest number $i<j$, for which
\(\arr{A}[i]\) is non-empty.
\begin{prop}\label{result:predsucc}
  For every $\varepsilon>0$, there is a  \ctpalgo that computes, for
  an array \arr{A}, an array \arr{B} that provides predecessor links
  for \arr{A} with work and space $\bigO(\lenof{\arr{A}}^{1+\varepsilon})$ on a Common CRCW PRAM.
\end{prop}
\begin{proof}[Proof sketch]
	Let $n=\lenof{\arr{A}}$ and $\delta=\frac{\varepsilon}{2}$.
        All cells of \arr{B} are initially set to a null value.
        In
        the first round, the algorithm considers subintervals of
        length $n^\delta$ and establishes   predecessor links within
        them. To this end, it uses, for each interval a
        $n^\delta{\times} n^\delta$-table whose entries $(i,j)$ with
        $i<j$ are initialised by $1$, if $\arr{A}[i]$ is non-empty
        and, otherwise, by $0$. Next, for each triple $i,j,k$ of
        positions in the interval, entry $(i,j)$ is set to $0$ if
        $i<k<j$ and   $\arr{A}[k]$ is non-empty.  It is easy to see
        that afterwards entry $(i,j)$ still carries a 1 if and only if
        $i$ is the predecessor of $j$. For all such pairs $i<j$, $\arr{B}[j]$
        is set to $i$. Finally, for each interval, the unique
        non-empty cell with a null value is identified or a flag is set if the interval has no non-empty cells at all.
	For every interval, $(n^{\delta})^3 = n^{3\delta}$ processors suffice for this computation, i.e.\ one processor for each triple $i,j,k$ of positions in the interval.
	Since there are $\frac{n}{n^\delta} = n^{1-\delta}$ intervals, this yields an overall work of $n^{1-\delta}\cdot n^{3\delta} = n^{1+2\delta} = n^{1+\varepsilon}$.
	Similarly, the algorithm uses \(n^{1-\delta}\) many tables of size \(n^{2\delta}\) each.
	Thus, it requires \(n^{1+\delta} \le n^{1+\varepsilon}\) space.

	In the next round, intervals of length $n^{2\delta}$ are
        considered and each is viewed as a sequence of  $n^{\delta}$
        smaller intervals of length $n^{\delta}$. The goal in the
        second round is to establish predecessor links for the cells
        of each of the smaller intervals that have not yet obtained a
        predecessor link. This can be done similarly with the same
        asymptotic work and space as in round 1. After
        $\lceil\frac{1}{\delta}\rceil$ rounds, this process has
        established predecessor links for all cells (besides for the
        first non-empty cell and all smaller cells).
      \end{proof}
      An analogous result holds for the computation of successor links.

      \subsubsection{Padded Sorting}
\autoref{result:sort-lower} mentions another notorious obstacle for
\ctpalgos: they cannot sort (compactly) in constant time, at all. And similarly to \autoref{coro:compact-delta} linearly many
processors do not suffice to sort  in a (slightly)
non-compact fashion. The \emph{$\lambda$-padded sorting
	problem} asks to sort $n$ given items into an array of length
$(1+\lambda)n$ with empty cells in spots without an item.

\begin{propC}[{\cite[Theorem~5.4]{DBLP:journals/iandc/Chaudhuri96}}]\label{cited-result:psort-lower-bound}
	Solving the $\lambda$-padded sorting problem  on a Priority CRCW PRAM using $\mu(n)n$ processors requires $\Omega(\log{\frac{\log{n}}{\log((\lambda(n) + 2)(\mu(n)+1))}})$ time.
\end{propC}
The following corollary follows from
\autoref{cited-result:psort-lower-bound} in the same way as \autoref{coro:compact-delta} follows
from \autoref{prop:lowerbound:lac}.
\begin{cor}\label{coro:sort-delta} %
  Any  constant-time algorithm for the $\lambda$-padded sorting problem, for a
  constant $\lambda$,
 requires $n^{1+\varepsilon}$ processors, for some constant
  $\varepsilon>0$.
\end{cor}

Thus, we cannot rely on $c$-padded sorting, for any constant $c$ in
constant time with a linear number of processors.

However, with the help of the algorithms underlying \autoref{result:prefix-sums} and \autoref{result:lin-compaction}, it is possible to sort
$n$ numbers of polynomial size in constant time with
\(\bigO(n^{1+\varepsilon})\) work in a $c$-padded fashion. We will see later that this is
very useful for our purposes since fixed-length
tuples of numbers of linear size can be viewed as numbers of
polynomial size.

\begin{prop}\label{result:sorting}
  For every constant \(\varepsilon > 0\), \(\padeps > 0\), and $c>0$ there is a
  constant-time Common CRCW PRAM algorithm that sorts up to $n$
  integers  $a_1,\ldots,a_n$ in the range
  \(\range[0]{n^c{-}1}\)  and stores them (in order) into an array of
  size \((1+\padeps)n\).
  More precisely, it computes an array of
  size \((1+\padeps)n\) whose non-empty cells induce a sequence \(i_1,\ldots,i_n\)
  such that, for all $j<\ell$, $a_{i_j}\le a_{i_\ell}$.
  The algorithm requires work and space
  \(\bigO(n^{1+\varepsilon})\).
\end{prop}
\begin{proof}
  We assume that $n$ integers $a_1,\ldots,a_n$ are
  given in an array $\arr{A}$ of size $n$.
  First, each number $a_i$ is replaced by $a_i(n+1)+i$ to guarantee
  that all numbers are distinct, in the following.
  We also choose $\delta$ as \(\frac{\varepsilon}{3}\).

  The algorithm proceeds in two phases.
  First, it computes numbers \(b_1,\ldots,b_n\) of size at most \(\bigO(n^{1+\delta})\) such that, for all \(i,j\), it holds \(a_i < a_j\) if and only if \(b_i < b_j\).
  These numbers are then used in the second step to order the original sequence \(a_1,\ldots,a_n\).

  For the first phase, we assume that $n^c>n^{1+\delta}$
  holds, since otherwise the first phase can be skipped.

  The algorithm performs multiple rounds of a bucket-like sort,
  reducing the range by applying a factor of \((1+\padeps)
  \frac{1}{n^\delta}\) in each round.

  In the following, we describe the algorithm for the first round.
  More precisely, we show that in constant time with work $\bigO(n^{1+\delta})$ numbers $b_1,\ldots,b_n$ can be computed (in an array $\arr{B}$), such that, for each $i,j$ it holds $a_i<a_j$ if and only if $b_i<b_j$ and, for every $i\le n$, it holds $b_i\le (1+\padeps)n^{c-\delta}$.

  First, for each $i\le n$, the algorithm determines $c_i,d_i$ such that $a_i=c_i n^{c-1-\delta}+d_i$ and $0\le d_i<n^{c-1-\delta}$. Since $a_i\le n^c-1$, it holds $c_i\le n^{1+\delta}-1$.
  Here \(c_i\) is the \enquote{bucket} for \(a_i\).

  Next, each \(c_i\) is replaced by a number of size at most \((1+\padeps)n\).
  To this end, let $\arr{C}$ be an array of size $n^{1+\delta}$, all of whose entries are initially empty.
  For each $i$, the number $c_i$ is stored in $\arr{C}[c_i]$.
  Next, the array $\arr{C}$ is compacted to size at most \((1+\padeps)n-1\)
  using \autoref{result:lin-compaction} with some suitable \(\padeps'<\padeps\).
  Let $\arr{C'}$ denote this compacted version of $\arr{C}$.
  Each \(c_i\) is then assigned to its index in \(\arr{C}'\) as follows:
  For each index $j$ of $\arr{C'}$, if $\arr{C'}[j]$ is non-empty then $j$ is stored in  $\arr{C}[\arr{C'}[j]]$.
  We note that since the compaction is order-preserving, we have that $\arr{C}[c_i]<\arr{C}[c_k]$ if and only if $c_i<c_k$.

  Finally, for each $i\le n$, we let $b_i$ be $\arr{C}[c_i]n^{c-1-\delta}+d_i$.
  Since $\arr{C}$ only contains numbers of size at most \((1+\padeps)n-1\), we get that
  \[
  	b_i\le ((1+\padeps)n-1) n^{c-1-\delta} + n^{c-1-\delta} = (1+\padeps)  n^{c-\delta}.
  \]
  By repeating this procedure for at most \(\bigO(c)\) times, we obtain an array $\arr{B}$ with numbers $b_1,\ldots,b_n$  of size at most \(\bigO(n^{1+\delta})\) such that, for each $i\le n$, it holds   $a_i<a_j$ if and only if $b_i<b_j$.
  This concludes the first phase of the algorithm.

  In the second phase, a new array $\arr{D}$ of size
  \(\bigO(n^{1+\delta})\) is allocated and initialised by
  setting, for each \(i\le n\), $\arr{D}[b_i]=i$ (and empty for
  indices in $\arr{D}$ for which no such $i$ exists).
  Approximate order-preserving compaction yields an array $\arr{D'}$ of size \((1+\padeps)n\) such that the sequence $i_1,\ldots,i_n$ of the non-empty cells of $\arr{D'}$ fulfils $a_{i_1}<\cdots< a_{i_n}$.

  The work in the second phase is dominated by the compaction.
  Compacting an array of size  \(\bigO(n^{1+\delta})\) can be done with work and space \(\bigO( n^{(1+\delta)^2} )\) thanks to \autoref{result:lin-compaction}.
  Since we can assume that \(\delta \le 1\) holds, we have
  \[
  	n^{(1+\delta)^2}
  	= n^{1+2\delta+\delta^2}
  	\le n^{1+3\delta}
  	= n^{1+\varepsilon}.
  \]
  Thus, the second phase requires \(\bigO(n^{1+\varepsilon})\) work and space.
  This is also an upper bound for each of the constantly many rounds of the first phase.
\end{proof}

\autoref{result:sorting} suffices for our purposes, but the algorithm
can easily be adapted to integers of unbounded size. However, in
general, its running time might not be constant but rather
\(\bigO(\log_{n} m )\) and the work can be  bounded by a term of the form \(\bigO((1+\padeps)^{(1+\varepsilon) \log_{n} m
  }n^{1+\varepsilon})\), where $m$ is the maximum size of
  numbers.
\section{Databases and their Representations for PRAM Algorithms}\label{section:ctp:basics}\label{section:ctp:db-basics}

In this section, we first recall some concepts
from database theory.
Then we discuss three settings that differ by how  a
PRAM can access the relations of  databases.

Finally, we explain how relations can be represented with the help of arrays.
While the concrete representation will depend on the setting, arrays
will provide a unified interface allowing us to represent
intermediate results and to evaluate queries in a compositional manner.

\paragraph*{Notation for databases and queries.}
A \emph{database schema \schema} is a finite set of relation symbols,
where each symbol $R$ is equipped with a finite set $\attrof{R}$ of attributes.
A tuple $t=(a_1,\ldots,a_{|\attsetX|})$ over a finite list $\attsetX =
(\attA_1,\ldots,\attA_k)$ of attributes has, for each $i$, a value
$a_i$ for attribute $\attA_i$.
Unless we are interested in the
lexicographic order of a relation, induced by $\attsetX$, we can view $\attsetX$ as
a set.
An $R$-relation is a finite set of tuples over \attrof{R}.
The \emph{arity} \(\ar(R)\) of $R$ is $\lenof{\attrof{R}}$.
For $\attsetY\subseteq \attsetX$, we write $t[\attsetY]$ for the restriction of $t$ to $\attsetY$. And for $\attsetY\subseteq \attrof{R}$, we set $R[\attsetY] = \{t[\attsetY] \mid t\in R\}$.
A database \db over \schema consists of an $R$-relation $\db(R)$, for each $R\in \schema$.
We usually write $R$ instead of $\db(R)$ if \db is understood from the context.
That is, we do not distinguish between relations and relation symbols.
The size $\lenof{R}$ of a relation $R$ is the number of tuples in $R$.
By $\lenof{\db}$ we denote the number of tuple entries in database \db.

\begin{table}
	\caption{Operators of the relational algebra; \(\attsetX\) denotes a set of attributes, \(A\) an attribute, \(x\) an attribute or a constant, and \(E_1, E_2\) expressions of the relational algebra}\label{table:ra-ops}
	\begin{tabular}{l l c l l c l l}
		\toprule
		\multicolumn{2}{l}{Unary Operators} & \hspace{2em} & \multicolumn{5}{l}{Binary Operators}\\
		\midrule
		Select & \(\sel{E_1}{A}{x}\) & \hspace{2em} & Join & \(\join{E_1}{E_2}\)  & \hspace{2em} & Semijoin & \(\sjoin{E_1}{E_2}\)\\
		Project & \(\proj{E_1}{\attsetX}\)  & \hspace{2em} & Difference & \( E_1 \setminus E_2\) & \hspace{2em} & Union & \(E_1 \cup E_2\)\\
	\end{tabular}
\end{table}

In this article we consider expressions of the relational algebra to define queries over a database schema: Every relation symbol \(R\) from the schema is an expression, and given expressions \(E_1, E_2\) can be combined into a new expressions \(\tau(E_1)\) and \((E_1 \circ E_2)\) where \(\tau\) and \(\circ\) range over the unary and binary operators of the relational algebra presented in \autoref{table:ra-ops}\footnote{The relational algebra also features an operator for renaming attributes, but since this can trivially be done in constant time by a single processor, we omit it here.}, respectively.
As usual, we omit parentheses when writing expressions when it does not lead to ambiguity.
For details on (the operators of) the relational algebra,
we refer  to \cite{ABLMP21}.

Throughout this article, \(\db\) always
denotes the underlying database.
We always assume a fixed schema and therefore, in particular, a fixed maximal arity
of
tuples.

\paragraph*{Settings.}
We next describe the three settings  that we study in this article,
distinguished by how PRAMs can interact with (input) databases.

In the most \emph{general setting}, we do not specify how databases are
actually stored. We only assume that the tuples in a relation \(R\) are numbered from \(1\) to \(\lenof{R}\) and the following elemental operations can be carried out in constant time by a single processor.
\begin{itemize}
\item \(\opEqual{R_1,i_1,j_1;R_2,i_2,j_2}\) tests whether
  the $j_1$-th attribute of the $i_1$-th tuple of $R_1$ has the same
  value as  the $j_2$-th attribute of the $i_2$-th tuple of $R_2$;
\item \(\opEqualConst{R,i,j,x}\) tests whether
    the $j$-th attribute of the $i$-th tuple of $R$ has value \(x\);
\item \(\opOutput{R,i_1,j_2;i_2,j_2}\) outputs the value of the  $j_1$-th
  attribute of the $i_1$-th tuple of $R$ as  $j_2$-th attribute of the
  $i_2$-th output tuple;
\item \(\opNumTuples{R}\) returns the number of tuples in relation $R$.
\end{itemize}

We note that the operations \(\opEqualConst{R,i,j,x}\) are only required for selections with constants, e.g.\ \(\sel{R}{\attA}{a}\) for a domain value \(a\).
For queries without constants, they can be omitted.

The second setting is the \emph{ordered setting}.
Here we assume that there is a linear order on the data values
occurring in the database \(\db\), and that the following additional
elemental operation is available to access this order.

\begin{itemize}
\item \(\opLessThan{R_1,i_1,j_1;R_2,i_2,j_2}\) tests whether
  the $j_1$-th attribute of the $i_1$-th tuple of $R_1$ is less or equal than the
  value of the $j_2$-th attribute of the $i_2$-th tuple of $R_2$.
\end{itemize}

The third setting is a specialisation of the ordered setting, in which
we assume that data items are small numbers. More precisely, in the \emph{dictionary setting} we assume that the database \db \emph{has
  small data values}, that is, that its data
values are natural numbers from $\range{\cval\lenof{\db}}$, for some fixed
constant $\cval$.
It is in the spirit  of dictionary-based
compressed databases, see, e.g., \cite{ChenGK01}. In a nutshell, such a
database has a dictionary that maps data values to natural numbers and
uses these numbers instead of the original data values throughout
query processing. Such dictionaries are often defined attribute-wise,
but for the purpose of this article we assume for simplicity that
there is only one mapping for all values of the database.

To summarise, the general setting makes no assumptions on the data,
the ordered setting assumes that the data is ordered, and the
dictionary setting further assumes that it consists of small
numbers.

We will see later how the general and the ordered setting can be
translated into the dictionary setting  with
reasonable effort, and we will mostly study actual query processing for
the dictionary setting.

We note, that the queries we consider do not refer to any
order of data values. Therefore, a possible order has no semantical
use, but it is only used to get better query evaluation algorithms.

It should be noted that the settings considered here are very
analogous to the various alphabet settings in string algorithms for
non-fixed alphabets. For \emph{general alphabets}, algorithms can only
test whether two positions carry the same symbol, for \emph{general ordered
alphabets}, they can be compared with respect to some linear
order. Other algorithms assume \emph{integer alphabets} or
\emph{polymomial integer alphabets}, where the symbols are numbers or
numbers of polynomial size in the length of the string, respectively
(see, e.g., \cite{Breslauer92,ApostolicoB96}).

Next, we will explain how data values, tuples, and relations can be represented by a PRAM in memory.
This is in particular intended for storing (intermediate) results and having an interface for (database) operations.

\paragraph*{Representing data items.}
As indicated before, in this article, databases are represented by
arrays. We first explain how data items are represented such  that
each item only requires one memory cell.

In the dictionary setting a data value can be directly stored in a memory cell.
Consequently, for each relation, a tuple can be stored in a constant
number of cells.

However, in the general and the ordered setting we do not make any
assumptions on the size or nature of a single data item. Therefore, we
do not assume that a data item can itself be stored in a memory cell
of the PRAM.
Instead, we use token of the form \((R,i,j)\) to represent data values, where \(R\) is a relation symbol, \(i\in\range{\lenof{R}}\) and \(j\in\range{\lenof{\attrof{R}}}\).
More precisely, \((R,i,j)\) represents the value of the \(j\)-th attribute of the \(i\)-th tuple in relation \(R\).
We emphasise that different token \((R_1,i_1,j_1)\) and \((R_2,i_2,j_2)\) can represent the same data value.
Note that a single processor can test in constant time whether two such token represent the same data value using the operation \(\opEqual{R_1,i_1,j_1;R_2,i_2,j_2}\).
In the ordered setting the data values represented can be compared using \(\opLessThan{R_1,i_1,j_1;R_2,i_2,j_2}\).

This encoding extends to tuples in the natural fashion.
Furthermore, such token can be stored in memory cells of a PRAM.
We will refer to this encoding as \emph{token representation} of a data value or tuple.
\begin{exa}\label{example:token-representation}
	Consider the ternary relation \(R\) given by
	\[ R = \big\{(\texttt{"Hello"}, 6000,
          \texttt{blob}_1), (\texttt{"World"}, 3.14,
          \texttt{blob}_2), (\texttt{"Hello"}, 3.14,
          \texttt{blob}_3)\big\}.\]
	The tuples of \(R\) can be represented by the token sequences
	\[((R,1,1),(R,1,2),(R,1,3)), ((R,2,1),(R,2,2),(R,2,3)),
          \textnormal{ and } ((R,3,1),(R,3,2),(R,3,3)).\]
        However, the third tuple could as well be represented by \( ((R,1,1),(R,2,2),(R,3,3))\).

	Of course, even tuples not in \(R\) can be represented by
      token sequences.
	For instance, \(((R,3,1),(R,2,1))\) represents the tuple \((\texttt{"Hello"}, \texttt{"World"})\).
\end{exa}

\paragraph*{Representing databases by arrays.}

As indicated before, we represent relations by one-dimensional arrays,
whose cells contain tuples or their token representations and might
be augmented by additional data.

There are
various possibilities how to deal with ordered relations. One might
store a relation $R$ in an array and, if available, one or more
ordered lists of position numbers (with respect to some attribute
lists) in additional index arrays.

However, in this article, we rather represent a relation by one array,
which might be ordered with respect to some attribute
list. We made this choice to simplify notation and because, in most
results, only one order of a relation is relevant. The choice does not
affect the complexity results. In the framework with index arrays,
algorithms just need to keep relation arrays and index arrays in sync
(by pointers in both directions) and if a relation array is compacted,
its index arrays are compacted, as well.

Due to the impossibility of (perfect) compaction and sorting in
constant time, as discussed in \autoref{section:ctp:lower-bounds}, our algorithms will often yield
result arrays, in which not all cells represent result
tuples.

Therefore, a cell can be \emph{inhabited}, if it represents a
\enquote{useful} tuple, or \emph{uninhabited}, otherwise, indicated by
some flag.
In all three settings, we refer to the tuple represented by an
inhabited cell $\arr{A}[i]$ by $\arr{A}[i].t$, even if this
representation is via token as in the general and the ordered
setting. In the same spirit we say that a tuple occurs in a cell or a
cell contains a tuple when the cell actually contains a token
representation for the tuple. Please note that the terms ``inhabited''
and ``uninhabited''
are used on the level of tuples, whereas the terms ``non-empty'' and
``empty'' are used on the level of arrays. In particular, an inhabited
cell (not carrying a \emph{useful} tuple) can be non-empty.

We note that, in all three settings, all elemental operations can be performed by a single processor
in constant time (and space).
Since the arities of the  relations at hand are always fixed, this also
extends to tuples.

As mentioned before, there might be additional data stored in a cell, e.g., further Boolean flags and links to
other cells of (possibly) other arrays.
However, the number of items is always bounded by a constant,
therefore we still view the representation of a tuple together with
this additional data as a \emph{cell}, since all of it
can be read or written in time $\bigO(1)$.

We say that an array \emph{represents} a
relation $R$, if for each tuple $t$ of $R$, there is some inhabited cell
that contains $t$ and no inhabited cells represent other tuples not in $R$.
This definition allows that a tuple occurs more than once.
An array represents $R$ \emph{concisely}, if each tuple occurs in
exactly one inhabited cell. In this case, we call the array \emph{concise}.

To indicate that an array represents a
relation $R$ we usually denote it as $\relarr{R}$, $\relarr{R}'$, etc. We recall that $|\relarr{R}|$ denotes the size of an array representing $R$, whereas $|R|$ denotes the size of the relation $R$.

We call an array representing a relation \emph{compact} if it has no
uninhabited cells. We similarly adapt the
notion of $\padeps$-compact arrays to arrays representing
relations.

We often consider the  \emph{induced tuple sequence}
$t_1,\ldots,t_{\lenof{\arr{A}}}$ of an array \arr{A}.
 Here, $t_i=\arr{A}[i].t$ is a \emph{proper} tuple, if
 $\arr{A}[i]$ is inhabited, or otherwise $t_i$ is the \emph{empty
  tuple} \emptytuple.

\begin{exa}\label{example:arrays}
	The tuple sequence $[(1,5), \emptytuple, (3,4), (8,3), (1,5),
        \emptytuple, \emptytuple, (7,3)]$ from an array
        $\arr{A}$ of size eight has five \propertuples and three empty tuples.
	It represents the relation $R = \{(1,5), (3,4), (8,3), (7,3)\}$, but \emph{not} concisely.
	The sequence $[(1,5), (3,4), (7,3), \emptytuple, (8,3)]$ represents
        $R$ concisely and ordered with respect to the
        canonical attribute order.
        It is not compact but \(\frac{1}{4}\)-compact.
      \end{exa}

We assume throughout this article that, if a database is represented by
arrays, all these arrays are concise.

In the ordered and dictionary setting an array  $\relarr{R}$ that represents a relation $R$ is
\emph{$\attsetX$-ordered}, for some ordered list $\attsetX$ of attributes of $R$'s
schema, if \(i < j\) implies \(\arr{A}[i].t[\attsetX] \le \arr{A}[j].t[\attsetX]\), for all indices \(i,j\).
We note that this definition allows that tuples that agree on all attributes of $X$ occur
in any order. We call \arr{A} \emph{fully ordered}, if it is
$\attsetX$-ordered, for some $\attsetX$  that contains all
attributes of $R$'s
schema. In this case, the relative order is uniquely determined for all tuples.
We emphasize that, in neither setting, it is guaranteed that ordered arrays for the database relations are given a priori; we discuss this in more detail in \autoref{subsection:ctp:ordered}.%

\section{Algorithms for Basic Array Operations}\label{section:ctp:alg-for-basic-ops}

In this section, we present algorithms for some basic operations on database arrays which we will use as building blocks to implement database operations and the transformations between PRAM settings.
To illustrate some aspects of parallel query processing in the PRAM model  we first consider an example. Afterwards, we introduce some basic operations   in \autoref{section:ctp:operations}. In \autoref{section:ctp:sortcompact}, we give algorithms for padded sorting and compaction, based on the algorithms in \autoref{section:ctp:lower-bounds}.
In \autoref{section:ctp:we-searching} we study work-efficient  search and deduplication algorithms in the dictionary setting and in \autoref{section:ctp:se-searching} we consider space-efficient  search and deduplication algorithms for sorted arrays.

\begin{exa}\label{example:operation-trace-links}
We sketch how to evaluate the projection $\pi_B(R)$ given the array $\arr{R}=[(1,5), (3,4), (7,3), \emptytuple, (8,3)]$ representing \(R\) from \autoref{example:arrays}.

First, the array $\arr{R}'=[5,4,3, \emptytuple, 3]$ is computed and each tuple
$\arr{R}[i]$ is augmented by a link to $\arr{R}'[i]$ and vice versa.
For instance, $t_3 = (7,3)$ is augmented by a link to the third cell of \(\arr{R}'\) containing \(t_3' = 3\) and $t_5 = (8,3)$ by a link to the fifth cell containing $t_5' = 3$ (and vice versa).
We say that the tuples are mutually linked.

To achieve this, the tuples from $\arr{R}$ are loaded to processors $1,\ldots,5$, each processor applies the necessary changes to its tuple (or the token representation thereof), and then writes the new tuple to the new array $\arr{R}'$.
We note that it is not known in advance which cells of $\arr{A}$ are inhabited and therefore, we need to assign one processor for each cell.
Each processor only applies a constant number of steps, so the overall work is in $\bigO(\lenof{\arr{R}})$.

Towards a concise representation of $\pi_B(R)$, a second operation eliminates duplicates.
To this end, it creates a copy $\arr{R}''$ of $\arr{R}'$ and checks in parallel, with one processor per pair $(i,j)$ of indices with $i < j$, whether $\arr{R}''[i].t = \arr{R}''[j].t$ holds.
If $\arr{R}''[i].t = \arr{R}''[j].t$ holds, then cell $\arr{R}''[j]$ is flagged as uninhabited.
Lastly, the algorithm creates links from every cell in $\arr{R}'$ to the unique inhabited cell in $\arr{R}''$ holding the same tuple.
To this end, each processor for a pair $(i,j)$ with  $\arr{R}'[i].t = \arr{R}''[j].t$, for which $\arr{R}''[i]$ is still inhabited, augments the cell $\arr{R}'[j]$ with a link to $\arr{R}''[i]$. Furthermore, a link from $\arr{R}''[i]$ to \(\arr{R}[i]\) is added.
Note that multiple processors might attempt to augment a cell $\arr{R}''[i]$ with a link to $\arr{R}[i]$; but this is fine, even in the Common CRCW model.
Overall, we get two-step-links from $\arr{R}$ to $\arr{R}''$ and immediate links from $\arr{R}''$ to \enquote{representatives} in $\arr{R}$.

The second operation has a work bound of $\bigO(\lenof{\arr{R}}^2)$ because it suffices to assign one processor for each pair $(i,j)$ of indices and each processor only applies a constant number of steps.

Observe that the algorithms we sketched above can be realised even in the most general setting, since all steps can be implemented with \(\opnameEqual\) and by manipulating token representations.
\qed
\end{exa}
As in this example, algorithms for query evaluation will use basic operations. In fact, besides deduplication, also  compaction, (padded) sorting and search will be used.

We will see below that the work bound for deduplication can be considerably improved in the dictionary setting or if the array is ordered, thereby also improving the bound for
the ordered setting.

\subsection{Basic Operations}\label{section:ctp:operations}
We now describe the basic array operations which we will use as building
blocks for query evaluation algorithms
for PRAMs.

Just as in \autoref{example:operation-trace-links}, the operations
are usually applied to arrays representing relations, and produce an array as output. Furthermore, in all participating arrays, tuples might be augmented by additional data, in particular, by  links to other tuples. In fact, each time a tuple of an output
array results from some tuple of an input array links are added.

We emphasise that the input (and output) arrays of the following operations do \emph{not} necessarily have to represent relations concisely, some operations are even only meaningful if they do not.

\begin{itemize}

	\item $\opCompact{\arr{A}}$ copies the
          \propertuples in $\arr{A}$ into distinct cells of an array
          $\arr{B}$ of size at most $(1+\padeps)k$,
          where \(k\) is the number of \propertuples in \(\arr{A}\) and \(\padeps>0\) a constant.
          Mutual links are added between each inhabited cell \(\arr{A}[i]\) and the cell of \(\arr{B}\) the contents of \(\arr{A}[i]\) got copied to.
          Additionally, the operation preserves the relative order of tuples.

	\item $\opSort{\arr{A},\attsetX}$ returns an \(\attsetX\)-ordered array $\arr{B}$ of size \((1+\padeps)\lenof{\arr{A}}\) which contains all \propertuples of \(\arr{A}\).
	Here \(\attsetX\) is an ordered list of attributes from the relation represented by~\(\arr{A}\).
          Mutual links are added between each inhabited cell of \(\arr{A}\) and the corresponding cell of~\(\arr{B}\).
	\item $\opSearchTuples{\arr{A},\arr{B}}$ links every inhabited cell
          $\arr{A}[i]$ to an inhabited cell $\arr{B}[j]$,
          such that $\arr{B}[j].t = \arr{A}[i].t$ holds, if such a cell
          exists. We sometimes write \opSearchTuple{t,\arr{B}} in the special case that only one tuple $t$ is searched for.
	\item $\opDeduplicate{\arr{A}}$ chooses one representative tuple from
		each set of identical tuples, deactivates the other tuples and adds a link from each deactivated other tuple to the representative tuple for the set.
\end{itemize}

In the remainder of this section, we give \ctpalgos for these operations.
\begin{itemize}
\item For $\opnameCompact$ we give an algorithm that even works in the most general setting and relies on \autoref{result:lin-compaction}.
\item For $\opnameSort$ we give an algorithm for the dictionary setting.
\item For the other three operations, we give work-efficient algorithms in the dictionary setting and space-efficient algorithms that assumes that (at least) one array is ordered in an appropriate way.
\end{itemize}

   \subsection{Algorithms for Compaction and Padded Sorting}\label{section:ctp:sortcompact}

In this subsection, we present algorithms for compaction and sorting.
They are both easily inherited from their counterparts in \autoref{section:ctp:lower-bounds}.

The following result for $\opnameCompact$ basically follows from \autoref{result:lin-compaction}.

\begin{lemC}[{\cite[implied by Theorem~4.2]{GoldbergZ95}}]\label{result:alg-compact}
	For every constant $\padeps > 0$ and $\varepsilon > 0$, there is a \ctpalgo for $\opnameCompact$ that, given an array $\arr{A}$, requires $\bigO(\lenof{\arr{A}}^{1+\varepsilon})$ work and space on a Common CRCW PRAM.
\end{lemC}
We note that the result holds even for the most general setting.
\begin{proof}
The lemma follows from \autoref{result:lin-compaction} by observing that for any constant \(\padeps > 0\) there is an integer \(a\) such that \(\padeps < (\log n)^{-a}\) holds for all \(n > 2\).
The mutual links can be established as follows: Instead of compacting \(\arr{A}\) directly, the algorithm allocates an array \(\arr{A}'\) of the same size as \(\arr{A}\) and sets \(\arr{A}'[i] = i\) if \(\arr{A}\) is inhabited, and leaves \(\arr{A}'[i]\) empty otherwise.
Then \(\arr{A}'\) is compacted using \autoref{result:lin-compaction}.
Let \(\arr{B}'\) be the resulting array.
The desired output array \(\arr{B}\) can then be compiled as follows: Processor \(j\) checks whether \(\arr{B}'[j]\) contains an index \(i\).
If yes, it sets \(\arr{B}[j] = \arr{A}[i]\) and establishes mutual links between \(\arr{B}[j]\) and \(\arr{A}[i]\).
This procedure requires only work and space \(\bigO(\lenof{\arr{A}})\) in addition to the compaction.
The bounds are therefore inherited from \autoref{result:lin-compaction}. The length of \arr{B} can be determined by adding an  element to the end of \arr{A} and by following its final link to \arr{B}.
\end{proof}

The following result for $\opnameSort$ is easily obtained from \autoref{result:sorting}.

\begin{lem}\label{result:alg-sort}
	For all constants $\padeps > 0$ and $k$, there is a \ctpalgo for $\opnameSort$ that, given an array $\arr{A}$ and a  list \(\attsetX\) of $k$ attributes, requires $\bigO((\lenof{\arr{A}}+\lenof{\db})^{1+\varepsilon})$ work and space on a Common CRCW PRAM in the dictionary setting.
\end{lem}
\begin{proof}
	Let \(t_1,\ldots, t_n\) with \(n = \lenof{\arr{A}}\) be the tuple sequence given by \(\arr{A}\).
	We assume that \(\attsetX = (\attA_1,\ldots,\attA_k)\) contains all attributes occurring in \(\arr{A}\).
	Otherwise, \(\attsetX\) can be extended arbitrarily.
	In the following we denote by \(\le_{\attsetX}\) the lexicographic order induced by \(\attsetX\) on the tuples occurring in \(\arr{A}\).

	For each \propertuple \(t_i\), the algorithm computes its characteristic number \[c_i = \sum_{j = 1}^{k} t_i[\attA_j] m^{k-j}\] where \(m = \cval\lenof{\db} + 1\).
	Observe that, since the database has small values, \(m\) is larger than any value in \(\arr{A}\).
	Thus,  we have \(t_i \le_{\attsetX} t_j\) if and only if  \(c_i \le c_j\) for all \(i,j\in\range{n}\).
        	Furthermore, a processor can derive  \(t_i\) from \(c_i\) in constant time.

	Thanks to \autoref{result:sorting} an array of size \((1+\padeps)\max\{\lenof{\arr{A}},\lenof{D}\}\) representing a sequence $i_1,\ldots,i_n$, can be computed in constant time with $\bigO((\lenof{\arr{A}}+\lenof{\db})^{1+\varepsilon})$ work and space, such that $i_j<i_\ell$ implies $c_{i_j}\le c_{i_\ell}$. The result array \arr{B} can then be obtained by storing each tuple $t_{i_j}$ in cell $\arr{B}[i_j]$.
	If necessary, \(\arr{B}\) can be compacted using \(\opnameCompact[\padeps']\) for a sufficiently small \(\padeps'\) to yield an array of the desired size.

	The computation of mutual links is straightforward.
      \end{proof}

      \subsection{Work-Efficient Algorithms for Search and Deduplication}\label{section:ctp:we-searching}

      In the dictionary setting, both search and deduplication can be implemented in a work-efficient fashion based on hashing. The dictionary setting in itself can be seen as a form of hashing, more precisely, as mapping (arbitrary) data values injectively to small natural numbers $\bigO(\lenof{\db})$. Implicitly, these mappings induce  mappings of tuples of arity $k$ to numbers of size $\bigO(\lenof{\db}^k)$. However, this does not suffice
      to obtain  work-efficient algorithms for  search and deduplication operations. Instead, we lift  the basic hashing mappings to the level of tuples in a way that yields natural numbers of size $\bigO(\lenof{\arr{A}})$ for tuples from a given array $\arr{A}$. More formally,
      we call a mapping that maps each inhabited cell of an array $\arr{A}$ to a number from
$\range{\lenof{\arr{A}}}$, such that  $\arr{A}[i]$ and  $\arr{A}[j]$ get the same number if and only if $\arr{A}[i].t = \arr{A}[j].t$ holds, an
      \emph{array hash table}.
Array hash tables can be efficiently computed, as the following lemma states.

\begin{lem}\label{result:alg-compute-hash-values}
For every $\varepsilon>0$, there is a \ctpalgo that, in the dictionary setting, computes an array hash table for a given array \arr{A}, and requires $\bigO(\lenof{\arr{A}})$ work and $\bigO(\lenof{\arr{A}}\cdot \lenof{\db}^\varepsilon)$ space on an Arbitrary CRCW PRAM.
\end{lem}
We note that the array hash table computed by the algorithm underlying \autoref{result:alg-compute-hash-values} is not uniquely determined by \arr{A}, but can depend on  the \enquote{arbitrary} resolution of concurrent write operations.

\begin{proof}[Proof sketch]
	We describe the construction for $\varepsilon= 1$.

	Let $\attA_1,\ldots,\attA_\ell$ be the attributes of the relation $R$ represented by $\arr{A}$ in an arbitrary but fixed order and define $\attsetX_j = (\attA_1,\ldots,\attA_j)$ for all $j\in\range{\ell}$.
	The algorithm inductively computes hash values for tuples in $R[\attsetX_j]$ for increasing $j$ from $1$ to $\ell$. It uses \(\lenof{\arr{A}}\) processors.

	The idea is to assign, to each tuple $t\in R[\attsetX_j]$, a number in the range $\range{\lenof{\arr{A}}}$ as hash value and to augment each cell of $\arr{A}$ containing a \propertuple $t'$ with $t'[\attsetX] = t$ by this hash value.
	The challenge is that the same (projected) tuple $t \in R[\attsetX_j]$ might occur in multiple cells of $\arr{A}$.

	Recall that, in the dictionary setting, each value in the active domain is a number of size at most $\cval \lenof{\db}$.

        For the base case $j = 1$ the algorithm allocates an auxiliary array \arr{B} of size $\cval\lenof{\db}$.  For each $i\le \lenof{\arr{A}}$,  let $t_i=\arr{A}[i].t$.
	Each processor $i$ with a \propertuple writes its processor number $i$ into cell $t_i[\attA_1]$ of \arr{B}  and then augments  $t_i$ by the value actually written at position $t_i[\attA_1]$.
	Note that, for each value $a$, all processors $i$ with $t_i[\attA_1] = a$ will assign the same value to their tuple $t_i$, since precisely one processor among the processors with $t_i[\attA_1] = a$ succeeds in writing its  number to cell $t_i[\attA_1]$ on an Arbitrary CRCW PRAM.
	This needs $\bigO(\lenof{\arr{A}})$ work and $\bigO(\lenof{\arr{A}} + \lenof{\db})$ space.

	For $j > 1$ the algorithm proceeds similarly but also takes, for a tuple $t$, the previous hash value $h_{j-1}(t[\attsetX_{j-1}])$ for $t[\attsetX_{j-1}]$ into account, in addition to $t[\attA_j]$.
	It allocates an auxiliary array \arr{B} of size $\lenof{\arr{A}}\cdot \cval\lenof{\db}$ which is interpreted as two-dimensional array and each processor  writes its number $i$ into cell $(h_{j-1}(t_i[\attsetX_{j-1}]), t_i[\attsetX_j])$ of \arr{B}, if $t_i$ is a \propertuple.
	The number written into this cell is then the hash value for $t_i[\attsetX_j]$, as in the base case.

	Writing and reading back the processor numbers requires $\bigO(\lenof{\arr{A}})$ work and \arr{B}  requires $\bigO(\lenof{\arr{A}}\cdot \lenof{\db})$ space.
	The same bounds hold for the recursive invocations of \opnameComputeHashvalues.
	Since the recursion depth is $\ell$, the procedure requires $\bigO(\ell\cdot \lenof{\arr{A}}) = \bigO(\lenof{\arr{A}})$ work and, because the space for \arr{B}  can be reused, $\bigO(\lenof{\arr{A}}\cdot \lenof{\db})$ space in total.
        This concludes the proof for the case $\varepsilon=1$ and, of course, the same construction works for $\varepsilon>1$.

         For $\varepsilon< 1$, each attribute can be replaced  by a list of \(\lceil\frac{1}{\varepsilon}\rceil\) attributes with values of size \(\bigO(\lenof{\db}^\varepsilon)\). The factor $\lenof{\db}$ in the space bound above can thus be replaced by $\lenof{\db}^\varepsilon$, as in the statement of the lemma. The asymptotic overall work is not affected by this adaptation, since it only adds a factor of \(\lceil\frac{1}{\varepsilon}\rceil\).
      \end{proof}

\begin{lem}\label{result:alg-search-dictionary}
  For every $\varepsilon>0$, there is a \ctpalgo for $\opnameSearchTuples$ that, given arrays \arr{A} and  \arr{B} requires  \wsbounds%
    	{\lenof{\arr{A}} + \lenof{\arr{B}}}%
    	{(\lenof{\arr{A}} + \lenof{\arr{B}})\cdot \lenof{\db}^\varepsilon} in the dictionary setting.
\end{lem}

\begin{proof}[Proof sketch]
First an array hash table for (the concatenation of) $\arr{A}$ and $\arr{B}$ is computed with $\bigO(\lenof{\arr{A}} + \lenof{\arr{B}})$ work and $\bigO((\lenof{\arr{A}} + \lenof{\arr{B}})\cdot \lenof{\db}^\varepsilon)$ space, similarly as in
	\autoref{result:alg-compute-hash-values}.
	For a \propertuple $t$, let $h(t)$ denote the hash value in the range $\range{\lenof{\arr{A}} + \lenof{\arr{B}}}$ assigned to $t$.
	The algorithm then allocates an auxiliary array \arr{C} of size $\lenof{\arr{A}} + \lenof{\arr{B}}$ and, for each \propertuple $s_i$ in $\arr{B}$, it writes, in parallel, $i$ into cell $h(s_i)$ of \arr{C}. Here $s_i$ denotes the $i$-th tuple from~\(\arr{B}\).  Other processors might attempt to write an index to cell $h(s_i)$, but only one will succeed.
	This requires $\bigO(\lenof{\arr{B}})$ work to write the indices and $\bigO(\lenof{\arr{A}} + \lenof{\arr{B}})$ space for \arr{C}.

	For each \propertuple $t$ in $\arr{A}$ it is then checked in parallel, if cell $h(t)$ contains an index~$i$.
	If yes, then $t$ is marked and augmented with a pointer to cell $\arr{B}[i]$, since $\arr{B}[i].t = t$.
	If not, then $t$ has no partner tuple in $\arr{B}$, thus $t$ is not augmented by a link.
      \end{proof}
      We note that in the special case that \arr{A} consists of a single tuple $t$, the algorithm is worse than the naive algorithm that just compares $t$ with each tuple of  \arr{B} since that also requires work $\bigO(\lenof{\arr{B}})$ but does not need any additional space.

      Array hash tables also yield a work-efficient algorithm for deduplication in a straightforward way.
      \begin{lem}\label{result:alg-deduplicate-dictionary}
    For every $\varepsilon>0$, there is a \ctpalgo for \opnameDeduplicate in the dictionary setting that, on an Arbitrary CRCW PRAM, requires \wsbounds
 			{\lenof{\arr{A}}}%
			{\lenof{\arr{A}}\cdot \lenof{\db}^\varepsilon}.
\end{lem}
\begin{proof}[Proof sketch]
 The algorithm first computes an array hash table for $\arr{A}$ with $\bigO(\lenof{\arr{A}})$ work and $\bigO(\lenof{\arr{A}}\cdot \lenof{\db}^\varepsilon)$ space, thanks to
	\autoref{result:alg-compute-hash-values}.
	For a \propertuple $t$, let $h(t)$ denote the hash value in the range $\range{\lenof{\arr{A}}}$ assigned to $t$.
	The algorithm then allocates an auxiliary array \arr{B} of size $\lenof{\arr{A}}$ and, for each \propertuple $t_i$ in $\arr{A}$, it writes, in parallel, $i$ into cell $h(t_i)$ of \arr{B}. Here $t_i$ denotes the $i$-th tuple from \(\arr{A}\).  Several processors might attempt to write an index into the same cell $h(t_i)$, but only one will succeed.
	This requires $\bigO(\lenof{\arr{A}})$ work to write the indices and $\bigO(\lenof{\arr{A}})$ space for \arr{B}.

	For each \propertuple $t_i$ in $\arr{A}$ it is the index $j$ stored in cell $h(t_i)$ is then compared, in parallel, with $i$. If $i=j$ then $t_i$ is a representative and remains active. Otherwise, $t_i$ is marked as a duplicate and a pointer to $t_j$ is added.
\end{proof}

      \subsection{Order-Based Algorithms for Search and Deduplication }\label{section:ctp:se-searching}

Even though the algorithms of \autoref{section:ctp:we-searching} are work-optimal, they require some additional space, e.g., $\bigO(\lenof{\arr{A}}\cdot \lenof{\db}^\varepsilon)$ in \autoref{result:alg-deduplicate-dictionary}. In this subsection, we consider algorithms that are slightly more space-efficient whenever \(\lenof{\arr{A}} < \lenof{\db}\). They require that data values are equipped with a linear order, but do not rely on the dictionary setting. However,  they are not work-optimal and require  that at least one of the arrays is ordered and has some additional structure.

In sequential database processing, indexes implemented by search trees make it possible to efficiently test whether a given tuple is in a given relation and to find the largest ``smaller'' tuple, if it is not. In a concise ordered array, this is possible in constant time with work $\bigO(\lenof{\arr{A}}^{\varepsilon})$ by ``\(\lenof{\arr{A}}^{\varepsilon}\)-ary search''. That is, in a first round a sub-array of length \(\lenof{\arr{A}}^{1-\varepsilon}\) containing the tuple (if it is there) is identified by inspecting \(\lenof{\arr{A}}^{\varepsilon}\) cells, and so on. However, in arrays with uninhabited cells, this might fail, since these are not helpful for this process. To this end, the algorithms in this subsection use
predecessor and successor links, linking each (inhabited or uninhabited) cell to the next smaller or larger  \emph{inhabited cell}, respectively. We say an array is
\emph{fully linked} if such links are available. By \autoref{result:predsucc} we know that predecessor and successor links can be computed in constant time with work and space \(\bigO(\lenof{\arr{A}}^{1+\varepsilon})\).

     For fully ordered arrays,  deduplication is very easy.

      \begin{lem}\label{result:alg-deduplicate-ordered}
        For every $\varepsilon>0$, there is a \ctpalgo   on an Arbitrary CRCW PRAM for \opnameDeduplicate that, given a fully ordered array \arr{A} requires work and space $\bigO(\lenof{\arr{A}}^{1+\varepsilon})$.
\end{lem}
\begin{proof}[Proof sketch]
  The algorithm basically uses the algorithm of \autoref{result:predsucc} twice.
  During the first time, it computes  predecessor and successor links just as in \autoref{result:predsucc}.
  Then, for every proper tuple $t$, $t$ is marked as a representative, if the predeccessor tuple of $t$ is different from $t$. Then the algorithm of \autoref{result:predsucc} is applied a second time but
  this time cells whose tuple are not representatives are treated like uninhabited cells resulting in the computation of a link to the next representative (of the same tuple) and from representatives to the previous representative (of a smaller tuple).%
\end{proof}

In fully linked and ordered arrays, it is possible to search efficiently for tuples.

\begin{prop}\label{result:alg-search-fullylinkedB}
  For every $\varepsilon>0$, there is a  \ctpalgo on a CREW PRAM for
$\opnameSearchTuples$ that, given  arrays \arr{A} and  \arr{B} requires work and space
$\bigO(\lenof{\arr{A}}\cdot\lenof{\arr{B}}^{\varepsilon})$
if the array \(\arr{B}\) is fully linked and fully ordered. In fact, if the algorithm computes, for each tuple $t$ of \arr{A}  the largest index \(i\) with $\arr{B}[i].t[\attsetX]\le t$ (or the smallest index \(j\) with \(\arr{B}[j].t[\attsetX] \ge t\)). Here, $\attsetX$ denotes the list of attributes, such that \(\arr{B}\) is $\attsetX$-ordered.
\end{prop}
      If  \arr{B} is fully ordered, but not fully linked, the combination with \autoref{result:predsucc} yields an upper bound of $\bigO((\lenof{\arr{A}}+\lenof{\arr{B}})\cdot\lenof{\arr{B}}^{\varepsilon})$.
\begin{proof}[Proof sketch]
  Let $n=\lenof{\arr{B}}$. The algorithm performs the following search, for each proper tuple $t$ of \arr{A} in parallel.
  In the first round, using $n^\varepsilon$ processors, the algorithm aims to identify a number $j$ with $j n^{1-\varepsilon} \le i < (j+1) n^{1-\varepsilon}$ for the largest index \(i\) with $\arr{B}[i].t[\attsetX]\le t$. To this end, it tests, for all cells with positions of a form $k=\ell\cdot n^{1-\varepsilon}$ for \(0\le \ell < n^\varepsilon\) whether $\arr{B}[k]$ (or its successor if  $\arr{B}[k]$ is not inhabited) contains a tuple \(t'\) such that \(t'[X] \le t\) holds and whether this does not hold for position $(\ell+1)n^{1-\varepsilon}$ or its successor. The search continues recursively in the thus identified sub-interval. After  $\lceil\frac{1}{\varepsilon}\rceil$ rounds it terminates. Since, in each round, $n^\varepsilon$ processors are used (per tuple $t$), the statement follows.
\end{proof}
We note that, unlike for \autoref{result:alg-search-dictionary}, the algorithm is also quite efficient for the single tuple case ($\lenof{\arr{A}}=1$).

As an alternative to fully-linked arrays, bounded-depth search trees with degree about $n^\varepsilon$ could be used as index structures. They can be defined in the obvious way. The work for a search is then asymptotically the same as for fully linked ordered arrays.

Interestingly, there is a dual algorithm with almost dual bounds in the case that \arr{A} is fully ordered, instead of \arr{B}.

\begin{prop}\label{result:alg-search-fullylinkedA}
  For every $\varepsilon>0$, there is a  \ctpalgo on an Arbitrary CREW PRAM for
  $\opnameSearchTuples$ that, given  arrays \arr{A} and  \arr{B} requires work and space $\bigO((\lenof{\arr{A}}+\lenof{\arr{B}})\cdot\lenof{\arr{A}}^{\varepsilon})$,
if the array \(\arr{A}\) is fully ordered.
\end{prop}
However, unlike the algorithm of \autoref{result:alg-search-fullylinkedB}, this algorithm does not guarantee to identify  the largest index \(i\) with $\arr{B}[i].t[\attsetX]\le t$ (or the smallest index \(j\) with \(\arr{B}[j].t[\attsetX] \ge t\)).
\begin{proof}[Proof sketch]
  First, the algorithm computes a deduplicated version \arr{A'} of  \arr{A} and, in particular, computes links from all tuples of \arr{A} to their representatives in \arr{A'}.
   Then, it searches, for each \propertuple $s$ in $\arr{B}$, the smallest tuple $t$ in \arr{A'} with $t\ge s$.
	If $t=s$ then the cell of $t$ is marked and a link to the cell of $s$ is added.
	We note that if \(\arr{B}\) is not concise, multiple processors might attempt to write a link for a cell of a \propertuple \(t\) in \(\arr{A'}\), but only one of them succeeds.
Finally, for each \propertuple $t$ in $\arr{A}$, if its representative in \arr{A'} has obtained a link to a tuple in \arr{B}, this link is added to $t$.
\end{proof}

\section{Database Operations}\label{section:ctp:alg-for-db-ops}

In this section, we proceed similarly to \autoref{section:ctp:alg-for-basic-ops}.
We start by defining array-based operations for the operators of the relational algebra.
The main contribution of this section is the presentation of \ctpalgos for these operations and the analysis of their complexity with respect to work and space.
Before we do that in \autoref{section:ctp:dbops-upper-bounds}, we will first present some lower bounds in \autoref{section:ctp:dbops-lower-bounds} that are consequences of the impossibility results of \autoref{section:ctp:lower-bounds}.

\paragraph*{Array-based database operations.}
For database operations, unlike for the basic operations, we require that input relations are given by arrays representing them concisely. That is, each tuple occurs exactly once in the array.
Likewise, all algorithms produce output arrays which represent the result relation concisely.
However, neither for input nor for output relations we make any assumptions about the compactness of the representations, i.e.,   there can be uninhabited cells.

More precisely, we  actually present \ctpalgos for the following operations (on arrays) which correspond to the operators of the relational algebra (on relations).\footnote{We believe that distinguishing between operations on arrays and operations on relations also improves the readability of (pseudo-code) algorithms and proofs.}

Recall that we denote arrays representing relations \(R\),\(S\), etc.\ by \(\relarr{R}\), \(\relarr{S}\), etc.
\begin{itemize}
	\item ${\opSelection[\attA=x]{\relarr{R}}}$ returns a new array \(\relarr{R}'\) that represents \sel{R}{\attA}{x} concisely. Here \(\attA\) is an attribute of \(R\) and \(x\) is either an attribute of \(R\) or a domain value.

	\item ${\opProjection[\attsetX]{\relarr{R}}}$] returns a new array \(\relarr{R}'\) that represents \proj{R}{\attsetX} concisely. Here \(\attsetX\) is a set or sequence of attributes from \(R\).

	\item $\opDifference{\relarr{R},\relarr{S}}$ returns a new array that represents \(R\setminus S\) concisely.
	Here \(R\) and \(S\) are relations over the same set of attributes.

	\item $\opUnion{\relarr{R},\relarr{S}}$ returns a new array that represents \(R\cup S\) concisely.
	Here \(R\) and \(S\) are relations over the same set of attributes.

	\item $\opSemiJoin{\relarr{R},\relarr{S}}$ returns a new array \(\relarr{R}'\) that represents \(\sjoin{R}{S}\) concisely.

	\item $\opJoin{\relarr{R},\relarr{S}}$ returns a new array that represents \(R\Join S\) concisely.
\end{itemize}

We note that the relational algebra has an additional rename operator,
which, of course, does not require a parallel algorithm.

\subsection{Lower Bounds}\label{section:ctp:dbops-lower-bounds}
In this subsection, we show two lower bound results for algorithms that compute the semijoin of two unary relations \(R\)
  and \(S\):
\begin{itemize}
\item  there is no \ctpalgo using a polynomial number of
  processors, that stores the result in an array of size
  \(\lenof{\sjoin{R}{S}}\);
\item any algorithm that stores the result   in an array of size $(1+\padeps)\lenof{\sjoin{R}{S}}$, for each input, requires $n^{1+\varepsilon}$ processors.
\end{itemize}

These lower bounds rely on the lower bounds for the exact compaction problem  (\autoref{result:sort-lower}\ref{result:sort-lower-compaction}) and the $\lambda$-approximate compaction problem (\autoref{coro:compact-delta}), respectively. Note, that the lower bounds even hold if the input arrays are concise, compact and fully ordered.
Our lower bounds are stated in the following theorem.

\begin{thm}\label{result:dbops-lowerbound}
	Let $\padeps > 0$ be a constant. For every CRCW PRAM algorithm for the operation \(\opnameSemiJoin\) the following statements hold.
	\begin{enumerate}[(a)]
		\item If the algorithm uses a polynomial number of processors and produces compact output arrays, for every possible input, then it has running time $\Omega\left( \frac{\log n}{\log \log n} \right)$.
		\label{result:dbops-lowerbound-exact}
		\item If the algorithm only needs constant time and produces \(\padeps\)-compact output arrays, for every possible input, then it needs at least $\Omega(n^{1+\varepsilon})$ processors, for some constant $\varepsilon > 0$. %
		\label{result:dbops-lowerbound-general-sjoin}
	\end{enumerate}
	Here \(n\) is the size of the input arrays.
	The lower bounds even hold if the input relations \(R\) and \(S\) are unary and represented by concise, compact and fully ordered arrays $\relarr{R}$ and $\relarr{S}$, and the database has small values.
\end{thm}

\begin{proof}
  Both results basically use the same reduction from the respective compaction problem.

	We first prove the second part. Assume for contradiction that there is a CRCW PRAM algorithm for \(\opnameSemiJoin\)  that guarantees a \(\padeps\)-compact output array and runs in constant time using $\mu(n)$ processors, where $\mu(n)\ge n$ and $\mu(n)=o(n^{1+\varepsilon})$, for every $\varepsilon>0$. We show how this algorithm can be used to solve the approximate compaction problem with $\mu(n)$ processors in constant time.

	Let $\arr{B}$ be an input array of size $n$ with \(k\) non-empty cells for the \(\padeps\)-approximate compaction problem.

	In the first step create an array $\relarr{R}$ of size $n$ for the unary relation $R$ that consists of the even numbers from $2$ to $2n$.
	This can be done in parallel in constant time with $n$ processors.

	Let $\relarr{S}$ be an array of size $n$.
	In the second step, store, for every $i \in \range{n}$, the value $2i$ in $\relarr{S}[i]$ if $\arr{B}[i]$ is not empty, and otherwise store the value $2i+1$ in $\relarr{S}[i]$.
	Hence, the query result of $q$ exactly consists of all even numbers $2i$, for which $\arr{B}[i]$ is not empty.

	Note, that by construction, the arrays for the relations $R$ and $S$ are concise, compact, and fully ordered.

	From the array for the query result a solution for the linear approximate compaction problem can be obtained by replacing every value $2i$ in the result by $\arr{B}[i]$. The size of the compact array is at most $(1+\padeps)\lenof{\sjoin{R}{S}} = (1+\padeps)\lenof{\join{R}{S}} = (1+\padeps)k$.

	All in all, the algorithm takes constant time and uses $\mu(n)$ processors to solve the linear approximate compaction problem. Hence, by \autoref{prop:lowerbound:lac} it holds $\mu(n)=\Omega(n^{1+\varepsilon})$, for some $\varepsilon>0$.
	By the same construction, an algorithm for \(\opnameSemiJoin\) (or \(\opnameJoin\)) according to (a) yields an exact compaction algorithm that works in constant time with polynomially many processors, contradicting \autoref{result:sort-lower}\ref{result:sort-lower-compaction}.
      \end{proof}

      Of course, the result also holds for the join of two unary relations.

      We emphasise that \autoref{result:dbops-lowerbound} (b) really relies on the condition that the output array is \(\padeps\)-compact, since we will see below in \autoref{alg:semijoin} that  \(\opnameSemiJoin\)  can be computed with linear work without this condition, in the dictionary setting.

\subsection{Algorithms and Upper Bounds}\label{section:ctp:dbops-upper-bounds}

With the notable exception of the \opnameJoin operation, for most operations
our algorithms are simple combinations of the algorithms of
\autoref{section:ctp:alg-for-basic-ops}.
These simple algorithms are presented in \autoref{section:ctp:semijoin-operations}.

The algorithms for the \opnameJoin operation are slightly more involved and are presented in \autoref{section:ctp:join-algorithms}.

We present algorithms for the dictionary setting and algorithms that do not rely on the dictionary setting but assume that some input array  is suitably ordered (and possibly fully linked).

\subsubsection{Algorithms for the Operations of the Semijoin Algebra}
\label{section:ctp:semijoin-operations}

For the operators of the semijoin algebra, i.e.\ all except the \(\opnameJoin\)-operation, we will present \enquote{high-level} algorithms built upon the basic array operations introduced in \autoref{section:ctp:alg-for-basic-ops}.
Plugging in the concrete algorithms for the basic operation then yields different algorithms with different assumptions and (slightly) different complexities. More precisely, there is always an algorithm for the dictionary setting that has no further assumptions and there are algorithms that require that some arrays are suitably ordered.

An overview of the work and space bounds obtained by the algorithms for the dictionary setting is given in \autoref{table:complexity-bounds-semijoin-dictionary}. The bounds for the other algorithms are given in \autoref{table:complexity-bounds-semijoin-ordered}, together with the underlying assumptions.

\begin{table}
	\caption{Complexity bounds for the operations of the semijoin algebra in the dictionary setting}
\centering
	\begin{tabular}{l l l l}
		\toprule
		Operation & Result & Work bound & Space bound\\
		\midrule
		\(\opSelection[\attA=x]{\relarr{R}}\) & \autoref{alg:selection} & \(\bigO(\lenof{\relarr{R}})\) & \(\bigO(\lenof{\relarr{R}})\)\\
		\(\opProjection[\attsetX]{\relarr{R}}\) & \autoref{alg:project} & \(\bigO(\lenof{\relarr{R}})\) & \(\bigO(\lenof{\relarr{R}}\cdot\lenof{\db}^{\varepsilon})\)\\
		\(\opSemiJoin{\relarr{R},\relarr{S}}\) & \autoref{alg:semijoin} & \(\bigO(\lenof{\relarr{R}}+\lenof{\relarr{S}})\) & \(\bigO((\lenof{\relarr{R}} + \lenof{\relarr{S}})\cdot\lenof{\db}^{\varepsilon})\) \\
		\(\opDifference{\relarr{R},\relarr{S}}\) & \autoref{alg:difference} & \(\bigO(\lenof{\relarr{R}}+\lenof{\relarr{S}})\) & \(\bigO((\lenof{\relarr{R}} + \lenof{\relarr{S}})\cdot\lenof{\db}^{\varepsilon})\) \\
		\(\opUnion{\relarr{R},\relarr{S}}\) & \autoref{alg:union} & \(\bigO(\lenof{\relarr{R}}+\lenof{\relarr{S}})\) & \(\bigO((\lenof{\relarr{R}} + \lenof{\relarr{S}})\cdot\lenof{\db}^{\varepsilon})\) \\
		\bottomrule
	\end{tabular}
	\label{table:complexity-bounds-semijoin-dictionary}
\end{table}

\begin{table}
	\caption{Complexity bounds for the operations of the semijoin algebra in the ordered (and dictionary) setting.
	In the conditions for the \opnameSemiJoin operation, \(\attsetX\) is the set of common attributes of \(R\) and \(S\).}
\centering
	\begin{tabular}{l l l p{3.5cm}}
		\toprule
		Operation & Result & Work and space bound & Conditions\\
		\midrule
		\(\opSelection[\attA=x]{\relarr{R}}\) & \autoref{alg:selection} & \(\bigO(\lenof{\relarr{R}})\) & -- \\\addlinespace
		\(\opProjection[\attsetX]{\relarr{R}}\) & \autoref{alg:project} & \(\bigO(\lenof{\relarr{R}}^{1+\varepsilon})\) & if \(\relarr{R}\) is \(\attsetX\)-ordered\newline\phantom{if \(\relarr{R}\)} and fully linked\\\addlinespace
		\(\opSemiJoin{\relarr{R},\relarr{S}}\) & \autoref{alg:semijoin} & \(\bigO(\lenof{\relarr{R}}\cdot\lenof{\relarr{S}}^\varepsilon)\) & if \(\relarr{S}\) is \(\attsetX\)-ordered\newline\phantom{if \(\relarr{S}\)} and fully linked\\
		&&\(\bigO((\lenof{\relarr{R}}+\lenof{\relarr{S}})\cdot\lenof{\relarr{R}}^\varepsilon)\) & if \(\relarr{R}\) is \(\attsetX\)-ordered\newline\phantom{if \(\relarr{R}\)} and fully linked\\\addlinespace
		\(\opDifference{\relarr{R},\relarr{S}}\) & \autoref{alg:difference} & \(\bigO(\lenof{\relarr{R}}\cdot\lenof{\relarr{S}}^\varepsilon)\) & if \(\relarr{S}\) is fully ordered\newline\phantom{if \(\relarr{S}\)} and fully linked\\
		&&\(\bigO((\lenof{\relarr{R}}+\lenof{\relarr{S}})\cdot\lenof{\relarr{R}}^\varepsilon)\) & if \(\relarr{R}\) is fully ordered\newline\phantom{if \(\relarr{R}\)} and fully linked\\\addlinespace
		\(\opUnion{\relarr{R},\relarr{S}}\) & \autoref{alg:union} & \(\bigO(\lenof{\relarr{R}}\cdot\lenof{\relarr{S}}^\varepsilon+\lenof{\relarr{S}})\) & if \(\relarr{S}\) is fully ordered\newline\phantom{if \(\relarr{S}\)} and fully linked\\
		\bottomrule
	\end{tabular}
	\label{table:complexity-bounds-semijoin-ordered}
\end{table}

\begin{prop}\label{alg:selection}
	For every \(x\) that is either a domain value or an attribute there is a \ctpalgo for \opnameSelection[\attA=x] that, given an array \(\relarr{R}\), requires $\bigO(\lenof{\relarr{R}})$ work and space on an EREW PRAM.
	The output array is of size at most $\lenof{\relarr{R}}$.
	If $\relarr{R}$ is \(\attsetX\)-ordered, then the output is \(\attsetX\)-ordered, too.
\end{prop}

\begin{proof}
	The algorithm simply assigns  to each tuple one processor, which replaces the tuple by $\bot$ if it does not fulfil the selection condition.	Obviously, the output is ordered, if $\relarr{R}$ is ordered.
\end{proof}

\begin{prop}\label{alg:project}
  For every $\varepsilon>0$ and every  sequence \(\attsetX\) of attributes, there are \ctpalgos for \opnameProjection[\attsetX] that, given an array \(\relarr{R}\), have the following bounds on an Arbitrary CRCW PRAM.
  \begin{enumerate}[(a)]
	\item\WSBoundsDictionary%
		{\lenof{\relarr{R}}}%
		{\lenof{\relarr{R}}\cdot \lenof{\db}^{\varepsilon}};
	\label{bound:project-dictionary}
	\item Work and space $\bigO(\lenof{\relarr{R}}^{1+\varepsilon})$, if \(\relarr{R}\) is fully ordered.
	\label{bound:project-ordered}
      \end{enumerate}
      The algorithms also provide links from the tuples in \relarr{R} to the respective result tuples and from the tuples of the result relation to some ``source tuple''.
  The output array is of size  $\lenof{\relarr{R}}$.
  If $\relarr{R}$ is $\attsetY$-ordered, for some $\attsetY$, then the output is $\attsetY$-ordered, too.
\end{prop}
We note that in (b) the requirement that $\relarr{R}$ is fully linked is only needed to meet the space bound, since the computation of full links is possible within the given work bound (possibly using more space).

\begin{proof}[Proof for \autoref{alg:project}]
	The algorithm first computes with $\bigO(\relarr{R})$ work and space an array \relarr{R'} that has, for each tuple of \relarr{R}, the respective projected tuple.
	Then it invokes $\opDeduplicate{\relarr{R'}}$ to single out representatives for tuples with  multiple appearances. It then marks duplicates, i.e., cells that do not contain any representative, as uninhabited and outputs the resulting array.
        The links can be easily computed with the help of the links provided by $\opDeduplicate{\relarr{R'}}$.
	The work and space bounds are dominated by the deduplication, therefore the bounds follow from \autoref{result:alg-deduplicate-dictionary} and \autoref{result:alg-deduplicate-ordered}, respectively.
To apply \autoref{result:alg-deduplicate-ordered} it is crucial that  \(\relarr{R'}\) is fully ordered w.r.t.\ \(\attsetX\), which is the case since  \(\relarr{R}\) is ordered w.r.t.\ \(\attsetX\). %
\end{proof}

\begin{prop}\label{alg:semijoin}
  For every $\varepsilon>0$, there are \ctpalgos for \opnameSemiJoin that, given arrays \(\relarr{R}\) and \(\relarr{S}\) have the following bounds on an Arbitrary CRCW PRAM.
  \begin{enumerate}[(a)]
  \item\WSBoundsDictionary%
  	{\lenof{\relarr{R}} + \lenof{\relarr{S}}}%
  	{(\lenof{\relarr{R}} + \lenof{\relarr{S}})\cdot \lenof{\db}^{\varepsilon}};
  \label{bound:semijoin-dictionary}
  \item Work and space
	\(\bigO(\lenof{\relarr{R}} \cdot \lenof{\relarr{S}}^\varepsilon)\), if
    \(\relarr{S}\) is fully $\attsetX$-ordered and fully linked;\label{bound:semijoin-ordered}
  \item Work and space $\bigO((\lenof{\relarr{R}}+\lenof{\relarr{S}})\cdot\lenof{\relarr{R}}^{\varepsilon})$ if
  	\(\relarr{R}\) is fully $\attsetX$-ordered.
  \label{bound:semijoin-oredered-reverse}
\end{enumerate}
 Here, $\attsetX$ denotes the set of joint attributes of the relations $R$ and $S$.\footnote{In (b) and (c), $\attsetX$ can be any list of these attributes.}
  The output array is of size  $\lenof{\relarr{R}}$.
  If $\relarr{R}$ is \(\attsetY\)-ordered, then the output is \(\attsetY\)-ordered, too.
\end{prop}

\begin{proof}[Proof sketch]
	The \enquote{high-level} algorithm is identical for (a) and (c), but they yield different bounds  due to different implementations of some operations.
	Afterwards, we discuss how the algorithm is adapted for (b).

        First, the algorithm computes arrays \(\relarr{R'}\) and \(\relarr{S'}\) which represent $R' = \proj{R}{\attsetX}$ and $S' = \proj{S}{\attsetX}$ by applying the projection to $X$ to each tuple. These arrays are, in general, not concise. It further adds links between each tuple and its projection and vice versa.

        The algorithm then uses $\opSearchTuples{\relarr{R'},\relarr{S'}}$ to augment every tuple in $\relarr{R'}$ by a link, if there is a corresponding  tuple in $\relarr{S'}$. Tuples in \relarr{R'} without a corresponding tuple in \relarr{S'} are dropped. Finally, using $\lenof{\relarr{S}}$ processors, the algorithm computes, for those tuples, whose projection has a link to a tuple in \relarr{S'}, links to tuples in \relarr{S}, using the available pointers from \relarr{R}-tuples to \relarr{R'}-tuples and from \relarr{S'}-tuples to (some) \relarr{S}-tuples.
        For (a) and (c), the work and space for $\opnameSearchTuples$  dominate the overall work and space and the bounds therefore follow from \autoref{result:alg-search-dictionary} and \autoref{result:alg-search-fullylinkedA}, respectively.

        To achieve the stated work bound in (b), with the help of \autoref{result:alg-search-fullylinkedB}, \(\relarr{S'}\) is actually not materialised, but instead the tuple search algorithm is adapted to find partner tuples in~\relarr{S}, for each tuple in \relarr{R'}. This is possible, since \relarr{S} is fully \(\attsetX\)-ordered.
       \end{proof}

By almost the same algorithms, the same bounds hold for \opnameDifference.  However, a tuple remains in the result, if there is no partner tuple in \relarr{S}, and any list of attributes, for which \relarr{S} (or \relarr{R}) is ordered, works.

\begin{prop}\label{alg:difference}
  For every $\varepsilon>0$, there are \ctpalgos for \opnameDifference that, given arrays \(\relarr{R}\) and \(\relarr{S}\) have the following bounds on an Arbitrary CRCW PRAM.
  \begin{enumerate}[(a)]
  \item\WSBoundsDictionary%
  	{\lenof{\relarr{R}} + \lenof{\relarr{S}}}%
  	{(\lenof{\relarr{R}} + \lenof{\relarr{S}})\cdot \lenof{\db}^{\varepsilon}};
  \label{bound:difference-dictionary}
  \item Work and space
	\(\bigO(\lenof{\relarr{R}} \cdot \lenof{\relarr{S}}^\varepsilon)\), if
    \(\relarr{S}\) is fully  ordered and fully linked;\label{bound:difference-ordered}
  \item Work and space $\bigO((\lenof{\relarr{R}}+\lenof{\relarr{S}})\cdot\lenof{\relarr{R}}^{\varepsilon})$ if
  	\(\relarr{R}\) is fully  ordered.
  \label{bound:difference-ordered-reverse}
  \end{enumerate}
  The output array is of size  $\lenof{\relarr{R}}$.
  If $\relarr{R}$ is \(\attsetY\)-ordered, then the output is \(\attsetY\)-ordered, too.
\end{prop}

Finally, we consider the operation       \opnameUnion. It is trivial to compute a non-concise representation of \(R\cup S\) from concise representations of \(R\) and \(S\), by basically concatenating the arrays \relarr{R} and \relarr{S}. To obtain a concise representation, \autoref{alg:difference} can be used to compute a representation of $ R \setminus S$, which then can be concatenated with \relarr{S}. Conciseness is then guaranteed since $ R \setminus S$ and $S$ are disjoint. This yields the following result.

\begin{prop}\label{alg:union}
  For every $\varepsilon>0$, there are \ctpalgos for \opnameUnion that, compute concise representations of $R\cup S$, given concise arrays \(\relarr{R}\) and \(\relarr{S}\), and have the following bounds on an Arbitrary CRCW PRAM.
  \begin{enumerate}[(a)]
  \item\WSBoundsDictionary%
  	{\lenof{\relarr{R}} + \lenof{\relarr{S}}}%
  	{(\lenof{\relarr{R}} + \lenof{\relarr{S}})\cdot \lenof{\db}^{\varepsilon}};
  \label{bound:union-dictionary}
  \item Work and space \(\lenof{\relarr{R}} \cdot \lenof{\relarr{S}}^\varepsilon  + \lenof{\relarr{S}}\), if
    \(\relarr{S}\) is fully  ordered and fully linked.\footnote{Of course, the roles of $R$ and $S$ can be exchanged.}\label{bound:union-ordered}
  \end{enumerate}
  The output array is of size  $\lenof{\relarr{R}}+\lenof{\relarr{S}}$.
\end{prop}

\subsubsection{Algorithms for the \texorpdfstring{\opnameJoin}{Join}-Operation}
\label{section:ctp:join-algorithms}

In this section, we present algorithms for the \opnameJoin-operation. There is a straightforward naive algorithm that requires work $\Theta(\lenof{\relarr{R}} \cdot \lenof{\relarr{S}})$, comparing every tuple in \relarr{R} with every tuple in \relarr{S}. In fact, for unordered data, there is no better algorithm as will be shown in \autoref{result:dbops-lowerbound-general}.

In the dictionary setting, or if \relarr{R} and \relarr{S} are appropriately ordered one can do significantly better. However, \autoref{result:dbops-lowerbound} indicates that any constant time algorithm that computes the join of $R$ and $S$ and represents its result by a $\lambda$-compact array, for some $\lambda$, may require $\Omega(\lenof{\relarr{R}}^{1+\varepsilon}+\lenof{\relarr{S}}^{1+\varepsilon})$ processors.

On the other hand, since the output relation needs to be materialised by any algorithm, the work has to be $\Omega(\lenof{R\Join S})$.

We will see that these lower bounds can be almost matched: in fact, if one of the relations is properly ordered the work can be bounded by  $\bigO(\lenof{\relarr{R}}^{1+\varepsilon}+\lenof{\relarr{S}}^{1+\varepsilon}+\lenof{R\Join S}^{1+\varepsilon})$.

The algorithm first computes a $\lambda$-compact array \arr{A'} representing $\sjoin{R}{S}$  and a $\attsetX$-sorted array  \arr{B'} representing $\sjoin{S}{R}$, where $\attsetX$ is a list of all common attributes of $R$ and $S$. Afterwards, each tuple from  \arr{A'}  needs to be combined with each tuple from  \arr{B'} with the same  $\attsetX$-projection. Thanks to its ordering  the latter tuples are already grouped appropriately in~\arr{B'}.

However, to get the desired work bound, it does not suffice to merely $\lambda$-compact  \arr{B'}. To see this, consider the join of two relations $R$ and $S$ along a join attribute $X$. Let us assume, both relations consist of $n$ tuples and there are $n^{0.4}$ attribute values that occur $n^{0.3}$ times in both relations and all other attribute values occur only once. Each of the frequent values induces $n^{0.6}$ output tuples, the $n^{0.4}$ values together thus yield $n$ output tuples. The infrequent values result in slightly less than $n$ output tuples. Even if both $R$ and $S$ are $\lambda$-compact, it is possible that in $R$ and $S$ the tuples of each frequent value stretch over a subarray of $n^{0.6}$ cells. Assigning processors to each pair of such cells (one in $R$ and one in $S$) would require $n^{1.5}$ processors,  violating the desired upper bound.
To overcome this challenge, the algorithm will rather  compact each group separately.

We are now ready to present our algorithms for the \(\opnameJoin\) operation.

\newcommand{\wboundjoindict}[3][\varepsilon]{%
	(\lenof{\relarr{#3}}+\lenof{D})^{1+#1}
	+ \lenof{\relarr{#2}}^{1+#1}
	+ \lenof{#2 \bowtie #3}^{1+#1}%
}

\newcommand{\sboundjoindict}[3][\varepsilon]{%
	(\lenof{\relarr{#3}}+\lenof{D})^{1+#1}
	+ \lenof{\relarr{#2}}^{1+#1}
	+ \lenof{#2 \bowtie #3}^{1+#1}
      }

\begin{prop}\label{alg:join}
	For every $\varepsilon > 0$ and every $\lambda > 0$, there are \ctpalgos for \opnameJoin that, given arrays  \(\relarr{R}\) and \(\relarr{S}\) have the following bounds on an Arbitrary CRCW PRAM.
	Here, $\attsetX$ denotes the joint attributes of $R$ and $S$.
	\begin{enumerate}[(a)]%
	\item Work and space \(\bigO(\wboundjoindict{R}{S})\) in the dictionary setting; %
	\label{bound:join-dictionary}
	\item Work and space \(\bigO(\lenof{\relarr{S}}^{1+\varepsilon}
			+ \lenof{\relarr{R}}^{1+\varepsilon}
			+ \lenof{R \bowtie S}^{1+\varepsilon}
		)\)
		 if $\relarr{S}$ is fully $\attsetX$-ordered.
	\label{bound:join-ordered}
\end{enumerate}
In both cases, the size of the output array is bounded by $(1+\lambda)\lenof{R \bowtie S}$.
\end{prop}

\begin{proof}
  Let $R$ and $S$ denote the relations represented by \(\relarr{R}\) and \(\relarr{S}\), respectively.
  Let further $\delta = \min\{\frac{1}{3},\frac{\varepsilon}{3}\}$ and $\padeps' = \min\{\frac{1}{3},\frac{\padeps}{3}\}$.

  The algorithm for (a)  is by sorting \relarr{S} according to \autoref{result:alg-sort} and applying  the algorithm for (b) afterwards.

  The algorithm for (b) proceeds as follows.
  \begin{enumerate}[(1),ref={Step~(\arabic*)}]
  \item Compute array \subarr{A}{1} representing $\sjoin{R}{S}$ and $\padeps'$-compact it.\label{join-alg-step-A1}
  \item  Compute an $\attsetX$-sorted array  \subarr{B}{1} representing $\sjoin{S}{R}$.\label{join-alg-step-B1}
  \item Fully link \subarr{B}{1}  according to \autoref{result:predsucc}.\label{join-alg-step-B1link}
  \item Let \subarr{B}{2} be $\opProjection[\attsetX]{\subarr{B}{1}}$.\label{join-alg-step-B2}
  \item For each proper tuple $s$ of \subarr{B}{2}, compute the positions $i_1(s)$ and $i_2(s)$ of the first and last tuple of \subarr{B}{1} whose  $\attsetX$-projection is $s$, respectively.\label{join-alg-step-groupendpoints}
  \item For each proper tuple $s$ of \subarr{B}{2}, $\padeps'$-compact the subarray of \subarr{B}{1} from position $i_1(s)$ to $i_2(s)$. Let \subarr{B}{3} denote the result array.\label{join-alg-step-B3}
   \item For each proper tuple $s$ of \subarr{B}{2}, compute the positions $j_1(s)$ and $j_2(s)$ of the first and last tuple of \subarr{B}{3} whose  $\attsetX$-projection is $s$, respectively.\label{join-alg-step-j1-j2}
  \item For each proper tuple $t$ of  \subarr{A}{1}, find the unique corresponding tuple $s(t)$ in  \subarr{B}{2}.\label{join-alg-step-st}
\item For each proper tuple $t$ of \subarr{A}{1} and each proper tuple $s'$ between positions $j_1(s(t))$ and $j_2(s(t))$, construct an output tuple by combining $s'$ with $t$.\label{join-alg-step-output}
  \end{enumerate}

  \ref{join-alg-step-A1} can be done according to \autoref{alg:semijoin} (b) and \autoref{result:alg-compact}. \ref{join-alg-step-B1} uses the algorithm of \autoref{alg:semijoin} (c). \ref{join-alg-step-B1link} uses the algorithm of \autoref{result:predsucc} and \ref{join-alg-step-B2} the algorithm of \autoref{alg:project}. \ref{join-alg-step-groupendpoints} and \ref{join-alg-step-j1-j2} can be done by slightly adapted versions of an algorithm for \opnameSearchTuples using \autoref{result:alg-search-fullylinkedB}. \ref{join-alg-step-B3} uses \autoref{result:task-scheduling} for the allocation of processors, assigning $(i_2(s)-i_1(s)+1)^{1+\delta}$ processors to the group of each $s$, and \autoref{result:alg-compact} with parameter $\delta$ (instead of $\varepsilon$) for each group. Finally, \ref{join-alg-step-output} uses \autoref{result:task-scheduling} for the allocation of processors. For each tuple $t$ it assigns  $j_2(s(t))-j_1(s(t))+1$ processors. The position of a result tuple in the output array is just the number of the processor that writes it.

  For most steps the work and space bounds follow easily from the observations that all \arr{A}-arrays have size at most \lenof{\relarr{R}}, all \arr{B}-arrays have size at most \lenof{\relarr{S}}, and that $\lenof{\relarr{R}}\cdot \lenof{\relarr{S}}^\varepsilon$ is bounded by $\lenof{\relarr{R}}^{1+\varepsilon}$ or  $\lenof{\relarr{S}}^{1+\varepsilon}$.

  Towards the complexity analysis of \ref{join-alg-step-B3} and \ref{join-alg-step-output} we denote, for each tuple $s$  from \subarr{B}{2} the number $(i_2(s)-i_1(s)+1)^{1+\delta}$ of processors by $m_s$ and the number of proper tuples from \subarr{B}{1} matching $s$ by $n_s$. After \ref{join-alg-step-B3} the overall size of all groups in  \subarr{B}{3} is at most
  	\[
		(1+\padeps')\sum_{s\in \subarr{B}{2}} m_s
		= (1+\padeps')\sum_{s\in \subarr{B}{2}} (i_2(s)-i_1(s)+1)^{1+\delta}
		\le (1+\padeps')\left(\sum_{s\in \subarr{B}{2}} (i_2(s)-i_1(s)+1)\right)^{1+\delta}
              \]
              and the latter is in $\bigO(\lenof{\relarr{S}}^{1+\delta})$.
              Therefore, \autoref{result:task-scheduling} yields  a work and space bound $\lenof{\relarr{S}}^{(1+\delta)^2}$  for \ref{join-alg-step-B3}, which is $\lenof{\relarr{S}}^{1+\varepsilon}$ thanks to the choice of $\delta$. For step \ref{join-alg-step-output} we observe that, for each $s$ from \subarr{B}{2}, it holds $j_2(s)-j_1(s)+1\le (1+\padeps') n_s$. Therefore,
  	\[
		\sum_{t\in\sjoin{R}{S}}  (j_2(s(t))-j_1(s(t))+1) \le  \sum_{t\in\sjoin{R}{S}} (1+\padeps') n_{s(t)} =  (1+\padeps')\lenof{\join{R}{S}}.
        \]
	Thus,  \autoref{result:task-scheduling} establishes the work and space bound $\lenof{R \bowtie S}^{1+\varepsilon}$ for \ref{join-alg-step-output} and guarantees that the size of the output array is at most $(1+\padeps')^2 \lenof{\join{R}{S}}\le (1+\padeps) \lenof{\join{R}{S}}$.

\end{proof}

\section{Query Evaluation in the Dictionary Setting}\label{section:ctp:query-evaluation}

After studying algorithms for basic operations and operators of the
relational algebra, we are now prepared to investigate the complexity
of \ctpalgos for query evaluation.
In this section we will focus on the dictionary setting.
We will derive algorithms and bounds for other settings by a
transformation into the dictionary setting in
\autoref{section:ctp:othersettings}.

As mentioned in the introduction lower bound results for the size of
bounded-depth circuits for the clique problem and the clique
conjecture destroy the hope for parallel constant-time
evaluation algorithms for conjunctive queries with small work.

We therefore concentrate in this section on restricted query
evaluation settings. We study two restrictions of query languages
which allow efficient sequential algorithms, the semijoin algebra and
(free-connex)  acylic conjunctive queries. Furthermore, we present
a $\bigO(1)$-time parallel version of worst-case optimal join
algorithms.

For notational simplicity, we  assume in this section that database
relations are represented concisely by \emph{compact} arrays without
any uninhabited cells. This allows us to express bounds in terms of
the size of the database as opposed to the sizes of its
arrays.\footnote{Assuming that the input arrays are $\lambda$-compact,
for some constant $\lambda$, would suffice for the stated purpose.}

In the following, \IN thus always denotes the maximum number of cells in
any array representing a  relation of the underlying database that is addressed by the given
query.
Note that \(\lenof{\db}\in\bigO(\IN)\) for any database \(\db\) over a
fixed schema.
Analogously, \(\OUT\) denotes the number of tuples in the query result.

The results of this section are summarised in the right-most column of
\autoref{table:overview-main-results}.

\begin{table}
	\caption{Overview of results on query evaluation. Here CQ is
          short for \enquote{conjunctive query} and ``a-ordered''
          refers to the attribute-wise ordered setting.}
	\begin{center}
\begingroup%
\newcommand{\smallcref}[1]{\footnotesize\autoref{#1}}
\begin{tabular}{l c c c}
	\toprule
	\multirow{2}{*}{Query class} & \multicolumn{3}{c}{Work bound in the \ldots}\\
				& general setting & a-ordered setting & dictionary setting\\
	\midrule
	\addlinespace
	\multirowcell{2}[0pt][l]{Semijoin\\ Algebra}	& \(\bigO(\IN^2)\)	& \(\bigO(\IN^{1+\varepsilon})\)	& \(\bigO(\IN)\)\\
						&
                                           \smallcref{result:acyclic-etc-general}
                                & \smallcref{result:queries-a-ordered}
                                                  (a)& \smallcref{result:eval-semijoin-algebra}\\
	\addlinespace
	\midrule
	\addlinespace
	\multirowcell{2}[0pt][l]{Acyclic CQs}			& \(\bigO((\IN\cdot\OUT)^{1+\varepsilon} + \IN^2)\) & \multicolumn{2}{c}{\(\bigO((\IN\cdot\OUT)^{1+\varepsilon})\)}\\
						&
                                           \smallcref{result:acyclic-etc-general}
                                &
                                  \smallcref{result:queries-a-ordered}
                                                  (b) & \smallcref{result:acyclic} (a)\\
	\addlinespace
	\multirowcell{2}[0pt][l]{Free-Connex\\ Acyclic CQs} & \(\bigO((\IN+\OUT)^{1+\varepsilon} + \IN^2)\) & \multicolumn{2}{c}{\(\bigO((\IN+\OUT)^{1+\varepsilon})\)}\\
						&
                                           \smallcref{result:acyclic-etc-general}
                                &
                                  \smallcref{result:queries-a-ordered}
                                                  (b) & \smallcref{result:acyclic} (b)\\
	\addlinespace
	\multirowcell{2}[0pt][l]{CQs}	& \(\bigO((\IN^{\ghw(q)}\cdot\OUT)^{1+\varepsilon} + \IN^2)\) & \multicolumn{2}{c}{\(\bigO((\IN^{\ghw(q)}\cdot\OUT)^{1+\varepsilon})\)}\\
						&
                                           \smallcref{result:acyclic-etc-general}
                                &
                                  \smallcref{result:queries-a-ordered}
                                                  (b) & \smallcref{result:eval-ghw} (a)\\
	\addlinespace
	\multirowcell{2}[0pt][l]{CQs} & \(\bigO((\IN^{\fghw(q)} + \OUT)^{1+\varepsilon} + \IN^2)\) & \multicolumn{2}{c}{\(\bigO((\IN^{\fghw(q)}+\OUT)^{1+\varepsilon})\)}\\
						&
                                           \smallcref{result:acyclic-etc-general}
                                &
                                  \smallcref{result:queries-a-ordered}
                                                  (b) & \smallcref{result:eval-ghw} (b)\\
	\addlinespace
	\midrule
	\addlinespace
	\multirowcell{2}[0pt][l]{Natural Join\\ Queries} & \(\bigO((\agmbound + |D|)^{1+\varepsilon} + |D|^2)\) & \multicolumn{2}{c}{\(\bigO((\agmbound + |D|)^{1+\varepsilon})\)}\\
						& \smallcref{result:acyclic-etc-general} & \smallcref{result:queries-a-ordered}
                                                  (b) & \smallcref{result:worst-case-eval-dict}\\
    \addlinespace
	\bottomrule
\end{tabular}
\endgroup
 	\end{center}
	\label{table:overview-main-results}
\end{table}

\subsection{Query Plans}\label{section:query-plans}
To investigate the complexity of \ctpalgos for query evaluation we will utilize \emph{query plans}, which we briefly recall here from the literature.
Furthermore, we present a first result which allows us to, given a query plan, derive work and space bounds for query evaluation from bounds for intermediate results.

A \emph{query plan} \(\qplan = (V, E, \omega)\) over a database schema \(\schema\) is a finite, directed, ordered, binary tree\footnote{Sometimes query plans are defined more generously as acyclic graphs. However, regarding the size of intermediate results, there is no essential difference between such a query plan and its tree-unravelling.} with a node labelling function \(\omega\) that labels
\begin{itemize}
	\item every leaf with a relation symbol from \(\schema\);
	\item every node with a single child with either a select or a project operator; and
	\item every node with two children with either a semijoin, join, difference, or union operator.
\end{itemize}
We note that in case of a select \(\sigma_{\attA = x}\) or a project operator \(\pi_{\attsetX}\), \(\attA\) and \(x\), or \(\attsetX\), respectively, are part of the label.
For convenience, we also  do not consider nodes with a renaming operator, since attributes can always be renamed in constant time on a single processor.

Every query plan represents a query of the relational algebra\footnote{A query plan can also be understood as a parse tree of a relational algebra query.} but there might be multiple query plans for a query.
In fact, we will often describe query plans which do not directly correspond to a given query expression.
Moreover, we will often describe query plans as a sequence of operations, if it is clear that they can be arranged in a tree.

Given a database \(\db\) and a query plan \(\qplan\) over a common schema, we inductively associate each node \(v\) of \(\qplan\) with a relation \(R_{\db,v}\) as follows.
For a leaf \(v\) labelled with a relation symbol \(R = \omega(v)\), we set \(R_{\db,v} = D(R)\).
For nodes \(v\) labelled with a unary operator \(\tau = \omega(v)\), we set \(R_{\db,v} = \tau(R_{\db,w})\) where \(R_{\db,v}\) is the relation associated with the child \(w\) of \(v\).
And for nodes labelled with a binary operator \(\circ = \omega(v)\) , we set \(R_{\db,v} = R_{\db,w} \circ R_{\db,u}\) where \(R_{\db,w}\) and \(R_{\db,u}\) are the relations associated with the left child \(w\) and the right child \(u\) of \(v\), respectively.
The \emph{query result} \(\queryresult{\qplan}{\db}\) of a query plan \(\qplan\) with root node \(v\) on a database \(\db\) is \(R_{\db,v}\).
A \emph{query plan \(\qplan\) for a query \(\query\)} satisfies \(\qplan(\db) = \queryresult{\query}{\db}\) for all databases \(\db\).

The inductive definition of the \(R_{\db,v}\) gives rise to a straightforward evaluation scheme: Simply compute \(R_{\db,v}\) in a bottom-up fashion.
The following elementary result provides an upper bound for the work and space required by a \ctpalgo implementing this scheme, given an upper bound  for the size of the intermediate relations. %

\begin{thm}\label{result:eval-query-plan}
  For every query plan \(\qplan = (V, E, \omega)\) for a query $\query$ and every $\varepsilon>0$, $\lambda>0$ there is a \ctpalgo  on an Arbitrary CRCW PRAM and a constant $c$, such that the following holds.
 For every database $\db$, given in a $\lambda$-compact representation in the dictionary setting,  the algorithm evaluates
	$\queryresult{\query}{\db}$ with work and space \(cN^{1+\varepsilon}\), where \(N = \max(\lenof{\db},\max_v \lenof{R_{\db,v}})\).
	The output array is of size \(\bigO(\lenof{\query(\db)})\).

	If no node of \(\qplan\) is labelled with the join operator,  work \(cN\) suffices.
	The output array is then of size \(\bigO(\lenof{\db})\).
\end{thm}

\begin{proof}
	The \ctpalgo computes $\lambda$-compact arrays \(\arr{R}_{\db,v}\) for the relations \(R_{\db,v}\) in a bottom-up fashion.
	It maintains the invariant \(\lenof{\arr{R}_{\db,v}} \le (1+\lambda) N\), for all \(v\in V\). We show that for each node $v$, there is a constant $c_v$, such that the evaluation of \(\arr{R}_{\db,v}\) needs at most work and space \(c_vN^{1+\varepsilon}\). Then, $c$ can be chosen as the sum of all $c_v$.

	For a leaf node \(v\) the claim follows immediately, because the $\lambda$-compact array \(\arr{R}_{\db,v}\) for \(R_{\db,v}\) is given as input and $N$ bounds its size.
	Consider a node \(v\) which has a single child node \(w\).
	Then \(v\) is labelled with a select or project operator, and the algorithms from \autoref{alg:selection} and \autoref{alg:project} (a) guarantee that \(\arr{R}_{\db,v}\) can be computed with work \(\bigO(\lenof{\arr{R}_{\db,w}})\) and space \(\bigO(\lenof{\arr{R}_{\db,w}}\cdot \lenof{\db}^\varepsilon)\).
	Thanks to the invariant, and $N$ being a bound on $\lenof{\arr{R}_{\db,v}} $ and $\lenof{\arr{R}_{\db,w}} $, this amounts to work \(\bigO(N)\) and space \(\bigO(N^{1+\varepsilon})\).
	The invariant  holds since  \(\lenof{\arr{R}_{\db,v}}\le\lenof{\arr{R}_{\db,w}}\). The constant $c_v$ is obtained by adding the constant factor for the application of the select or project operator to $c_w$.

	Consider now a node \(v\) with two child nodes \(w\) and \(u\).
	If \(v\) is labelled with a set difference, union, or semijoin operator, the claim follows analogously to the case for the select and project operators by replacing the term \(\lenof{\arr{R}_{\db,w}}\) with \(\lenof{\arr{R}_{\db,w}} + \lenof{\arr{R}_{\db,u}}\) and using the corresponding \enquote{(a)-algorithms} from \autoref{section:ctp:alg-for-db-ops} (see \autoref{table:complexity-bounds-semijoin-dictionary} for an overview). 

	It remains to consider the case that \(v\) is labelled with the join operator.
	Thanks to \autoref{alg:join}\ref{bound:join-dictionary} the join \(R_{\db,v} = \join{R_{\db,w}}{R_{\db,u}}\) can be computed with work and space  \[
		\bigO\left(
			\left(\lenof{\arr{R}_{\db,u}} + \lenof{\db}\right)^{1+\varepsilon} + \lenof{\arr{R}_{\db,w}}^{1+\varepsilon} + \lenof{\join{R_{\db,w}}{R_{\db,u}}}^{1+\varepsilon}
		\right).
	\]
	Thanks to the invariant and $N$  bounding $\lenof{\arr{R}_{\db,w}} $ and $\lenof{\arr{R}_{\db,u}} $, the claim regarding work and space follows.
	Furthermore, \autoref{alg:join} guarantees an output size of \[(1+\lambda)\lenof{\join{R_{\db,w}}{R_{\db,u}}} =  (1+\lambda)\lenof{R_{\db,v}},\] and hence the invariant holds, since $N$ also bounds  \(\lenof{\arr{R}_{\db,v}}\).
        In all binary cases, the constant $c_v$ is obtained by adding the constant factor for the application of the respective operator to $c_w+c_u$.

	Overall, the query plan can be evaluated with the stated work and space bounds.
	Since all steps, except for the computation of a join, can be done with linear work, we get an improved work bound of  \(cN\), if no node of the join plan is labelled with the join operator. However, this work bound does not allow to apply compaction and therefore in this case  we only get the weaker bound \(\bigO(\lenof{\db})\) for the output size. 
\end{proof}

\subsection{Semijoin Algebra}\label{section:ctp:semijoin-algebra}
The semijoin algebra is the fragment of the relational algebra that
does \emph{not} use the join operation (nor Cartesian product), but
can use the semijoin operation instead --- besides
 selection, projection, rename, union, and set difference. It is well-known that semijoin queries produce only
query results of size $\bigO(\lenof{\db})$ \cite[Corollary~16]{DBLP:journals/jcss/LeindersB07} and can be evaluated in time
$\bigO(\lenof{\db})$ \cite[Theorem 19]{DBLP:journals/jolli/LeindersMTB05}.

Since semijoin algebra queries do not use the join operator, the following is a direct consequence of \autoref{result:eval-query-plan} and the fact that the size of all intermediate results can be bounded by \(\bigO(\IN)\), since they are themselves results of semijoin (sub-)queries.
\begin{thm}\label{result:eval-semijoin-algebra}
  For every $\varepsilon>0$, $\lambda>0$ and each query $\query$ of the semijoin algebra there is a \ctpalgo that,
  given a $\lambda$-compact representation of a database $\db$, evaluates
	$\queryresult{\query}{\db}$ on an Arbitrary CRCW PRAM with  \wsbounds{\IN}{\IN^{1+\varepsilon}} in the
        dictionary setting.
  The output array is of size \(\bigO(\IN)\).
\end{thm}
Altogether, semijoin queries can be evaluated work-optimally by a
\ctpalgo in the dictionary setting.

\subsection{Acyclic and other Conjunctive Queries}\label{section:ctp:acyclic-queries}

There are various fragments of the class of conjunctive queries, defined by restrictions on the structure of a query, that allow efficient query evaluation. For instance, there are the acyclic join queries and, more generally, acyclic queries, which can be evaluated efficiently by the well-known Yannakakis algorithm \cite{DBLP:conf/vldb/Yannakakis81}, and their restriction to free-connex queries that allows for efficient enumeration of query results \cite{DBLP:conf/csl/BaganDG07}. Furthermore, there are various generalisations of these fragments, e.g., based on the notion of generalised hypertree width. In this subsection, we study how queries of some of these fragments can be evaluated work-efficiently in parallel constant time. We first study acyclic queries and their variants and then generalise the results to queries of bounded generalised hypertree width. Before that we recall the definition of conjunctive queries and fix our notation.

Conjunctive queries are conjunctions of relation atoms. We write a
\emph{conjunctive query} (\emph{CQ} for short) \query as a rule of the
form $q\colon \atom \gets \atom_1,\dots,\atom_m$, where
$\atom, \atom_1,\dots,\atom_m$ are atoms and $m\ge 1$.
Here an atom has the form \(R(x_1,\ldots,x_k)\) where \(R\) is a relation symbol with arity \(k\) and \(x_1,\ldots,x_k\) are variables.
The atoms \(\atom_1,\ldots,\atom_m\) form the \emph{body} of \(\query\) and \(\atom\) is the head of \(\query\).
All relation symbols occurring in the body are from the database schema \(\schema\) and the relation symbol of the head is, on the other hand, \emph{not}.
Further, we only consider \emph{safe} queries; that is, every variable which occurs in the head, also occurs in (some atom of) the body.
If a variable occurs in the head atom, it is a \emph{free} variable.
Otherwise, the variable is \emph{quantified}.
A \emph{join query} is a conjunctive query with no quantified variables, i.e.\ every variable in a join query is free.
For more background on (acyclic) conjunctive queries, including their semantics, we refer to \cite{DBLP:books/aw/AbiteboulHV95,ABLMP21}.

In this section we will assume that there are no variable repetitions in any atom, i.e.\ the variables \(x_1,\ldots,x_k\) of an atom \(R(x_1,\ldots,x_k)\) are pairwise distinct. We note that we do not allow constants in our definition of conjunctive queries, either.
Both requirements can easily be established by proper precomputations using only  \(\opnameSelection[\attsetX]\) which only requires linear work and space.
Furthermore, we assume that no relation symbol appears more than once in the body of the query. This is just a matter of renaming (or copying) some input relations.

\paragraph*{Acyclic conjunctive queries.}

A conjunctive query $\query$ is \emph{acyclic}, if it has a join tree $\querytree{\query}$.
A \emph{join tree} for \(\query\) is an undirected, rooted tree $(V(\querytree{\query}),E(\querytree{\query}))$ where
the set of nodes $V(\querytree{\query})$ consists of (all) the atoms in $\query$ and for each variable $x$ in
$\querytree{\query}$ the set $\{v \in V(\querytree{\query}) \mid v \text{
  contains }x\}$ induces a connected subtree of $\querytree{\query}$.

An acyclic query $q$ with free variables $x_1,\ldots,x_\ell$, for which the query~$q'$ that results from~$q$ by adding a new atom $A'(x_1,\ldots,x_\ell)$ to its body, remains acyclic, is called \emph{free-connex acyclic}.\footnote{There are several equivalent definitions of free-connex acyclic queries. We chose this definition from \cite[Section 3.2]{HAL:braultbaron:tel-01081392}, where such queries are called \emph{starred}, since it is easy to state.}
Acyclic join queries are a special case of free-connex acyclic queries, in which  all variables of the body of the query are free.

Our results on acyclic conjunctive queries are obtained by an easy application of \autoref{result:eval-query-plan} to the already mentioned Yannakakis algorithm
\cite{DBLP:conf/vldb/Yannakakis81}. This algorithm consists of two parts. Given an acyclic conjunctive query $\query$, a join tree
$\querytree{\query}$ for \(\query\) and a database $\db$, it computes in the first part a \emph{full reduction} $\db'$ of $\db$ as in \cite[Section 3]{BernsteinG81}.
To this end, we associate with each node $v$ in
$\querytree{\query}$ a relation $S_v$. Initially,
$S_v = R_v(\db)$, where $R_v$ is the relation symbol of $v$.
The computation of the full reduction has two phases.

\begin{enumerate}[(1)]
	\item \textbf{bottom-up semijoin reduction:} All nodes
          are visited in bottom-up traversal order of $\querytree{\query}$. When a node
          $v$ is visited, $S_v$ is updated to $\sjoin{S_v}{S_w}$ for
          every child $w$ of $v$ in $\querytree{\query}$.
	\item \textbf{top-down semijoin reduction:} All nodes are
          visited in  top-down traversal order of $\querytree{\query}$. When a node $v$ is visited, the relation $S_w$ is updated to $\sjoin{S_w}{S_v}$ for every child $w$ of $v$ in $\querytree{\query}$.
\end{enumerate}
Let $\db'$ denote the resulting database, i.e., the database that has, for each node $v$ the relation $S_v$ in place of $R_v$ (and is identical with \db, otherwise).
\begin{propC}[{\cite[Theorem 1]{BernsteinG81}}]\label{result:fullreduction}
  For every acyclic conjunctive query \(\query\), database \(\db\), and full reduction \(\db'\) of \(\db\), it holds that
  \begin{enumerate}[(a)]
  \item $\query(\db')=\query(\db)$; and
  \item $\pi_{\attrof{\query}\cap\attrof{R}}(\db'(R))=\pi_{\attrof{R}}(\query(\db))$, for every relation symbol $R$ of \db.
  \end{enumerate}
\end{propC}

In the second part, the algorithm computes the join of the reduced relations by a succession of binary joins in a bottom-up manner. Since each tuple of an intermediate result relation for a node $v$ can be partitioned into a
subtuple of a tuple from $R_v$ and a subtuple of the output, the following bound holds.
\begin{lemC}[{\cite[Lemma 4.1]{DBLP:conf/vldb/Yannakakis81}}]\label{lem:yannakakis-intermediate}
  All intermediate result relation in the second part of the Yannakakis algorithm are of size at most $\IN\cdot\OUT$.
\end{lemC}

Now we are ready to state and prove our results about acyclic queries.
\begin{thm}\label{result:acyclic}
  For every $\varepsilon>0$, $\padeps>0$, and each acyclic conjunctive query \query, there are \ctpalgos that  given a $\padeps$-compact representation of a database \(\db\) in the dictionary setting, compute $\queryresult{\query}{\db}$ on an Arbitrary CRCW PRAM
  \begin{enumerate}[(a)]
  \item with work and space $\bigO((\IN\cdot\OUT)^{1+\varepsilon})$;
  \item with work and space $\bigO((\IN+\OUT)^{1+\varepsilon})$, if \query is a free-connex acyclic query.
  \end{enumerate}
  In both cases, the result can be represented by an array of size  \(\bigO(\OUT)\). 
\end{thm}
\begin{proof}
  Let all relations $R$ occurring on \query be represented concisely by arrays \(\relarr{R}\).

  In both cases, the algorithm first computes  the full reduction as in \autoref{result:fullreduction}.
 Since the full reduction constitutes a query plan without joins, and all intermediate result relations are subrelations of relations of the database, this can be done with work $\bigO(\IN)$ and space $\bigO(\IN^{1+\varepsilon})$  by \autoref{result:eval-query-plan}.

 Afterwards every node \(v\) in \(\querytree{\query}\) is associated with an array \(\relarr{S_v}\) of size at most \(\bigO(\IN)\) which represents \(S_v\) concisely.

Statement (a) now follows by applying \autoref{result:eval-query-plan} to  the second part of Yannakakis' algorithm thanks to the bound on intermediate results established by \autoref{lem:yannakakis-intermediate}. 

                Towards (b), let $\querytree{\query'}$ be a join tree for the query $\query'$ with the additional atom $R'(x_1,\ldots,x_\ell)$, where $x_1,\ldots,x_\ell$ are the free variables of \query. The query result $\query(\db')=\query(\db)$ can be obtained as the join of all $\pi_X(A)$, where $A$ is an atom of $\query'$ whose node is a neighbour of $R'(x_1,\ldots,x_\ell)$ in   $\querytree{\query'}$, and $X$ is the set of free variables of \query that occur in $A$. This follows from \autoref{result:fullreduction} and the observation that each free variable of \query occurs in some neighbour atom of $R'(x_1,\ldots,x_\ell)$ in $\query'$ since it occurs in \emph{some} atom by the safety of \query and thus in a neighbour atom by the connectedness condition of the acyclic query $\query'$.

                This join can be computed by a sequence of binary joins, each of which producing a result relation that is a projection of $\query(\db')$ and is therefore of size at most $\OUT$. Thus, again by \autoref{result:eval-query-plan}, the work and space are bounded by $\bigO((\IN+\OUT)^{1+\varepsilon})$.

                To achieve the size of the output array in (a), a final compaction might be necessary. For (b), it is already guaranteed by \autoref{alg:join}.%
\end{proof}

\paragraph*{General conjunctive queries.}
We next investigate how the algorithms for acyclic queries can be adapted to all conjunctive queries. More precisely, we give algorithms with  upper bounds for the work and space for the evaluation of conjunctive queries depending on their generalised hypertree width.

We use the notation of  \cite{DBLP:conf/pods/GottlobGLS16}.
A \emph{tree decomposition} of a query $\query$ is a pair $(\tree,\bag)$ where $\tree = (V,E)$ is an undirected, rooted tree, and $\bag$ is a mapping that maps every node $\tnode \in V$ in $\tree$ to a subset of variables in $\query$ such that the following three conditions are satisfied.\footnote{We note that Condition (1) is redundant thanks to safety and Condition (2).}
\begin{enumerate}[(1)]
	\item For every variable $x$ in $\query$ there is a node $\tnode \in V$ such that $x \in \bag(\tnode)$;
	\item for each atom $R(x_1,\ldots,x_r)$ in $\query$ there is a node $\tnode \in V$ such that $\{x_1,\ldots,x_r\}\subseteq \bag(\tnode)$; and
	\label{def:hypertree-decomp:condition2}
	\item for each variable $x$ in $\query$ the set $\{\tnode \in V| x \in \bag(\tnode)\}$ induces a connected subtree of $\tree$.
\end{enumerate}

A \emph{generalised hypertree decomposition} of a query $\query$ is a triple $(\tree,\bag,\cover)$ where $(\tree,\bag)$ is a tree decomposition of $\query$ and $\cover$ is a mapping\footnote{In the literature, $\lambda$ is often used to denote this mapping, but $\lambda$ has a different role in this article.} which maps every node $v$ of $\tree$ to a set of atoms from the body of $\query$ such that $\bag(\tnode) \subseteq \bigcup_{R(x_1,\ldots,x_r)\in\cover(v)}\{x_1,\ldots,x_r\}$.

The \emph{width} of a generalised hypertree decomposition $(\tree,\bag,\cover)$ is the maximal number $\max_{\tnode \in V} \lenof{\cover(\tnode)}$ of atoms assigned to any node of \tree.
The \emph{generalised hypertree width} $\ghw(\query)$ of a conjunctive query $\query$ is the minimal width over all of its generalised hypertree decompositions.

We note that a conjunctive query is acyclic, if and only if its generalised hypertree width is $1$ \cite[Theorem~4.5]{DBLP:journals/jcss/GottlobLS02}.\footnote{We note that the proof given in \cite{DBLP:journals/jcss/GottlobLS02} is for hypertree decompositions which impose an additional condition which is, however, not used in the proof.}%

A tree decomposition $(\tree,\bag)$ of a conjunctive query $\query$ is \emph{free-connex} if there is a set of nodes $U \subseteq V(\tree)$ that induces a connected subtree in $\tree$ and satisfies $\freeof{\query} = \bigcup_{u \in U} \bag(u)$ \cite[Definition 36]{DBLP:conf/csl/BaganDG07}. A generalised hypertree decomposition is \emph{free-connex} if its tree decomposition is free-connex. The \emph{free-connex generalised hypertree width} $\fghw(\query)$ of a conjunctive query is  the minimal width among its free-connex generalised hypertree decompositions.

\begin{thm}\label{result:eval-ghw}
	For every $\varepsilon > 0$, $\padeps>0$, and every conjunctive query \(\query\), there are \ctpalgos that, given a $\padeps$-compact representation of a database \(\db\) in the dictionary setting, computes $\queryresult{\query}{\db}$ on an Arbitrary CRCW PRAM
        \begin{enumerate}[(a)]
        \item with work and space
          $\bigO((\IN^{\ghw(q)} \cdot \OUT)^{1+\varepsilon})$;
        \item with work and space
          $\bigO((\IN^{\fghw(q)} + \OUT)^{1+\varepsilon})$.
        \end{enumerate}
        As in \autoref{result:acyclic}, the result can be represented by an array of size  \(\bigO(\OUT)\). 
\end{thm}

\begin{proof}
  We basically show that the reduction to the evaluation of an acyclic query of \cite[Section~4.2]{DBLP:journals/jcss/GottlobLS02} transfers, in both cases, well to our parallel setting.

  We first show (a).
  Let thus $\varepsilon > 0$ and  a  conjunctive query \(\query\) be fixed and let
  $(\tree,\bag,\cover)$ be a generalised hypertree decomposition of \query of width $\ghw(q)$.  By \cite[Lemma~4.4]{DBLP:journals/jcss/GottlobLS02} we can assume that it is \emph{complete}, i.e.,  for every atom \(\atom\) in the body of \(\query\), there is a node \(\tnode\) such that \(\atom\in\cover(v)\).

The idea is to compute an acyclic query \(\query'\) from \query and a database  \(\db'\) from \db and \query such that $\query'(\db')=\query(\db)$.

	We first describe \(\db'\) and \(\query'\) and then show how \(\db'\) can be computed in constant time.

	For every $\tnode \in V(\tree)$, we define a new relation \(R_\tnode\) by means of the conjunctive query \(\query_\tnode\) with free variables \(\bag(\tnode)\) and body \(\cover(\tnode)\).
	The database $\db'$ consists of all these relations $R_\tnode$.
	That is, the schema of $\db'$ is \(\{R_\tnode \mid \tnode \in V(\tree)\}\) and \(\db'(R_\tnode) = \query_\tnode(\db)\) for all \(\tnode\in V(\tree)\).

	The query $\query'$ is then the conjunctive query with the same head as \(\query\) and its body consists of the head atoms of all queries \(\query_\tnode\).
	It is straightforward to verify that $\queryresult{\query}{\db} = \queryresult{\query'}{\db'}$ by inlining the definitions of the \(\query_\tnode\) and observing that every atom of \(\query\) occurs in one of the \(\query_\tnode\) thanks to the decomposition being complete.
	Moreover, \(\query'\) is acyclic because the generalised hypertree decomposition \((\tree,\bag,\cover')\) for \(\query'\) with \(\cover'(\tnode) = R_\tnode\) for every node \(\tnode\) has width \(1\).

	The relations \(R_\tnode\) can be computed using at most \(\lenof{\cover(\tnode)}\) many \(\opnameJoin\) operations, followed by a \(\opnameProjection[\freeof{\query_\tnode}]\) operation.
	The arrays representing the relations in \(\cover(\tnode)\) have size at most \(\bigO(\IN)\) and since \((\tree,\bag,\cover)\) has width \(\ghw(q)\), the relation \(R_\tnode\) and the intermediate results in its computation have size at most \(\bigO(\IN^{\ghw(q)})\).
	The computation can thus be carried out with work and space \(\bigO(\IN^{\ghw(q)(1+\varepsilon)})\) thanks to \autoref{result:eval-query-plan}.

        Statement (a) then follows by \autoref{result:acyclic} (a).

        Statement (b) follows by  \autoref{result:acyclic} (b) and the observation that $\query'$ is free-connex acyclic if it is applied to a free-connex generalised hypertree decomposition. This can easily be seen with the original definition of free-connex acyclicity in \cite[Definition 22]{DBLP:conf/csl/BaganDG07}, requiring that the free variables of the query induce a connected subtree of the join tree.
\end{proof}

\subsection{Weakly Worst-Case Optimal Work for Joins Queries}\label{section:ctp:worst-case-natural-joins}
This section is concerned with the evaluation of  \emph{join queries}\footnote{Following tradition, we denote join queries in this subsection slightly differently from the previous subsection.} $q = R_1 \Join \ldots \Join R_m$ over some schema $\schema = \{R_1,\ldots, R_m\}$ with attributes $\attrof{q} = \bigcup_{i=1}^m \attrof{R_i}$.
It was shown in \cite{DBLP:journals/siamcomp/AtseriasGM13} that $|\queryresult{\query}{\db}| \le \agmbound$ holds for every database \(\db\) and that this bound is tight for infinitely many databases \db (this is also known as the AGM bound).
Here $\agmbound$ denotes $\agmterm$, where $x_1,\ldots,x_m$ is a fractional edge cover of $\query$ defined as a solution of the following linear program.
\begin{multline*}\label{def:fractional-edge-cover-lp}
	\textnormal{minimise } \textnormal{$\sum_{i=1}^{m} x_i$} \textnormal{ subject to}	\sum_{i: \attA\in \attrof{R_i}} x_i \ge 1 \textnormal{ for all } \attA\in \attrof{q}\\ \textnormal{ and }	x_i \ge 0	\textnormal{ for all } 1\le i\le m
\end{multline*}

There exist sequential algorithms whose running time is bounded by the above worst-case size of the query result (and the size of the database), i.e., by $\bigO(\agmbound+|D|)$ \cite{NgoPRR18}. They  are called \emph{worst-case optimal}. However, these algorithms are not captured by query plans and therefore their parallel constant-time counterparts cannot be obtained by a direct application of \autoref{result:eval-query-plan}. 

We say that a join query $q$ has \emph{weakly worst-case optimal} \ctpalgos, if, for every $\varepsilon>0$, there is a \ctpalgo that evaluates $q$ with work $(\agmbound+|D|)^{1+\varepsilon}$. 
In this subsection, we show that join queries indeed have weakly worst-case optimal \ctpalgos, if the arrays representing the relations are suitably ordered. In the dictionary setting the latter requirement can actually be dropped, since the arrays can be appropriately sorted in a padded fashion and compacted, within the overall work and space bounds.

Before we formally state this result, we fix some notation. 	Let \(\attsetX = (\attA_1,\ldots,\attA_k)\) be an ordered list of all attributes occurring in \(\query\).
	 For each \(i\in\{1,\ldots,k\}\), $\attsetZ_i$ denotes the subsequence of $\attsetX$ consisting of all attributes of \(\attrof{R_i}\).

\begin{thm}\label{result:worst-case-eval-ordered}
	For every $\varepsilon > 0$, $\padeps>0$ and join query $\query = R_1 \Join \ldots \Join R_m$, there is a \ctpalgo that, given $\padeps$-compact arrays $\relarr{R}_1,\ldots,\relarr{R}_m$ representing the relations \(R_1,\ldots,R_m\) concisely, where each $\relarr{R}_i$ is $\attsetZ_i$-ordered, computes $\queryresult{\query}{\db}$ represented by a $\padeps$-compact array and requires $\bigO\big((\agmbound + \lenof{\db})^{1+\varepsilon}\big)$ work and space on an Arbitrary CRCW PRAM.
\end{thm}
In the dictionary setting, we can drop the requirement that the arrays are suitably sorted, since they can be sorted with work $\bigO(|D|^{1+\varepsilon})$.
\begin{cor}\label{result:worst-case-eval-dict}
	For every $\varepsilon > 0$, $\padeps>0$ and join query $\query = R_1 \Join \ldots \Join R_m$, there is a \ctpalgo that, given $\padeps$-compact arrays $\relarr{R}_1,\ldots,\relarr{R}_m$ representing the relations \(R_1,\ldots,R_m\) concisely, in the dictionary setting, computes $\queryresult{\query}{\db}$ represented by a $\padeps$-compact array and requires $\bigO\big((\agmbound + \lenof{\db})^{1+\varepsilon}\big)$ work and space on an Arbitrary CRCW PRAM.
\end{cor}

\begingroup%
\newcommand{\relationindiceswithatt}[2][i]{#1: \attA_{#2}\in \attrof{R_{#1}}}
\newcommand{\deltaval}[1][]{\frac{1}{8#1}\varepsilon}%
\newcommand{\infactoreps}{\frac{1}{4}\varepsilon}%
\newcommand{\minrelsizeatt}[2][R_i]{\mathnotation{\min_{\relationindiceswithatt{#2}}|#1|}}

A \ctpalgo can proceed, from a high-level perspective, similarly to the sequential attribute elimination join algorithm, see e.g.\ \cite[Algorithm~10]{ABLMP21}.
The core of this algorithm is captured by the following result.
\begin{propC}[{\cite[Algorithm~10 and Proposition~26.1]{ABLMP21}}]\label{cited-result:aejoin-correctness}
	For a join query \(\query=R_1\Join\ldots\Join R_m\) with attributes \(\attsetX = (\attA_1,\ldots,\attA_k)\) and each database \(\db\), consider the family \(\left(L_j\right)_{1\le j\le k}\) of sets \(L_j\), which are inductively defined as follows:
	\begin{itemize}
			\item \(L_1 = \bigcap_{\relationindiceswithatt{1}} \proj{R_i}{\attA_1}\) and,
			\item for each $j$ with  \(1 < j \le k\), $L_j$ be the union of all relations \[V_t = \{t\} \times \bigcap_{\relationindiceswithatt{j}}\nolimits \proj{\sjoin{R_i}{\{t\}}}{\attA_j}\] for each $t\in L_{j-1}$.
	\end{itemize}
	Then the relation \(L_k\) is the query result \(\queryresult{\query}{\db}\).
\end{propC}

We will also use the following inequalities proven in \cite{ABLMP21} in our complexity analysis.
\begin{lemC}[{\cite[pp.\ 228-229]{ABLMP21}}]\label{cited-result:aejoin-inequalities}
	Let \(\query = R_1\Join\ldots\Join R_m\) be a join query with attributes \(\attsetX = (\attA_1,\ldots,\attA_k)\) and \(x_1,\ldots,x_m\) be a fractional edge cover of \(\query\).
	Furthermore, let \(L_1,\ldots,L_{k}\) be defined as in \autoref{cited-result:aejoin-correctness}.
	For every database \(\db\) it holds that
	\begin{enumerate}[(a)]
		\item \(\minrelsizeatt{1} \le \agmbound\), and
		\item \(\sum_{t\in L_{j-1}} \minrelsizeatt[\sjoin{R_i}{\{t\}}]{j} \le \agmbound\) for all \(j\in\range[2]{k}\).
	\end{enumerate}
\end{lemC}

\begin{proof}[Proof for \autoref{result:worst-case-eval-ordered}]

	In a nutshell, the algorithm computes iteratively, for increasing $j$ from $1$ to $k$, the relations $L_j$ defined in \autoref{cited-result:aejoin-correctness} and finally outputs \(L_k\).\footnote{We note that the algorithm cannot be cast as a query plan, since it relies on a low-level grouping operation to achieve the desired bounds. Thus, \autoref{result:eval-query-plan} is not applicable.}
	The correctness follows immediately thanks to \autoref{cited-result:aejoin-correctness}.

	Further on, we discuss how the relations $L_j$ can be computed and that \(|\arr{L}_j| \le (1+\padeps)\agmbound\)
	holds for every $j$, where $\arr{L}_j$ is the array representing $L_j$ concisely. Finally, the result array  $\arr{L}_k$ can be $\lambda$-compacted within the required bounds thanks to \autoref{result:alg-compact}.

Let $\attsetX$ and the $\attsetZ_i$ be defined as before.
	Let furthermore $\attsetX_j = (\attA_1,\ldots,\attA_j)$, for each $j\in\range{k}$, be the prefix of \(\attsetX\) up to attribute \(\attA_j\).
         By $\relsof{j}$ we denote the set of indices $i$, for which $\attA_j\in \attrof{R_i}$.
	For each \(j\in\range{k}\), and \(i\in \range{m}\), let \(\attsetY_{i,j}\) be the subsequence of \(\attsetX_j\) consisting of attributes in \(\attrof{R_i}\).

        Without loss of generality, we assume  \(\padeps < \frac{1}{2}\).

	During an initialisation phase the algorithm computes \(\attsetY_{i,j}\)-ordered, fully linked arrays $\arr{P}_{i,j}$, for each   \(j\in\range{k}\), and \(i\in \range{m}\), that represent $\proj{R_i}{\attsetY_{i,j}}$ concisely, respectively.
The order is inherited, since the arrays \(\relarr{R}_i\) are sorted according to \(\attsetY_{i,k}\). It thus suffices, for each \(i\), to establish full links in \(\relarr{R}_i\) and to apply \(\opnameProjection[\attsetY_{i,j}]\) for up to \(m\) values of \(j\).
	By doing this in decreasing order of $j$, links from each tuple in  $\arr{P}_{i,j}$ to its projection in  $\arr{P}_{i,j-1}$ can be established.\footnote{If \(\attA_j\) does not occur in \(\attsetY_{i,j}\), then \(\arr{P}_{i,j-1}\) is just \(\arr{P}_{i,j}\) and no computation is necessary.} In the special case of \(\attsetY_{i,j}\) consisting solely of \(\attA_j\), all tuples are linked to the first cell in \(\arr{P}_{i,j-1}\).

	This initial phase requires work and space $\bigO(\sum_{i=1}^m\lenof{\relarr{R}_i}^{1+\varepsilon}) = \bigO(|D|^{1+\varepsilon})$, for each $i$, thanks to \autoref{result:predsucc} and \autoref{alg:project} (b).
	The arrays \(\arr{P}_{i,j}\) have size \(\bigO(|D|)\).

	The computation of \(L_1\) is straightforward: The algorithm picks an~\(i\) for which \(\attA_1\in\attrof{R_i}\) holds and initialises \(\arr{L}_1\) as a copy of \(\arr{P}_{i,1}\).
	Then it computes the semijoins of \(\arr{L}_1\) with all remaining \(\arr{P}_{k,1}\) for which \(\attA_1\in\attrof{R_k}\).
	This requires $\bigO(|D|^{1+\varepsilon})$ work and space thanks to \autoref{alg:semijoin} (b).
	Furthermore, the output array $\arr{L}_1$ has size at most $|D|$.
	Thus, compacting $\arr{L}_1$ with \(\opnameCompact\) yields an array representing $L_1$ of size at most \((1+\padeps)\lenof{L_1}\) and requires work and space \(\bigO(|D|^{1+\varepsilon})\) thanks to \autoref{result:alg-compact}.
	Clearly, $|L_1| \le \minrelsize[\proj{R_i}{\attA_1}] \le \minrelsize \le \agmbound$
	thanks to \autoref{cited-result:aejoin-inequalities}.
	Thus, the size of the compacted array \(\arr{L}_1\) is bounded by $(1+\padeps)\agmbound$.

	To compute \(L_j\), for $j > 1$, the algorithm operates in two phases: a grouping phase and an intersection phase.
	In both phases only those input relations \(R_i\) with \(\attA_j \in \attrof{R_i}\) participate.
	Note that for these \(i\) we have in particular that \(\attA_j\) occurs in \(\attsetY_{i,j}\) (but not in \(\attsetY_{i,j-1}\)).

	Towards the grouping phase, observe that the tuples in the arrays $\arr{P}_{i,j}$ are grouped by $\attsetY_{i,j-1}$ since they are even \(\attsetY_{i,j}\)-ordered.
	Furthermore, note that the group in $\arr{P}_{i,j}$ for a tuple $t\in L_{j-1}$ is essentially a list of the values in $\proj{\sjoin{R_i}{\{t\}}}{\attA_j}$ annotated with $t[\attsetY_{i,j-1}]$.
	Similar to our algorithm for the \(\opnameJoin\) operation, each group in \(\arr{P}_{i,j}\) is compacted and each \propertuple $t\in L_{j-1}$ is augmented by links to (the first and last tuple of) its group in $\arr{P}_{i,j}$ as follows.
	\begin{enumerate}
		\item Determine the minimum and maximum indices of each group in $\arr{P}_{i,j}$ with $|\arr{P}_{i,j}|$ processors: processor $p$ checks whether $\arr{P}_{i,j}[p].t$ is a \propertuple and differs from its predecessor (resp.\ successor) with respect to attributes $\attsetY_{i,j-1}$.
		The representative in $\arr{P}_{i,j-1}$ is augmented with these indices.
		Since $\arr{P}_{i,j}$ is \(\attsetY_{i,j}\)-ordered, the \propertuples between the minimum and maximum assigned to a \propertuple \(t\) in \(\arr{P}_{i,j-1}\) are then precisely the tuples $t'$ with $t'[\attsetY_{i,j-1}] = t$.

		For every \propertuple $t$ in $\arr{P}_{i,j-1}$ let $\arr{G}_{t,i,j}$ denote the (sub)array of the group for $t$ in $\arr{P}_{i,j}$ for every $1\le i \le m$.
		\item Compact each group \(\arr{G}_{t,i,j}\) in parallel using \(\opnameCompact[\padeps']\) with $\padeps' = \min\{\frac{1}{3},\frac{\padeps}{3}\}$ and parameter \(\delta = \min\{\frac{1}{3},\frac{\varepsilon}{3}\}\) (the choice of $\padeps'$ is required for the complexity bounds in the intersection phase).
		To this end, the algorithm creates, for each \propertuple \(t\) in $\arr{P}_{i,j-1}$ a task description \(d_t\) for \(m_t = \lenof{\arr{G}_{t,i,j}}^{1+\delta}\) processors with links to \(\arr{G}_{t,i,j}\) and the tuple \(t\) (in \(\arr{P}_{i,j-1}\)) itself.

		Invoking the algorithm guaranteed by \autoref{result:task-scheduling} yields a schedule of size
		\begin{multline*}
			(1+\padeps)\sum_{t \in \arr{P}_{i,j-1}} m_t
			= (1+\padeps)\sum_{t \in \arr{P}_{i,j-1}} \lenof{\arr{G}_{t,i,j}}^{1+\delta}\\
			\le (1+\padeps)\Big( \sum_{t \in \arr{P}_{i,j-1}} \lenof{\arr{G}_{t,i,j}}\Big)^{1+\delta}
			= (1+\padeps) \lenof{\arr{P}_{i,j}}^{1+\delta}.
		\end{multline*}

		Thanks to \((1+\delta)^2 \le (1+\varepsilon)\) computing the schedule requires work and space \(\bigO(|D|^{1+\varepsilon})\).
		Moreover, the arrays \(\arr{G}_{t,i,j}\) can be compacted with \(\opnameCompact[\padeps']\), and the tuples in \(\arr{P}_{i,j}\) can be augmented by links to the \(\padeps'\)-compact arrays within the same bounds.

		\item Finally, links are established from \propertuples in $\arr{L}_{j-1}$ to their respective cells in $\arr{P}_{i,j-1}$ for every $1\le i\le m$ as follows.
		Since \(L_{j-1}\) is a relation over \(\attsetX_{j-1}\) which is a superset of \(\attsetY_{i,j-1}\), it is straightforward to obtain a (possibly) non-concise representation of \(\proj{Y_{i,j-1}}{L_{j-1}}\) with mutual links to/from \(\arr{L}_{j-1}\) with work and space \(\bigO(\lenof{\arr{L}_{j-1}})\).
		Then, for each \(1\le i\le m\), \opnameSearchTuples is applied to find, for each \propertuple \(t\) in \(\arr{L}_{j-1}\), the respective tuple \(t[\attsetY_{i,j-1}]\) in \(\arr{P}_{i,j-1}\).
		Each of the (constantly many) applications of \opnameSearchTuples requires work and space $\bigO(|\arr{L}_{j-1}| \cdot |\arr{P}_{i,j-1}|^{\varepsilon}) = \bigO((\agmbound)\cdot |D|^{\varepsilon})$, because $\arr{P}_{i,j-1}$ is fully ordered (w.r.t.\ \(\attsetY_{i,j-1}\)) and fully linked.
		These links, together with the links in $\arr{P}_{i,j-1}$ to the \(\padeps'\)-compact group arrays $\arr{G}_{t,i,j}$ allow to determine, for each tuple $t$ in $L_{j-1}$, the smallest group array $\arr{G}_{t',j}$ with $t' = t[\attsetY_{i,j-1}]$ among the group arrays $\arr{G}_{t',i,j}$ with work $\bigO(|\arr{L}_{j-1}|) = \bigO(\agmbound)$.
		Note that, in general, multiple tuples in \(\arr{L}_{j-1}\) may be linked to the same group.

		For tuples \(t\in L_{j-1}\), we will write \(\arr{G}_{t,i,j}'\) and \(\arr{G}_{t,j}'\) to denote the groups \(\arr{G}_{t',i,j}\) and \(\arr{G}_{t',j}\), respectively, where \(t' = t[\attsetY_{i,j-1}]\).
	\end{enumerate}
	This phase requires work and space \(\bigO((\agmbound + |D|)\cdot |D|^{\varepsilon})\) in total.

	Observe that a group array \(\arr{G}_{t,i,j}'\) for some \(t\in L_{j-1}\) represents \(\sjoin{\proj{R_i}{Y_{i,j}}}{\{t\}}\) concisely. Since \(t\) determines the values for all attributes except \(\attA_j\), \(\arr{G}_{t,i,j}'\) can also be viewed as an array representing \(\proj{\sjoin{R_i}{\{t\}}}{\attA_j}\) concisely, where every value is annotated with \(t[\attsetY_{i,j-1}]\).
	In particular, a group \(\arr{G}_{t,j}'\) represents the smallest set \(\proj{\sjoin{R_i}{\{t\}}}{\attA_j}\), for \(i\in\relsof{j}\).

	In the intersection phase, the algorithm computes the union of all
        \[V_t=\{t\} \times \bigcap_{\relationindiceswithatt{j}}\nolimits \proj{\sjoin{R_i}{\{t\}}}{\attA_j}\]
        for $t\in  L_{j-1}$ in parallel.
	It proceeds in two steps.
	First, it joins all \(t\in L_{j-1}\) with all \propertuples in \(\arr{G}_{t,j}'\), i.e., for all \(t\in L_{j-1}\), it computes the relation \(\{t\}\times \proj{\sjoin{R_i}{\{t\}}}{\attA_j}\) for the \(i\) minimizing the size of \(\proj{\sjoin{R_i}{\{t\}}}{\attA_j}\).

	Then, in the second step, it performs semijoins with all\footnote{For each tuple, $t$, the semijoin is redundant, for some $\ell$. But it is convenient to do it in this way.} group arrays \(\arr{G}_{t,\ell,j}\) to effectively compute the intersection of the sets \(\{t\}\times \proj{\sjoin{R_\ell}{\{t\}}}{\attA_j}\) for all \(\ell\).

	For the first step, the algorithm has to assign $|\arr{G}_{t,j}|$ processors\footnote{
		We note that this corresponds (up to a logarithmic factor) to the running time stated by \cite[Claim~26.3]{ABLMP21} and required for the complexity analysis of \cite[Algorithm~10]{ABLMP21}; an implementation for the operation described in \cite[Claim~26.3]{ABLMP21} is, e.g., the Leapfrog-Join \cite{DBLP:conf/icdt/Veldhuizen14}, \cite[Proposition~27.10]{ABLMP21}.
              }
               to each tuple $t\in L_{j-1}$; recall that $\arr{G}_{t,j}$ is the smallest group array for \(t\).
	To this end, it creates, for each \(t\in L_{j-1}\), a task description \(d_t\) for \(m_t = \lenof{\arr{G}_{t,j}}\) processors and, as usual, with links to the group array \(\arr{G}_{t,j}\) and the cell of \(t\) in \(\arr{L}_{j-1}\).
	Since \(\arr{L}_{j-1}\) already contains links to the group arrays, these task descriptions can be computed with work and space linear in \(\lenof{\arr{L}_{j-1}}\) and stored in an array of size \(\lenof{\arr{L}_{j-1}}\).
	Thanks to \autoref{result:task-scheduling} it is then possible to compute a schedule of size
	\begin{align}
		M = (1+\padeps')\sum_{t\in L_{j-1}} m_t \label{worst-case-dict-def-M}
	\end{align}
	in constant time and with work and space \(\bigO\big(\lenof{\arr{L}_{j-1}}^{1+\varepsilon} + M^{1+\varepsilon}\big)\).

	To finalise the first step, the algorithm then allocates an array \(\arr{L}_{j}\) of size \(M\) and uses \(M\) processors to write, for all \(t\in L_{j-1}\), the join of \(\{t\}\) with all tuples in \(\arr{G}_{t,j}\) into \(\arr{L}_{j}\): If the \(m\)-th overall processor is the \(\ell\)-th processor assigned to \(t\), it writes \(\join{t}{\arr{G}_{t,j}[\ell].t}\) into \(\arr{L}_{j}[m]\), if \(\arr{G}_{t,j}[\ell]\) is inhabited.

	For the second step of the intersection phase, it suffices to compute the semijoins of \(\arr{L}_{j}\) with \(\arr{P}_{i,j}\) for all \(i\) with \(\attA_j\in \attrof{R_i}\).
	Because the arrays \(\arr{P}_{i,j}\) are \(\attsetY_{i,j}\)-ordered and have size \(\bigO(|D|)\) the second step required work and space $\bigO(M\cdot|D|^\varepsilon)$ thanks to \autoref{alg:semijoin} (b).%

	To establish the desired work and space bounds for the intersection phase as well as upper bounds for the size of \(\arr{L}_j\), we prove that \(M \le (1+\padeps)\agmbound\).
	As pointed out above, an array $\arr{G}_{t,i,j}$ represents $\sjoin{\proj{R_i}{\attsetY_{i,j}}}{\{t\}} = \proj{\sjoin{R_i}{\{t\}}}{\attsetY{i,j}}$ concisely.
	Further, we have that $|\sjoin{\proj{R_i}{\attsetY_{i,j}}}{\{t\}}| = |\proj{\sjoin{R_i}{\{t\}}}{\attsetY_{i,j}}| \le |\sjoin{R_i}{\{t\}}|$.
	Since the groups where compacted using \opnameCompact[\padeps'] in the grouping phase, we can conclude that $(1+\padeps')\lenof{\sjoin{R_i}{\{t\}}}$ is an upper bound for $\lenof{\arr{G}_{t,i,j}}$.
	In particular, for the group arrays $\arr{G}_{t,j}$ of minimal size, we have
	\begin{align}
		\lenof{\arr{G}_{t,j}} \le (1+\padeps')\minrelsize[\sjoin{R_i}{\{t\}}].\label{worst-case-dict-size-min-G}
	\end{align}
	Therefore, we have
	\begingroup%
	\allowdisplaybreaks%
	\begin{align*}
		M &\stackrel{(\ref{worst-case-dict-def-M})}{=} (1+\padeps')\sum_{t\in L_{j-1}} m_t\\
		&= (1+\padeps')\sum_{t\in L_{j-1}} \lenof{\arr{G}_{t,j}}\\
		&\stackrel{(\ref{worst-case-dict-size-min-G})}{\le} (1+\padeps')\sum_{t\in L_{j-1}} (1+\padeps')\minrelsize[\sjoin{R_i}{\{t\}}]\\
		&\stackrel{(3)}{\le} (1+\padeps) \sum_{t \in L_{j{-}1}} \minrelsize[\sjoin{R_i}{\{t\}}]\\
		&\stackrel{(4)}{\le} (1+\padeps)\agmbound
	\end{align*}
	\endgroup
	where~(3) holds because \((1+\padeps')^2 \le (1+\padeps)\); and~(4) thanks to \autoref{cited-result:aejoin-inequalities} (b).
	Therefore, altogether \(M = \lenof{\arr{L}_j} \le (1+\padeps)\agmbound\).
        This completes the description of the computation of the arrays $\arr{L}_j$, for $j>1$.
\end{proof}
\endgroup
\section{Query Evaluation in Other Settings}\label{section:ctp:othersettings}

In this section we discuss query evaluation in other settings than the
dictionary setting.

We first consider the general setting, where data
can only be compared with respect to equality. In this setting it
turns out that the naive algorithms for the evaluation of database
operations are already work-optimal: for \opnameSelection the work is
linear and for the other operators quadratic.  We also show that with
quadratic work in its size, a database can be translated into the
dictionary setting in constant parallel time. Therefore, all results
of \autoref{section:ctp:query-evaluation} apply to the general setting with an extra additive term
$\bigO(\lenof{\db}^2)$ (or $\bigO(\IN^2)$ --- cf.\ \autoref{table:overview-main-results}, left column).

Then we consider the setting, where data can be compared with respect
to a linear order. Without any further assumption, the linear order
does not give us better algorithms while the lower bounds of the general
setting cease to exist. However, with a somewhat stronger assumption
it is possible to translate a given database into the dictionary setting
in constant parallel time with work
$\bigO(\lenof{\db}^{1+\varepsilon})$ and in this way most bounds of
\autoref{section:ctp:query-evaluation} for query evaluation apply (cf.\ \autoref{table:overview-main-results}, middle column). We refer  to this setting
as the  \emph{attribute-wise
ordered} (or \emph{a-ordered}) setting.

\subsection{Query Evaluation in the General Setting}\label{subsection:ctp:generalsetting}

We first give the lower bound mentioned above.

\begin{thm}\label{result:dbops-lowerbound-general}
	 For the operations \opnameProjection, \opnameSemiJoin, \opnameJoin,
         \opnameDifference, \opnameUnion every \ctpalgo requires
         \(\Omega(n^2)\) work for infinitely many input databases in
         the general setting. 	Here \(n\) is the size of the input arrays.

         	The lower bound for \opnameProjection even holds for
                binary input relations and the others for unary input
                relations, and for compact and concise input arrays.
         In         the case of \opnameJoin the infinitely many
         databases can be chosen such that
         the join result has at most one tuple, for the join at hand.
      \end{thm}
      \begin{proof}
        	We first prove the lower bound for \opnameSemiJoin. To
                this end, we define a family of databases \(\db_n\) and show that the elemental operation \(\opnameEqual\) has to be invoked at least \(n^2\) times to evaluate \(\sjoin{R}{S}\) correctly.
	For any \(n > 0\), let \(R_n = \{2,\ldots,2n\}\) be the set of all even numbers from \(2\) to \(2n\) and \(S_n = \{1,\ldots,2(n-1)+1\}\) be the set of all odd numbers from \(1\) to \(2(n-1)+1\).
	Note that \(\lenof{R_n} = \lenof{S_n} = n\) and let \(\db_n\)
        be the database over schema \(R,S\) with \(\db_n(R) = R_n\)
        and \(\db_n(S) = S_n\). Let $\arr{R}_n$ and $\arr{S}_n$ be
        compact arrays representing $R_n$ and $S_n$ concisely.
	The query result of \(\sjoin{R}{S}\) on \(\db_n\) is empty.

	Assume for the sake of a contradiction that there is a
        \ctpalgo for the general setting that evaluates
        \(\sjoin{R}{S}\) on \(\db_n\) with less than \(n^2\)
        invocations of \(\opnameEqual\).
	Then there are \(i,j\) such that \(\opEqual{R_n,i,1;S_n,j,1}\) is never invoked.
	Furthermore, by construction any invocation of
        \(\opEqual{R_n,i',1;S_n,j',1}\) returns false on \(\db_n\), and
        invocations \(\opEqual{R_n,i_1,1; S_n,i_2,1}\) or
        \(\opEqual{R_n,i_1,1; S_n,i_2,1}\) return true if and only if \(i_1
        = i_2\).
	Now consider the database \(\db_n'\) that is defined like
        \(\db_n\) except that the \(i\)-th value of \(R_n\) and the
        \(j\)-th value of \(S_n\) are set to \(2n+1\). Let $\arr{S}_n'$ be
        a compact array representing $S_n'$ concisely.

	Since  \(\opEqual{R_n,i_1,1; S_n,i_2,1}\) is never invoked,
        every invocation of \(\opnameEqual\) still yields the same
        result as before. As furthermore the sizes of $\arr{S}_n$ and
        $\arr{S}_n'$ are the same, we can conclude that the
        computation on \(\db_n'\) does not differ from the
        computation on \(\db_n\) and the algorithm outputs the empty query result.
	This is a contradiction because the query result of
        \(\sjoin{R}{S}\) on \(\db'_n\) is \(\{2n+1\}\).

        For \opnameJoin,
         \opnameDifference, and \opnameUnion exactly the same construction can be
        used. For \opnameDifference, the contradiction results from
        $2n+1$ being still present in the output result for
        $\opDifference{\arr{R}_n, \arr{S}_n'}$ and for \opnameUnion, the contradiction results from
        $2n+1$ being present twice in the output array for
        $\opUnion{\arr{R}_n, \arr{S}_n'}$.

        For \opnameProjection, a similar construction can be used with a binary relation containing
        tuples of the form $(1,n+1),\ldots,(n,2n)$. If the algorithm does
        not invoke any operations of the form
        \(\opEqual{R_n,i_1,1, R_n,i_2,1}\), for some
	      $i_1\neq i_2$ (and $i_1,i_2\in\range{n}$), then the
        output array for the database in which $(i_2,i_2+n)$ is
        replaced by $(i_1,i_2+n)$ has a duplicate entry for $i_1$.
      \end{proof}

      As mentioned before, the lower bounds of
      \autoref{result:dbops-lowerbound-general} can easily be matched by
      straightforward algorithms.

      \begin{thm}\label{result:alg-dbops-general}
        There are \ctpalgos for the following operations that, given arrays
  \arr{R} and \arr{S}, representing \(R\) and \(S\) concisely, have the following bounds on an Arbitrary CRCW
  PRAM.
  \begin{enumerate}[(a)]
  \item Work and space $\bigO(\lenof{\arr{R}})$ for $\opSelection[\attA=x]{\arr{R}}$.
  \item Work and space $\bigO(\lenof{\arr{R}}^2)$ for \opProjection[\attsetX]{\arr{R}}.
  \item Work and space $\bigO(\lenof{\arr{R}}\cdot \lenof{\arr{S}})$ for \opSemiJoin{\arr{R},\arr{S}}
  and \opDifference{\arr{R},\arr{S}}.
  \item Work and space $\bigO(\lenof{\arr{R}}\cdot \lenof{\arr{S}})$ for \opJoin{\arr{R},\arr{S}}.
  \item Work and space $\bigO((\lenof{\arr{R}}+\lenof{\arr{S}})^2)$ for \opUnion{\arr{R},\arr{S}}.
  \end{enumerate}

       \end{thm}
       \begin{proof}
          The algorithms for \opnameSelection,  \opnameSemiJoin,
          \opnameDifference and \opnameJoin are entirely
          straightforward.
          The algorithm for  \opnameProjection, first
          adapts every tuple in the obvious way and then removes
          duplicates by using one processor for each pair $(t,t')$ of tuples
          and removing (that is: making its cell uninhabited) the
          tuple with the larger index, in case $t=t'$.
         The algorithm for \opnameUnion first concatenates the two
         arrays and then removes duplicates in the same fashion.
       \end{proof}

 As mentioned before, the results of
 \autoref{section:ctp:query-evaluation} translate into the general
 setting by  adding an additive term
$\bigO(\IN^2)$. However, these results are not obtained by applying
the above algorithms but rather by transforming the database at hand
into a database in the dictionary setting.

To turn a database \(\db\) in the general or ordered setting into a representation for the dictionary setting, the data values occurring in \(\db\) have to be \enquote{replaced} by small values; that is, numbers within \(\range{\cval\lenof{\db}}\) for some constant \(\cval\).
More precisely, we are interested in an injective mapping \(\keyf\colon\adomof{\db}\to \range{\cval\lenof{\db}}\).
Here \(\adomof{\db}\) is the set of all data values occurring in (some relation of) \(\db\).
The idea is then to compute a database \(\db' = \keyf(\db)\) with small values, evaluate a given query \(\query\) on \(\db'\), and then obtain the actual result by applying the inverse of \(\keyf\) to the query result \(\queryresult{\query}{\db'}\).

Since a PRAM cannot directly access the data values in the general and ordered setting, we will use a data structure that allows to map a token representation for a data value \(a\) to \(\keyf(a)\) and vice versa.
Recall that token representations are not unique, i.e.\ there may be
multiple token representations for the same data value \(a\).
All these representations have to be mapped to the same \emph{key} \(\keyf(a)\).

We define a \emph{dictionary} as a data structure that allows a single processor to carry out the following operations for every relation \(R\) of the database in constant time.
\begin{itemize}
	\item \(\opKey{R,i,j}\) returns \(\keyf(a)\) where \(a\) is the value of the \(j\)-th attribute of the \(i\)-th tuple in relation \(R\).
	\item \(\opKeyOut{k,i,j}\) outputs the value \(\keyf^{-1}(k)\) as the \(j\)-th attribute of the \(i\)-th output tuple.
\end{itemize}
Here \(\keyf\) is an arbitrary but fixed injective function \(\keyf\colon\adomof{\db}\to \range{\cval\lenof{\db}}\) for some constant \(\cval\).

Observe that, using \(\opKey{R,\cdot,\cdot}\) an array representing a relation \(R\) in the general setting can be translated into an array representing \(\keyf(R)\) in the dictionary setting in constant time with linear work.

The following result states that dictionaries can be computed in
constant time with quadratic work, in the general setting.

\begin{lem}\label{result:general->dict}
	There is a \ctpalgo that, given a
        database \(\db\)  in the general
        setting, computes a dictionary with work and space
        $\bigO(\lenof{\db}^2)$ on a Common CRCW PRAM.
\end{lem}

\begin{proof}
	The idea is to compute an array \(\arr{D}\) of size \(\bigO(\lenof{\db})\), whose cells contain all (token representations of) data values occurring in \(\db\).
	This array will be the underlying data structure of the dictionary.
	We note that \(\arr{D}\) is (most likely) not concise.
	The key \(\keyf(a)\) for a data value \(a\) is then the index of a fixed representative of \(a\) in  \(\arr{D}\).
	In the following we will detail how to compile \(\arr{D}\), and how the operations \(\opnameKey\), and \(\opnameKeyOut\) can be implemented on top of it.

	For each input relation \(R\) and \(j\in\range{\ar(R)}\) array \(\arr{B}_{R,j}\) of size \(\lenof{R}\) is computed using \(\lenof{R}\) processors as follows: Processor \(i\) writes the token \((R,i,j)\) to cell \(\arr{B}_{R,j}[i]\).
	Thus, \(\arr{B}_{R,j}\) is a compact array that contains (token representations for) all data values occurring for the \(j\)-th attribute in any tuple of \(R\).
	Concatenating all arrays \(\arr{B}_{R,j}\), where \(R\) ranges over all input relations and \(j\) over \(\range{\ar(R)}\) yields the array \(\arr{D}\).
	The size of \(\arr{D}\) is clearly bounded by \(\sum_{R\in\schema}\ar(R)\cdot\lenof{R}\le \cval\lenof{\db}\) where \(\cval = \sum_{R\in\schema}\ar(R)\).
	This step requires work and space
        \(\bigO(\lenof{\db})\).

	Next, the algorithm can identify duplicates in \(\arr{D}\)
        similarly as in \autoref{result:alg-dbops-general} and link every cell \(\arr{D}[i]\) to a representative cell \(\arr{D}[j]\) which contains a token representation for the same data value as \(\arr{D}[i]\).
	This requires work \(\bigO(\lenof{\db}^2)\) and space \(\bigO(\lenof{\db})\).

	It remains to show that \(\opnameKey\) and \(\opnameKeyOut\) can be carried out in constant time by a single processor given \(\arr{D}\).
	For \(\opKey{R,i,j}\) we observe that a single processor can determine the index \(\ell\) of the token \((R,i,j)\) in \(\arr{D}\): Indeed, \((R,i,j)\) has been stored in \(\arr{B}_{R,j}[i]\), \(\arr{D}\) is the (compact) concatenation of constantly many \(\arr{B}_{S,m}\), and a single processor can obtain the size of any relation (and, hence, \(\arr{B}_{S,m}\)) using \(\opnameNumTuples\).
	Given \(\ell\) it is then easy to determine the representative
        cell \(\arr{D}[k]\) for \(\arr{D}[\ell]\) using the links
        established before (possibly \(k = \ell\) holds).
	The index \(k\) is then returned as the unique key for the data value represented by \((R,i,j)\).

	The implementation of \(\opKeyOut{k,i,j}\) is straightforward.
	Indeed, \(k\) is an index of \(\arr{D}\) and \(\arr{D}[k]\) contains a token representation \((R,i_1,j_1)\) for the data value with the unique key \(k\).
	Thus, invoking \(\opOutput{R,i_1,j_1;i,j}\) has the desired effect.
      \end{proof}

      \begin{cor}\label{result:acyclic-etc-general}
        The results of \autoref{result:eval-semijoin-algebra}, \autoref{result:acyclic}, \autoref{result:eval-ghw}
        and \autoref{result:worst-case-eval-dict} hold in the general
        setting with the proviso that an additional term $\bigO(\IN^2)$
        (or $\bigO(\lenof{\db}^2)$ for
        \autoref{result:worst-case-eval-dict}) is added to the stated
        work and space bounds.
      \end{cor}
      \noindent We note that constants occurring in select subqueries
      can be handled as follows. By  using \(\bigO(\IN)\) processors,
      which brute force it can be checked using \opnameEqualConst, and
      \opnameKey whether the constant occurs in the database and to
      which number it is mapped.
      In case no token representation is found, the constant does not
      occur in the database and the subquery result is
      empty. Otherwise the obtained number can be used to evaluate the
      select subquery in a straightforward fashion.

\subsection{Query Evaluation in the Ordered Setting}\label{subsection:ctp:ordered}

We finally consider the ordered setting, where the elemental operation
\opnameLessThan is available to access the database. We note that the
ordered setting does not guarantee arrays that store the tuples of a
relation in any particular order.

The lower bounds from
\autoref{result:dbops-lowerbound-general} do not hold anymore in the
presence of \opnameLessThan, and therefore we hoped that its presence
would allow us  to find algorithms that improve the bounds of the general setting. Unfortunately, we were not able to find such algorithms.

An obstacle is that we do not know whether there is a \ctpalgo
for comparison-based padded sorting on CRCW PRAMs. Such algorithms
exist for the \emph{parallel comparison model} \cite[Section
4]{AzarV87}, but to the best of our knowledge the best deterministic
PRAM algorithms have non-constant time bounds
\cite{DBLP:conf/spaa/HagerupR93,DBLP:conf/esa/ChongR98}. Therefore, it
is not clear how to sort a relation array, in general.

However, if we assume that  suitable
orderings for the database relations \emph{do exist}, the results of \autoref{result:acyclic}, \autoref{result:eval-ghw}
        and \autoref{result:worst-case-eval-dict} can be transferred with
        the same bounds by a reduction to the dictionary setting.

In the \emph{attribute-wise ordered setting} (\emph{a-ordered
  setting}), for each relation $R$ and each attribute $A$ of $R$,
there is some $\attsetX$-ordered array representing $R$, such that $A$ is the
first attribute of $\attsetX$.

\begin{lem}\label{result:aordered->dict}
	For every \(\varepsilon>0\), there is a \ctpalgo that, given a
        database \(\db\) in the attribute-wise ordered setting,
        computes a dictionary with  work and space
        $\bigO(\lenof{\db}^{1+\varepsilon})$ on an Arbitrary CRCW PRAM.
       \end{lem}
       \begin{proof}
         	In the ordered setting an algorithm can essentially proceed similarly to the algorithm for the general setting detailed above.
 	To yield the better upper bound of
        \(\bigO(\lenof{\db}^{1+\varepsilon})\) it would suffice to use
        the algorithm for deduplication for ordered arrays (cf.\ \autoref{result:alg-deduplicate-ordered}), since this is the only operation with a non-linear work bound.
 	However, the array \(\arr{D}\) is not necessarily ordered.

 	We will show how to construct another array \(\arr{C}\) that, like \(\arr{D}\), contains all (token representations of) data values in \(\db\), and is piecewise ordered.
	It is then possible to search, for each cell of \(\arr{D}\), a representative cell in \(\arr{C}\).
	We note that the unordered version \(\arr{D}\) is still required to carry out \(\opnameKey\) in constant time (by a single processor).

 	To construct \(\arr{C}\), the algorithm first computes ordered arrays \(\arr{C}_{R,j}\) which correspond to the arrays \(\arr{B}_{R,j}\).
 	Let \(R\in\schema\) and \(j\in\range{\ar(R)}\). Furthermore, let \(\attsetX_j\) be an ordered list of all attributes of \(R\) that starts with the \(j\)-th attribute \(\attA_j\) of \(R\).
 	The array \(\arr{C}_{R,j}\) can now be derived from the array \(\arr{A}_{R,\attsetX_j}\) as follows with \(\lenof{R}\) processors:
 	Processor \(p\) writes the token representation \((R,i,j)\) for the value of the \(j\)-th attribute of the \(i\)-tuple of \(R\) into cell \(\arr{C}_{R,j}[p]\) if \(i = \arr{A}_{R,\attsetX_j}[p]\).
 	Since \(\arr{A}_{R,\attsetX_j}\) induces the lexicographical order w.r.t.\ \(\attsetX_j\) and \(\attsetX_j\) starts with \(\attA_j\), the array \(\arr{C}_{R,j}\) is fully ordered w.r.t.\ \(\{\attA_j\}\).
 	The array \(\arr{C}\) is then the concatenation of all \(\arr{C}_{R,j}\).
 	Computing \(\arr{C}\) requires linear work and space.

	To determine representatives, fix an arbitrary linear order on the set of all pairs \((R,j)\) where \(R\) is a relation symbol from \(\schema\) and \(j\in\range{\ar(R)}\).
	For each pair \((R,j)\) the algorithm invokes the algorithm of
        \autoref{result:alg-deduplicate-ordered} to find representatives for all data values in the subarray \(\arr{C}_{R,j}\) of \(\arr{C}\).
	Since there are only constantly many pairs \((R,j)\) this can be done in constant time with \(\bigO(\lenof{\db}^{1+\varepsilon})\) work and \(\bigO(\lenof{\db})\) space thanks to \autoref{result:alg-deduplicate-ordered}  and each \(\arr{C}_{R,j}\) being fully ordered.
	The representative for a cell \(\arr{D}[\ell]\) is then the representative cell for \(\arr{D}[\ell]\) of the subarray \(\arr{C}_{R,j}\) for the smallest pair \((R,j)\) for which such a representative cell exists in \(\arr{C}_{R,j}\).

	The operations \(\opnameKey\) and \(\opnameKeyOut\) can then be implemented analogously to the general setting.
	That is, instead of returning an index of \(\arr{D}\), \(\opKey{R,i,j}\) returns the index of the representative cell in \(\arr{C}\) for the cell of \(\arr{D}\) containing \((R,i,j)\), and \(\opKeyOut{k,i,j}\) can use the token representation in \(\arr{C}[k]\) to output the proper value.
       \end{proof}

By combining \autoref{result:eval-semijoin-algebra} with \autoref{result:aordered->dict}, and since $\IN^{1+\varepsilon}$ is dominated by the upper bounds of \autoref{result:acyclic}, \autoref{result:eval-ghw}
        and \autoref{result:worst-case-eval-dict}, we can conclude
        following result.
       \begin{cor}\label{result:queries-a-ordered}
		\hfill
         \begin{enumerate}[(a)]
         \item   For every $\varepsilon>0$ and each query $\query$ of the semijoin algebra there is a \ctpalgo that,
  given a database $\db$, evaluates
	$\queryresult{\query}{\db}$ on an Arbitrary CRCW PRAM with
        work and space $\bigO(\IN^{1+\varepsilon})$ in the
        attribute-wise ordered setting.
        \item  The results of \autoref{result:acyclic}, \autoref{result:eval-ghw}
        and \autoref{result:worst-case-eval-dict} hold in the
        attribute-wise ordered setting with the same upper bounds.
         \end{enumerate}
       \end{cor}
It should be mentioned that a relaxed form of the dictionary setting
where data values can be of polynomial size in $\lenof{\db}$ can be
easily transformed into the a-ordered setting with the algorithm of \autoref{result:sorting}.

\section{Conclusion}\label{section:ctp:conclusion}

This article is meant as a first study on work-efficient \ctpalgos for
query evaluation and many questions remain open. The results are very
encouraging as they show that quite work-efficient \ctpalgos for
query evaluation are possible. In fact, the results give a hint at
what could be a good notion of \emph{work-efficiency} in the context
of constant-time parallel query evaluation. Our impression is that
work-optimality is very hard to achieve in constant time and that
query evaluation should be considered as
work-efficient for a query language, if there are constant-time
parallel algorithms with $\bigO(T^{1+\varepsilon})$ work, for every
$\varepsilon>0$, where $T$ is the best sequential time of an
evaluation algorithm. Of course, it would be nice if this impression
could be substantiated by lower bound results, but that seems to be challenging.

It also remains open whether \(\bigO(1)\)-time parallel versions of more advanced query evaluation algorithms can be implemented.
An intriguing example is the PANDA algorithm \cite{DBLP:journals/theoretics/KhamisNS25}.
A  \(\bigO(1)\)-time parallel version of PANDA would yield significant improvements over our results on conjunctive and natural join queries.
However, while we believe that \ctpalgos for PANDA's operations \cite[Section~5.1]{DBLP:journals/theoretics/KhamisNS25} can be implemented, PANDA also computes statistics, e.g.\ the degree of input values, to decide which operations are invoked.
Due to the limitation of \(\bigO(1)\)-time parallel algorithms (e.g., \autoref{prop:lb-parity}), it is not possible to compute these statistics exactly, within our framework.
It is likely that approximations of these statistics would still be sufficient, but verifying this claim is beyond the scope of this article.

Some results presented in this article can be slightly improved. Indeed, Daniel Albert has shown in his Master's thesis that predecessor and successor links can be computed with work and space $\bigO(n\log (n))$ for arrays of length $n$, improving \autoref{result:predsucc} \cite{Albert24}. The algorithm uses the technique from \cite[Theorem~1]{DBLP:journals/siamcomp/FichRW88} to find the first \(1\) in an array of length $n$ in constant time with work $\bigO(n)$.\footnote{A more precise description can be found on pages 33 and 34 in the appendix of \cite{Albert26}. The description is for an input with $\log n$ items, though.} As a consequence deduplication
(\autoref{result:alg-deduplicate-ordered}), and the evaluation of   \opnameProjection[\attsetX] (\autoref{alg:project}) are possible with work $\bigO(n\log (n))$ (as compared to $\bigO(n^{1+\epsilon})$). However, the lower bound from \autoref{result:dbops-lowerbound} shows that such an improvement is not possible for the operation \(\opnameSemiJoin\).

There remain many questions for future work. Is it possible to improve some of the work bounds of the form $\bigO(T^{1+\varepsilon})$ to $\bigO(T\log(T))$? Are there more lower bound results preventing this? Can the work bounds be improved for particular kinds of join queries?

And of course, we are still interested in the dynamic setting, that is if the database is modified, but auxiliary data can be used.

\section*{Acknowledgements}
\noindent%
	We are grateful to Uri Zwick for clarifications regarding results in \cite{GoldbergZ95} and to Jonas Schmidt and Jennifer Todtenhoefer for careful proofreading.
	We thank Martin Dietzfelbinger for helpful discussions.
	Furthermore, we thank the reviewers of ICDT for many insightful suggestions. Finally, we thank the reviewers of this journal article for their insightful comments. In particular, we are grateful for the suggestion to state \autoref{result:eval-query-plan} as a separate result about the evaluation of query plans, and to derive many of the subsequent results from it.

\bibliographystyle{alphaurl}
\bibliography{main}

\newpage
\appendix

\section{Revisiting Linear Compaction and Approximate Prefix Sums}\label{sec:approx-prefix-sums}

The purpose of this section is to revisit the proof for \autoref{result:lin-compaction} sketched by Goldberg and Zwick \cite{GoldbergZ95} and to analyse the space requirements of the compaction algorithm.

Compaction is closely related to prefix sums, as discussed in \autoref{subsec:aprroxcompact}.
Goldberg and Zwick \cite{GoldbergZ95} showed that approximate prefix sums can be computed in constant time using a CRCW PRAM with polynomial work and space.
This then gives rise to an algorithm for linear compaction as stated in \autoref{result:lin-compaction}.
Before we revisit the algorithm and analyse its required space, let us give a formal definition of approximate prefix sums.

\begin{defi}
	For any \(\padeps > 0\), a sequence \(s_1,\ldots,s_n\) of natural numbers is called a \emph{consistent \(\padeps\)-approximate prefix sums sequence} of a sequence \(a_1,\ldots,a_n\) of natural numbers, if for every \(i\in\range{n}\) the following conditions hold.
	\begin{enumerate}
	\item \(\sum_{j=1}^i a_j \le s_i \le (1+\padeps)\sum_{j=1}^i a_j\)
	\item \(s_i - s_{i-1} \ge a_i\) if \(i>1\) and \(s_i\ge a_i\) if \(i=1\).
	\end{enumerate}
	The first condition ensures that each \(s_i\) is an approximation of the exact prefix sum and the second condition ensures consistency, that is, that the difference between consecutive (approximate) prefix sums is at least as large as the respective number.
\end{defi}
Since we consider only consistent \(\padeps\)-approximate prefix sums we will omit the term \enquote{consistent} from here on out.
We restate the result of Goldberg and Zwick from \autoref{result:lin-compaction}.

\RestateGo{\restateprefixsums}
\prefixsums*

To prove \autoref{result:prefix-sums}, Goldberg and Zwick proceed in three major steps.
Based on an algorithm for \emph{approximate counting} by Ajtai~\cite{DBLP:conf/dimacs/Ajtai93}, they first show that the sum of \(n\) integers can be approximated in constant time with polynomial work and space.
In the second step this is used to compute approximate prefix sums with polynomial work and space.
The third step reduces the required work and space to \(\bigO(n^{1+\varepsilon})\).
Our presentation will proceed analogously.
The required result on approximate counting is the following.

\begin{propC}[{\cite[Theorem~2.1]{DBLP:conf/dimacs/Ajtai93}}]\label{cited-result:counting}
	For every integer \(a > 0\) there is an \ctpalgo that, given an array \(\arr{A}\) whose cells contain either \(1\) or \(0\), computes a number \(k\) satisfying \(\#_1(\arr{A}) \le k \le (1+\padeps(\lenof{\arr{A}}))\#_1(\arr{A})\) where \(\#_1(\arr{A})\) is the number of cells of~\(\arr{A}\) containing a \(1\) and \(\padeps(n) = (\log n)^{-a}\).
	The algorithm requires polynomial work and space on a Common CRCW PRAM.
\end{propC}

Let us point out that \autoref{cited-result:counting} is phrased with our application in mind.
Ajtai actually stated it in terms of first-order formulas over the signature \((+,\times,\le)\) where \(+,\times,\le\) are binary relation symbols with the usual intended meaning.
For every \(n\) let \(\calM_n\) be the structure over \((+,\times,\le)\) with universe \(M_n = \range[0]{n-1}\) and the usual interpretation of the arithmetic relations \(+\),\(\times\), and \(\le\).
Ajtai proved the following.

For every \(a'>0\) there is a formula \(\varphi(x,Y)\) over \((+,\times,\le)\) with a free variable \(x\) and a free unary relation variable \(Y\), such that for every integer \(n>0\), \(m\in M_n\), and \(A\subseteq M_n\) the following holds.
\begin{enumerate}[ref={Statement~(\arabic*)}]
	\item If \(\lenof{A} \le (1 - \padeps'(n))m\) then \(\calM_n \not\models \varphi(m,A)\); and
	\label{ajtai-smaller}
	\item if \(\lenof{A} \ge (1 + \padeps'(n))m\) then \(\calM_n \models \varphi(m,A)\).
	\label{ajtai-larger}
\end{enumerate}
Here \(\padeps'(n)\) is of the same form as \(\padeps(n)\) in \autoref{cited-result:counting}, that is, \(\lambda'(n) = (\log n)^{-a'}\).

We will briefly discuss how this implies \autoref{cited-result:counting} as stated above.
Let $a$ be given and $a'$ be such that \(\lambda'(n)\le \frac{1}{3}\lambda(n)\), for sufficiently large $n$.
Let \(\arr{A}\) be an array with entries from \(\{0,1\}\).
We fix \(n = \lenof{\arr{A}}\) and the interpretation \(A = \{i\in M_n\mid \arr{A}[i+1] = 1\}\) for \(Y\).
Further, we assume that \(\lenof{A}\ge (\log n)^a\) holds, and that \(n\) is sufficiently large. For small \(n\), \(\#_1(\arr{A})\) can be determined exactly, in a  trivial way.
We discuss the case \(\lenof{A}< (\log n)^a\), and how a PRAM can distinguish these two cases, in the end.

It is well-known that first-order formulas can be evaluated by a PRAM in constant time with a polynomial number of processors and space (cf., e.g.\ Immerman \cite[Lemma~5.12]{Immerman}, the space bound is stated in the proof).
Thus, it is possible to test, for all \(0\le m\le n-1\) in parallel, whether \(\calM_n \models \varphi(m,A)\) holds, and to write the results into an array \(\arr{C}\) of size \(n\).%

Observe that \(\calM_n \models \varphi(0,A)\) holds and let \(m_1\) be the largest number from \(M_n\) such that \(\calM_n \models \varphi(m,A)\) holds for all \(m\le m_1\).
Given the array \(\arr{C}\), the number \(m_1\) can easily be determined with quadratic work and linear space.

We claim that \(k = (1+\padeps'(n))(m_1 + 1)\) satisfies \(\lenof{A}\le k \le (1+3\padeps(n))\lenof{A}\).
Note that the factor \(3\) in the upper bound is of no consequence, since \(a\) can always be chosen sufficiently large.
We first show that \(\lenof{A}\le k\) holds.
If \(m_1 + 1 = n\) then \(\lenof{A}\le n \le k\) holds trivially.
Otherwise, \(m_1 + 1 \le n-1\) and we have \(\calM_n\not\models\varphi(m_1+1,A)\) due to the choice of \(m_1\).
Thus, by contraposition of \ref{ajtai-larger} we know that \(\lenof{A} < (1+\padeps'(n))(m_1+1) = k\).

Towards \(k \le (1+\padeps(n))\lenof{A}\), note that the contraposition of \ref{ajtai-smaller} yields \(\lenof{A} > (1-\padeps'(n))m_1\) since \(\calM_n \models \varphi(m_1,A)\) holds.
Because \(1-\padeps'(n) < 1\), this implies \(\lenof{A} + 1 \ge (1-\padeps'(n))(m_1+1)\).
Therefore, we can conclude
\[
	k = (1+\padeps'(n))(m_1 + 1) = \frac{(1+\padeps'(n))(1-\padeps'(n))(m_1 + 1)}{1-\padeps'(n)} \le \frac{1+\padeps'(n)}{1-\padeps'(n)}(\lenof{A} + 1).
\]

We claim that \(\frac{1+\padeps'(n)}{1-\padeps'(n)}\le 1+\padeps(n)\) holds for sufficiently large $n$.
Indeed, we have
\begin{align*}
\frac{1+\padeps'(n)}{1-\padeps'(n)} \le 1+\padeps(n) \Leftrightarrow~&1+\padeps'(n) \le (1-\padeps'(n))(1+\padeps(n)) \\
\Leftrightarrow~&1+\padeps'(n) \le 1 + \padeps(n) - \padeps'(n) - \padeps(n)\padeps'(n)\\
\Leftrightarrow~& 0 \le \padeps(n) - 2\padeps'(n) - \padeps(n)\padeps'(n)\\
\Leftrightarrow~& 0 \le \padeps(n) - 2\padeps'(n) - \padeps(n)\padeps'(n)
\end{align*}
where the last inequality holds, because \(\padeps(n) \le 1\) and \(3\padeps'(n) \le \padeps(n)\) by our choice of \(\padeps(n)\).

Therefore, we have \( k  \le (1+\padeps(n))(\lenof{A}+1)\).
Finally, we obtain \((1+\padeps(n))(\lenof{A}+1) \le (1+3\padeps(n))(\lenof{A})\) because
\[ (1+\padeps(n))(\lenof{A}+1) = 1 + \padeps(n) + \lenof{A} + \padeps(n)\lenof{A} = \left(1+\padeps(n) + \frac{1}{\lenof{A}} + \frac{\padeps(n)}{\lenof{A}}\right)\lenof{A},\]
and we have \(\frac{1}{\lenof{A}} \le \padeps(n)\) as well as \(\frac{\padeps(n)}{\lenof{A}} \le \padeps(n)\), since \(\lenof{A}\ge (\log n)^a\) by assumption and \(\padeps(n) = \frac{1}{(\log n)^a} \le 1\).

We conclude the proof sketch for \autoref{cited-result:counting} by stating the following lemma, which justifies our assumption \(\lenof{A}\ge (\log a)^a\):
since PRAMs can evaluate first order formulas in constant time with a polynomial number of processors and polynomial space, it indeed states that PRAMs can compute \(\lenof{A}\) exactly, or detect that \(\lenof{A}\ge (\log n)^a\) holds.

\begin{lemC}[{\cite[Lemma~2.2]{DBLP:conf/dimacs/Ajtai93}}]
	For every \(a > 0\), there is a formula \(\psi(x,Y)\) over \((+,\times,\le)\) with a free variables \(x\) and a free unary relation variable \(Y\), such that, for all \(A\subseteq M_n\), \(\calM_n\models \psi(m,A)\) if and only if \(m = \min\{\lenof{A}, (\log n)^a\}\).
\end{lemC}
We note that \cite{DBLP:conf/dimacs/Ajtai93} gives a self-contained proof but attributes this result to \cite{ParisW87}. However in the latter paper the setting is slightly different, in that it considers arithmetic over the natural numbers.

We now turn towards the first step of the proof for \autoref{result:prefix-sums}.
It is captured by the following result.
\begin{propC}[{\cite[Theorem~3.1]{GoldbergZ95}}]\label{cited-result:sums}
	For every fixed integer \(a > 0\), there is a \ctpalgo that, given an array \(\arr{A}\)  containing a sequence \(a_1,\ldots,a_{\lenof{\arr{A}}}\) of natural numbers of size at most \(\bigO(\log \lenof{\arr{A}})\), computes an integer \(s\) satisfying \(\sum_{i=1}^{\lenof{\arr{A}}} a_i\le s \le (1+\padeps(n))\sum_{i=1}^{\lenof{\arr{A}}}a_i\) where \(\padeps(n) = (\log n)^{-a}\).
	The algorithm requires polynomial work and space on a Common CRCW PRAM.
\end{propC}

For a number \(a_i\), we denote by \(a_{i,j}\) the bit at position \(j\) of the binary representation of \(a_i\).
The idea for \autoref{cited-result:sums} is to approximate, for each position \(j\), how many  numbers $a_i$ have a 1 at  position \(j\).
That is, for each \(j\), the algorithm computes an approximation \(k'_j\) of \(k_j = \sum_{i=1}^{\lenof{\arr{A}}} a_{i,j}\) using \autoref{cited-result:counting}.
For this purpose, an algorithm can write the binary representations of the numbers \(a_1,\ldots,a_n\) into a two-dimensional array of size \(\bigO((\log n)\cdot n)\) where \(n = \lenof{\arr{A}}\), filling up empty cells with \(0\).
Assuming the numbers are written row-wise into the array, \autoref{cited-result:counting} can be applied to each column in parallel.
Overall this requires polynomial work and space.

The resulting numbers \(k'_0,\ldots,k'_m\) with \(m\in\bigO(\log n)\) satisfy \(k_j\le k'_j \le (1+\padeps(n))k_j\) and have size \(\bigO(\log n)\).
Note that \(\sum_{j=1}^n k_j\cdot 2^j = \sum_{i=1}^{n}a_i\), and, consequently, for \(s = \sum_{j=1}^n k'_j\cdot 2^j\) we have that \(\sum_{i=1}^n a_i\le s \le (1+\padeps(n))\sum_{i=1}^n a_i\) holds.
Since there are only \(\bigO(\log n)\) many numbers \(k'_j\) their exact sum, that is \(s\), can be computed with polynomial work and space on a Common CRCW PRAM.\footnote{%
	This seems to be folklore. It has been proved for polynomial-size, constant-depth circuits in, e.g., \cite[Theorem~1.21]{DBLP:books/daglib/0097931} and these circuits can be simulated by a Common CRCW PRAM in constant time with polynomial work and space, see, e.g., \cite[Theorem~10.9]{DBLP:books/aw/JaJa92}, where the space bound is implicit in the proof.%
}

The next step is to show that prefix sums can be computed in constant time (with polynomial work and space).
Note that this mainly involves establishing consistency.

\begin{propC}[{\cite[Theorem~3.2]{GoldbergZ95}}]\label{cited-result:prefix-sums-poly}
	For every fixed \(a > 0\), there is a \ctpalgo that, given an array \(\arr{A}\) containing a sequence \(a_1,\ldots,a_{\lenof{\arr{A}}}\) of natural numbers of size at most \(\bigO(\log \lenof{\arr{A}})\), computes consistent \(\padeps\)-approximate prefix sums of \(a_1,\ldots,a_{\lenof{\arr{A}}}\) where \(\padeps(n) = (\log n)^{-a}\).
	The algorithm requires polynomial work and space on a Common CRCW PRAM.
\end{propC}

The algorithm for \autoref{cited-result:prefix-sums-poly} uses \emph{approximate summation trees} which were originally introduced by Goodrich, Matias, and Vishkin \cite{DBLP:conf/soda/GoodrichMV94,DBLP:conf/ipps/GoodrichMV93} for computing prefix sums using randomisation.

\begin{defiC}[{\cite[Sections~2 and~3, resp.]{DBLP:conf/soda/GoodrichMV94,DBLP:conf/ipps/GoodrichMV93}}]\label{definition:sumtree}
	Let \(\padeps\) be an accuracy function.
	A \emph{\(\padeps\)-approximate summation tree} \(\tree\) for a sequence \(a_1,\ldots,a_n\) of integers is a complete, ordered binary tree with \(n\) leaves in which every node \(\tnode\) is labelled with an integer \(\tilde{s}(\tnode)\) such that
	\begin{itemize}
	\item the \(i\)-th leaf (in, e.g., pre-order) is labelled with \(a_i\), for all \(1\le i\le n\); and
	\item \(s(\tnode) \le \tilde{s}(\tnode) \le (1+\padeps(n))s(\tnode)\) holds for each node \(v\).
	Here \(s(\tnode)\) is the sum \(\sum_{k=i}^j a_k\) where \(a_i,\ldots,a_j\) are the leaf labels of the subtree rooted at \(\tnode\).
	\end{itemize}
	A \(\padeps\)-approximate summation tree is \emph{consistent} if \(\tilde{s}(\tnode) \ge \tilde{s}(\tnode_\ell) + \tilde{s}(\tnode_r)\) where \(\tnode_\ell\) and \(\tnode_r\) denote the children of \(\tnode\) holds for all inner nodes \(\tnode\).
\end{defiC}

To compute \(\padeps\)-approximate prefix sums the algorithm first constructs a consistent \(\padeps\)-approximate summation tree and then derives the prefix sums from it.

Indeed, given a \(\padeps\)-approximate summation tree \(T\) for a sequence \(a_1,\ldots,a_n\), consistent prefix sums can be derived as follows.
For each \(i\) let \(V_i\) be the set of nodes in \(T\) that occur as left children of a node on the path from the leaf for \(a_i\) to the root but are not on the path themselves.
That is, a node \(\tnode\) is in \(V_i\) if and only if \begin{enumerate}
\item \(\tnode\) is the left child of a node \(\tnodeB\);
\item \(\tnode\) does not occur on the path from the node for \(a_i\) to the root; but
\item \(\tnodeB\) does occur on the path from the node for \(a_i\) to the root.
\end{enumerate}
Observe that the subtrees rooted at the nodes in \(V_i\) induce a partition of \(a_1,\ldots,a_{i-1}\).
Thus, for \(b_i = \sum_{\tnode\in V_i}\tilde{s}(\tnode)+a_i\) we have
\[ \sum_{j=1}^i a_j \le b_i \le (1+\padeps(n))\sum_{j=1}^i a_j.\]
Since \(V_i\) contains at most \(\log n\) nodes, the sum \(b_i\) can be computed with polynomial work and space, cf.\ the proof of \autoref{cited-result:sums}.
For a proof that the sequence \(b_1,\ldots,b_n\) is consistent, we refer to \cite[Lemma~3.2]{DBLP:conf/soda/GoodrichMV94}.%

It remains to revisit how a PRAM can construct a consistent \(\padeps\)-approximate summation tree.
For any integer \(b > 0\), a \(\padeps'\)-approximate summation tree \(T'\) where \(\padeps'=(\log n)^{-b}\) can be constructed with polynomial work and space by applying the algorithm of  \autoref{cited-result:sums} for each inner node of \(T'\).
To obtain a consistent \(\padeps\)-approximate summation tree \(T\) the label of every inner node of \(T'\) is multiplied with \((1+\padeps'(n))^h\) where \(h\) is the height of the node.
More precisely, \(\tilde{s}(\tnode) = (1+\padeps'(n))^h\tilde{s}(\tnode')\) where \(\tnode'\) is the corresponding node of \(\tnode\) in \(T'\) and \(h\) is the height of \(\tnode\) in \(T\).

We follow the proof given for \cite[Lemma~2.2]{DBLP:conf/ipps/GoodrichMV93}%
	\footnote{We note that the proof in \cite[Lemma~2.2]{DBLP:conf/ipps/GoodrichMV93} uses a slightly different factor since their summation trees assert \((1-\padeps(n))s(\tnode) \le \tilde{s}(\tnode)\) instead of \(s(\tnode) \le \tilde{s}(\tnode)\).}
to show that \(T\) is a consistent \(\padeps\)-approximate summation tree.
We start by proving that \(T\) is consistent.

Let \(\tnode\) be an inner node of \(T\), \(h\) its height, \(\tnode_\ell, \tnode_r\) its children, and \(\tnode', \tnode'_\ell, \tnode'_r\) the corresponding nodes in \(T'\).
As in \autoref{definition:sumtree} we denote by \(s(\tnodeB)\) the exact sum of all leaf labels of the subtree rooted at \(\tnodeB\).
We have that
\begingroup%
\allowdisplaybreaks%
\begin{align*}
	\tilde{s}(\tnode_\ell) + \tilde{s}(\tnode_r)
	&= (1+\padeps'(n))^{h-1}\tilde{s}(\tnode'_\ell) + (1+\padeps'(n))^{h-1}\tilde{s}(\tnode'_r)\\
	&\le (1+\padeps'(n))^{h}s(\tnode'_\ell) + (1+\padeps'(n))^{h}s(\tnode'_r)\\
	&=(1+\padeps'(n))^{h}s(\tnode')\\
	&\le (1+\padeps'(n))^{h}\tilde{s}(\tnode')\\
	&=\tilde{s}(\tnode)
\end{align*}
\endgroup%
where the inequalities holds because \(T'\) is an \(\padeps'\)-approximate summation tree.
Thus, \(T\) is consistent.

We now argue that \(T\) is a \(\padeps\)-approximate summation tree.
Indeed, for any inner node \(\tnode\) of \(T\) we have \(s(\tnode) \le \tilde{s}(\tnode) \le (1+\padeps'(n))^{\log(n)+1}s(\tnode)\) since \(T'\) is an approximate summation tree (and all nodes have height at most \(\log n\)).
The claim then follows by the following observation.\footnote{The observation might be straightforward to some readers, but we feel that a self-contained presentation should give its proof, as well. }%
\begin{obs}
	For all integers \(a > 1\) there is an integer \(b>1\) such that for sufficiently large \(n\) the inequality
	\[ (1+\padeps'(n))^{\log n} \le 1+2(\log n)\padeps'(n) \le 1+\padeps(n)\]
	holds where \(\padeps(n) = (\log n)^{-a}\) and \(\padeps'(n) = (\log n)^{-b}\).
\end{obs}
\begin{proof}[Explanation]
	The second inequality can always be established by choosing \(b\) large enough.
	We thus focus on the first inequality.
	Let \(m = \log n\), so that we can write \((1+\frac{1}{m^b})^m\) for \((1+\padeps(n)')^{\log n}\).
	Thanks to the binomial theorem we have
	\[
			(1+\frac{1}{m^b})^m = \sum_{k=0}^m \binom{m}{k} 1^{m-k} \left(\frac{1}{m^b}\right)^k = \sum_{k=0}^m \binom{m}{k} \left(\frac{1}{m^b}\right)^k
	\]
	where the second equality holds simply because \(1^{m-k} = 1\) for all \(k\in\range[0]{m}\).
	Pulling out the first two addends and then shifting \(k\) yields
	\[
		 \sum_{k=0}^m \binom{m}{k} \left(\frac{1}{m^b}\right)^k = 1 + m\cdot \frac{1}{m^b} + \sum_{k=2}^m \binom{m}{k} \left(\frac{1}{m^b}\right)^k = 1 + m\cdot \frac{1}{m^b} + \sum_{k=1}^{m-1} \binom{m}{k+1} \left(\frac{1}{m^b}\right)^{k+1}.
	\]
	To establish the inequality it thus suffices to show that \(r =  \sum_{k=1}^{m-1} \binom{m}{k+1} \left(\frac{1}{m^b}\right)^{k+1}\) is less or equal than \(m\cdot \frac{1}{m^b}\).
	Indeed, then it follows that
	\[
		(1+\frac{1}{m^b})^m = 1 + m\cdot \frac{1}{m^b} + \sum_{k=1}^{m-1} \binom{m}{k+1} \left(\frac{1}{m^b}\right)^{k+1} \le 1 + 2\cdot m\cdot \frac{1}{m^b} = 1 + 2(\log n)\cdot \padeps'(n)
	\]
	holds.

	For establishing \(r \le m\cdot \frac{1}{m^b}\), we observe that it suffices to show that \(\binom{m}{k+1}\left(\frac{1}{m^b}\right)^{k} \le 1\) holds, for every \(k\in\range{m-1}\).
	Expanding the binomial coefficient yields
	\begin{multline*}\allowdisplaybreaks
		\binom{m}{k+1}\left(\frac{1}{m^b}\right)^{k} = \frac{m!}{(k+1)!\cdot (m-k-1)!}\cdot \left(\frac{1}{m^b}\right)^{k}\\
		= \frac{m\cdot (m-1)\cdot \ldots \cdot (m-k)}{(k+1)!}\left(\frac{1}{m^b}\right)^{k}
		=  \frac{m\cdot (m-1)\cdot \ldots \cdot (m-k)}{(k+1)!\cdot m^{bk}}.
	\end{multline*}
	Finally, note that there are \(k+1\) factors in the numerator, each of which is less or equal than \(m\), and \(m\) occurs at least \(bk \ge 2k \ge k + 1\) times in the denominator, since \(k\ge 1\) and \(b > 1\) by assumption.
	Since also \((k+1)!\ge 1\), we can conclude that \(\binom{m}{k+1}\left(\frac{1}{m^b}\right)^{k} \le 1\) holds.
\end{proof}

The last step to prove \autoref{result:prefix-sums} is to show that the work and space bounds can be improved to \(\bigO(n^{1+\varepsilon})\).
\begin{proof}[Proof of \autoref{result:prefix-sums}]
Due to \autoref{cited-result:prefix-sums-poly} there is a \ctpalgo that computes prefix sums with polynomial work and space.
Let \(n=\lenof{\arr{A}}\) and \(c\) be a constant such that this algorithm requires work and space \(\bigO(n^c)\).
Furthermore, let \(\varepsilon > 0\) be arbitrary but fixed, and define \(\delta = \frac{\varepsilon}{c}\).

Let \(\padeps'(n)\) be an accuracy function of the form \(\padeps'(n) = (\log n)^{-b}\) such that \((1+\padeps'(n))^2 \le (1+\padeps(n))\) holds for all \(n > 2\).
Note that this can be done by choosing a large enough \(b\) which only depends on \(a\).

The algorithm operates as follows in three steps.
It first divides the given array \(\arr{A}\) into \(m = n^{1-\delta}\) subarrays \(\arr{A}_1,\ldots,\arr{A}_m\) of size at most \(n^{\delta}\).
Then, for each subarray \(\arr{A}_k\) in parallel, arrays \(\arr{B}_k\) containing consistent \(\padeps'\)-approximate prefix sums for \(\arr{A}_k\) are computed using the algorithm guaranteed by \autoref{cited-result:prefix-sums-poly}.
This step requires work and space \(\bigO(n^{1+\varepsilon})\) because
\[ \left(n^{\delta}\right)^c\cdot n^{1-\delta} = n^{\frac{\varepsilon}{c}c}\cdot n^{1-\frac{\varepsilon}{c}} = n^{1+\varepsilon - \frac{\varepsilon}{c}}\le n^{1+\varepsilon}.\]

In the second step, the algorithm initialises an array \(\arr{C}\) of size \(n^{1-\delta}\) by setting \(\arr{C}[1] = 0\) and, for every index \(k > 1\), \(\arr{C}[k]\) to the largest prefix sum computed for \(\arr{A}_{k-1}\), i.e.\ \(\arr{B}_{k-1}[n^\delta]\).
It then recursively computes an array \(\arr{D}\) containing consistent \(\padeps'\)-approximate prefix sums for \(\arr{C}\).

Finally, in the third step, the prefix sums for the \(\arr{A}_k\) and \(\arr{C}\) are combined to prefix sums for \(\arr{A}\).
For an index \(i\in\range{n}\) let \(k_i\in\range{m}\) and \(i'\in\range{n^\delta}\) be such that \(n^\delta\cdot (k_i-1) + i' = i\).
That is, the \(i\)-th cell of \(\arr{A}\) is the \(i'\)-th cell of subarray \(\arr{A}_{k_i}\).
Let \(\arr{B}\) be the array of length \(n\) with \(\arr{B}[i] = \arr{B}_{k_i}[i] + \arr{D}[k_i]\).
The algorithm returns the array \(\arr{B}\).

The last step can easily be performed with linear work and space.
The recursive call in the second step thus requires \(\bigO(n^{1+\varepsilon})\) as well.
Note that the recursion depth is at most \(\log_{n^\delta} n \in \bigO(\frac{1}{\delta})\).
Therefore, the algorithm runs in constant time and requires work and space \(\bigO(n^{1+\varepsilon})\) as claimed.

In the following we prove that the algorithm is correct.
That is, we prove that \(\arr{B}\) contains consistent \(\padeps\)-approximate prefix sums for \(\arr{A}\).

Let \(i\in\range{n}\) and let \(i'\) and \(k_i\) be as before.
We first show that \(\arr{B}\) contains (not necessarily consistent) approximate prefix sums for \(\arr{A}\).
Indeed, we have
\begingroup%
\allowdisplaybreaks%
\begin{align*}
	\sum_{j=1}^{i} \arr{A}[j] &= \sum_{\ell=1}^{k_i-1} \sum_{j=1}^{n^\delta} \arr{A}_\ell[j] + \sum_{j=1}^{i'} \arr{A}_{k_i}[j]\\
	&\le \sum_{\ell=1}^{k_i-1} \arr{B}_\ell[n^\delta] + \arr{B}_{k_i}[i']\\
	&=\sum_{\ell=1}^{k_i-1} \arr{C}[\ell+1] + \arr{B}_{k_i}[i']\\
	&=\sum_{\ell=1}^{k_i} \arr{C}[\ell] + \arr{B}_{k_i}[i']\\
	&\le \arr{D}[k_i] + \arr{B}_{k_i}[i']\\
	&=\arr{B}[i]
\end{align*}%
\endgroup
where the first inequality holds because the \(\arr{B}_\ell\) contain prefix sums for the \(\arr{A}_\ell\), the second and third equalities hold by definition of \(\arr{C}\) (in particular, \(\arr{C}[1] = 0\)), and the last inequality holds because \(\arr{D}\) contains approximate prefix sums for \(\arr{C}\).

Conversely, we have
\begin{align*}
	\arr{B}[i] &= \arr{D}[k_i] + \arr{B}_{k_i}[i']\\
	&\le (1+\padeps'(n))\sum_{\ell=1}^{k_i} \arr{C}[\ell] + \arr{B}_{k_i}[i']\\
	&=(1+\padeps'(n))\sum_{\ell=1}^{k_i-1} \arr{B}_\ell[n^\delta] + \arr{B}_{k_i}[i']\\
	&\le (1+\padeps'(n))\sum_{\ell=1}^{k_i-1} (1+\padeps'(n))\sum_{j=1}^{n^\delta} \arr{A}_\ell[j]
	+ (1+\padeps'(n))\sum_{j=1}^{i'} \arr{A}_{k_i}[j]\\
	&\le (1+\padeps'(n))^2 \Big[\sum_{\ell=1}^{k_i-1} \sum_{j=1}^{n^\delta} \arr{A}_\ell[j] + \sum_{j=1}^{i'} \arr{A}_{k_i}[j]\Big]\\
	&=(1+\padeps'(n))^2 \sum_{j=1}^{i} \arr{A}[j]\\
	&\le(1+\padeps(n)) \sum_{j=1}^{i} \arr{A}[j]
\end{align*}
where the first inequality holds because \(\arr{D}\) contains \(\padeps'\)-approximate prefix sums for \(\arr{C}\), the second equality holds by definition of \(\arr{C}\), and the second inequality holds because the \(\arr{B}_\ell\) contain \(\padeps'\)-approximate prefix sums for the \(\arr{A}_{\ell}\).
The final inequality holds by choice of \(\padeps'\).

It remains to show consistency for \(i\ge 2\).
That is, we have to assert that \(\arr{B}[i] - \arr{B}[i-1] \ge \arr{A}[i]\) holds.
We make a case distinction.
If \(i\) and \(i-1\) both map to the same subarray, i.e.\ if \(k_i = k_{i-1}\) holds, then
\[ \arr{B}[i] - \arr{B}[i-1] = (\arr{B}_{k_i}[i'] + \arr{D}[k_i]) - (\arr{B}_{k_i}[i'-1] + \arr{D}[k_i]) = \arr{B}_{k_i}[i'] - \arr{B}_{k_i}[i'-1] \ge \arr{A}_{k_i}[i'] = \arr{A}[i]\]
where the inequality holds because \(\arr{B}_{k_i}\) contains consistent approximate prefix sums for \(\arr{A}_{k_i}\).

Otherwise, we can conclude that \(i\) is mapped to the first cell of \(\arr{A}_{k_i}\) and \(i-1\) is mapped to the last cell of \(\arr{A}_{k_i-1}\).
Thus, we have
\begingroup%
\allowdisplaybreaks%
\begin{align*}
	\arr{B}[i] - \arr{B}[i-1] &= (\arr{B}_{k_i}[1] + \arr{D}[k_i]) - (\arr{B}_{k_i-1}[n^\delta] + \arr{D}[k_i-1])\\
	&= \arr{D}[k_i] - \arr{D}[k_i-1] - \arr{B}_{k_i-1}[n^\delta] + \arr{B}_{k_i}[1]\\
	&\ge \arr{C}[k_i] - \arr{B}_{k_i-1}[n^\delta] + \arr{B}_{k_i}[1]\\
	&=\arr{B}_{k_i}[1]\\
	&\ge\arr{A}_{k_i}[1] = \arr{A}[i]
\end{align*}
\endgroup
where the first inequality holds because \(\arr{D}\) contains consistent approximate prefix sums for \(\arr{C}\), the last equality holds because we have \(\arr{C}[k_i] = \arr{B}_{k_i-1}[n^\delta]\) by definition, and the last inequality holds because \(\arr{B}_{k_i}\) contains approximate prefix sums for \(\arr{A}_{k-1}\).
\end{proof}

 \end{document}